\documentclass[hidelinks,onefignum,onetabnum, final]{siamart220329}

\usepackage[utf8]{inputenc}
\usepackage{amsfonts}
\usepackage{amssymb}
\usepackage{graphicx}
\usepackage{color}
\usepackage{todonotes}
\usepackage{algorithm}
\usepackage{algorithmic}
\usepackage{subfig}
\usepackage{xspace}
\usepackage{etoolbox}
\usepackage{xcolor}
\usepackage{multirow}
\usepackage{subfig}
\usepackage{enumitem}
\usepackage[normalem]{ulem}
\usepackage{tcolorbox}

\usepackage{balance}

\usepackage{mathtools}

\usepackage{tikz}
\usetikzlibrary{matrix, decorations, patterns, positioning, shapes, calc, intersections, arrows, fit}

\usepackage{cleveref}
\Crefname{definition}{Definition}{Definitions}
\crefname{definition}{Def.}{Defs.}
\Crefname{theorem}{Theorem}{Theorems}
\crefname{theorem}{Theorem}{Theorems}
\Crefname{lemma}{Lemma}{Lemmas}
\crefname{lemma}{Lemma}{Lemmas}
\Crefname{section}{Section}{Sections}
\crefname{section}{Section}{Sections}
\Crefname{subsection}{Subsection}{Subsections}
\crefname{subsection}{Subsection}{Subsections}
\Crefname{appendix}{Appendix}{Appendices}
\crefname{appendix}{Appendix}{Appendices}

\Crefname{figure}{Figure}{Figures}
\crefname{figure}{Fig.}{Figs.}

\Crefname{algorithm}{Algorithm}{Algorithms}
\crefname{algorithm}{Alg.}{Algs.}

\Crefname{equation}{Equation}{Equations}
\crefname{equation}{equation}{equations}

\newcommand{\lowerbound}{{\sc LB}\xspace}
\newcommand{\maxp}{{\sc P_{\max}}\xspace}
\newcommand{\lowerboundmatrix}{{\sc LB(Matrix)}\xspace}
\newcommand{\lowerboundtensor}{{\sc LB(Tensor)}\xspace}
\newcommand{\odata}{{\sc O}\xspace}
\newcommand{\init}[1]{\hat{#1}}
\newcommand{\tmp}[1]{q_{prev}}
\newcommand{\lbbasedpartition}{{\it \cref{alg:3dmultittm}~(fast)}\xspace}
\newcommand{\bestconfigAAO}{{\it \cref{alg:3dmultittm}~(best)}\xspace}
\newcommand{\bestconfigAAOGen}{{\it \cref{alg:genMultittm}~(best)}\xspace}
\newcommand{\bestconfigSeq}{{\it TTM-in-Seq}\xspace}
\newcommand{\lbSeq}{{\it $C_{LB}$(TTM-in-Seq)}\xspace}
\newcommand{\lbFirstTTM}{{\it $C_{LB}$(1st TTM)}\xspace}

\definecolor{pastelviolet}{rgb}{0.8, 0.6, 0.79}
\definecolor{babyblueeyes}{rgb}{0.63, 0.79, 0.95}
\definecolor{pastelyellow}{rgb}{0.99, 0.99, 0.59}
\definecolor{pastelgreen}{rgb}{0.47, 0.87, 0.47}
\definecolor{lightgreen}{rgb}{0, 0.5, 0}
\definecolor{pastelred}{rgb}{1.0, 0.41, 0.38}
\colorlet{patternblue}{blue!60}
\definecolor{shadecolor}{RGB}{225,225,225}

\usepackage{amssymb,amsfonts,amsmath}

\usepackage[mathscr]{eucal}

\usepackage{stmaryrd}

\usepackage{amsbsy}

\usepackage{bm}

\usepackage{braket}

\newcommand{\Tra}{{\sf T}}

\newcommand{\V}[2][]{{\bm{#1\mathbf{\MakeLowercase{#2}}}}} 
 
\newcommand{\M}[2][]{{\bm{#1\mathbf{\MakeUppercase{#2}}}}} 
\newcommand{\Mn}[3][]{{\bm{#1\mathbf{\MakeUppercase{#2}}}}^{(#3)}} 
\newcommand{\MnTra}[4][]{{\bm{#1\mathbf{\MakeUppercase{#2}}}}^{(#3)\Tra}} 
 
\newcommand{\T}[2][]{\boldsymbol{#1\mathscr{\MakeUppercase{#2}}}} 

\newcommand{\X}{\T{X}}
\newcommand{\Y}{\T{Y}}
\newcommand{\mult}{$(d{+}1)$-ary multiplication}

\newcommand{\starontop}[1]{{#1}^*}

\definecolor{pdfurlcolor}{rgb}{0,0,0.6}
\definecolor{pdfcitecolor}{rgb}{0,0.6,0}
\definecolor{pdflinkcolor}{rgb}{0.6,0,0}

\headers{Communication Lower Bounds for Multi-TTM Computation}{H. AL Daas, G. Ballard, L. Grigori, S. Kumar, and K. Rouse}

\title{Communication Lower Bounds and Optimal Algorithms for Multiple Tensor-Times-Matrix Computation}

\author{Hussam Al Daas\thanks{Rutherford Appleton Laboratory, Didcot, Oxfordshire, UK
		(\email{hussam.al-daas@stfc.ac.uk}).}
	\and Grey Ballard \thanks{Wake Forest University, Winston-Salem, NC, USA
		(\email{ballard@wfu.edu}).}
	\and Laura Grigori \thanks{Inria Paris, France (\email{laura.grigori@inria.fr}, \email{suraj.kumar@inria.fr}).}
	\and Suraj Kumar\footnotemark[3]
	\and Kathryn Rouse \thanks{Inmar Intelligence, Winston-Salem, NC, USA
		(\email{kathryn.rouse@inmar.com}).}}

\usepackage{amsopn}

\newcommand{\response}[1]{{#1}}

\begin{document}

\maketitle

\begin{abstract}
Multiple Tensor-Times-Matrix (Multi-TTM) is a key computation in algorithms for computing and operating with the Tucker tensor decomposition, which is frequently used in multidimensional data analysis. We establish communication lower bounds that determine how much data movement is required to perform the Multi-TTM computation in parallel. The crux of the proof relies on analytically solving a constrained, nonlinear optimization problem. We also present a parallel algorithm to perform this computation that organizes the processors into a logical grid with twice as many modes as the input tensor. We show that with correct choices of grid dimensions, the communication cost of the algorithm attains the lower bounds and is therefore communication optimal. Finally, we show that our algorithm can significantly reduce communication compared to the straightforward approach of expressing the computation as a sequence of tensor-times-matrix operations.
\end{abstract}




\section{Introduction}
\label{sec:intro}

The Tucker tensor decomposition is a low-rank representation or approximation that enables significant compression of multidimensional data.
The Tucker format consists of a core tensor, which is much smaller than the original data tensor, along with a factor matrix for each mode, or dimension, of the data.
Computations involving Tucker-format tensors, such as tensor inner products, often require far fewer operations than with their full-format, dense representations.
As a result, the Tucker decomposition is often used as a dimensionality reduction technique before other types of analysis are done, including computing a CP decomposition \cite{BA98}, for example.

A 3-way Tucker-format tensor can be expressed using the tensor notation $\T{T} = \T{G} \times_1 \M{A}^{(1)} \times_2 \M{A}^{(2)} \times_3 \M{A}^{(3)}$, where $\T{G}$ is the 3-way core tensor, $\M{A}^{(n)}$ is a tall-skinny factor matrix corresponding to mode $n$, and $\times_n$ denotes the tensor-times-matrix (TTM) operation in the $n$th mode \cite{KB09}.
Here, $\T{T}$ is the full-format representation of the tensor that can be constructed explicitly by performing multiple TTM operations.
We call this collective operation the Multi-TTM computation, which is the focus of this work.

Multi-TTM is a fundamental computation in the context of Tucker-format tensors.
When the Tucker decomposition is used as a data compression tool, Multi-TTM is exactly the decompression operation, which is necessary when the full format is required for visualization \cite{KAC20}, for example.
In the case of full decompression, the input tensor is small and the output tensor is large.
One of the quasi-optimal algorithms for computing the Tucker decomposition is the Truncated Higher-Order SVD algorithm \cite{Tucker66,LMV00}, in which each factor matrix is computed as the leading left singular vectors of a matrix unfolding of the tensor.
In this algorithm, the smaller core tensor is computed via Multi-TTM involving the larger data tensor and the computed factor matrices.
When the computational costs of the matrix SVDs are reduced using randomization, Multi-TTM becomes the overwhelming bottleneck computation \cite{MSK20,SGLTU20}.

Since the overall size of multidimensional data grows quickly, there have been many recent efforts to parallelize the computation of the Tucker decomposition and the operations on Tucker-format tensors \cite{ABK16,CC+17,MS22,CLC18,Ballard:TuckerMPI:TOMS20}. 
There has also been recent progress in establishing lower bounds on the communication costs of parallel algorithms for tensor computations, including the Matricized-Tensor Times Khatri-Rao product (MTTKRP)~\cite{Ballard:MTTKRP:IPDPS18,Ballard:PPSC:MTTKRP,ZKBSH22} and symmetric tensor contractions \cite{Solomonik:SISC:2021}.
However, to our knowledge, no communication lower bounds have been previously established for computations relating to Tucker-format tensors.
In this work, we prove communication lower bounds for a class of Multi-TTM algorithms. Additionally, we provide a parallel algorithm that attains the lower bound to within a constant factor and is therefore communication optimal.

To minimize the number of arithmetic operations in a Multi-TTM computation, the TTM operations should be performed in sequence, forming temporary intermediate tensors after each step. \response{A single TTM corresponds to a matrix multiplication along a particular mode of the tensor, therefore a series of matrix multiplications is performed in the sequence approach to compute the final result.}
One of the key observations of this work is that when Multi-TTM is performed in parallel, this approach may communicate more data than necessary, even if communication-optimal algorithms are used for each individual TTM.
By considering the Multi-TTM computation as a whole, we can devise \response{\emph{atomic}} parallel algorithms that can communicate less than this TTM-in-Sequence approach, often with negligible increase in computation.
Our proposed algorithm provides greatest benefit when the input and output tensors vary greatly in size.

\noindent The main contributions of this paper are to
\begin{itemize}
	\item establish communication lower bounds for the parallel \response{atomic} Multi-TTM computation; 
	\item propose a communication optimal parallel algorithm;
	\item show that in many typical scenarios, the straightforward approach based on a sequence of TTM operations communicates more than performing Multi-TTM as a whole.
\end{itemize}

The rest of the paper is organized as follows. \cref{sec:relatedWork} describes previous work on communication lower bounds for matrix multiplication and some tensor operations. In~\cref{sec:notations}, we present our notations and preliminaries for the general Multi-TTM computation. 
To reduce the complexity of notations, we first focus on $3$-dimensional Multi-TTM computation for which we present communication lower bounds and a communication optimal algorithm in~\cref{sec:3dLowerBounds} and~\cref{sec:3dUpperBounds}, respectively.
In~\cref{sec:experiments}, we validate the optimality of the proposed algorithm and show that it significantly reduces communication compared to the TTM-in-Sequence approach with negligible increase in computation in many practical cases. 
We present our general results in~\cref{sec:genLowerBounds,sec:genUpperBounds}, and propose conclusions and perspectives in~\cref{sec:conclusions}.

\section{Related Work}
\label{sec:relatedWork}
A number of studies have focused on communication lower bounds for matrix multiplication, starting with the work by Hong and Kung~\cite{Hong:STOC81} to determine the minimum number of I/O operations for sequential matrix multiplication using the red-blue pebble game. Irony et al.~\cite{IRONY:JPDC04} extended this work for the parallel case. Demmel et al.~\cite{Demmel:IPDPS:CARMA} studied memory independent communication lower bounds for rectangular matrix multiplication based on aspect ratios of matrices. Recently, Smith et al.~\cite{smith2019tight:mm} and Al Daas et al.~\cite{ABGKR22} have tightened communication lower bounds for matrix multiplication. Ballard et al.~\cite{Ballard:3NL} extended communication lower bounds of the matrix multiplication for any computations that can be written as $3$ nested loops. Christ et al.~\cite{Christ:EECS-2013-61} generalized the method to prove communication lower bounds of $3$ nested loop computations for arbitrary loop nesting. We apply their approach to our Multi-TTM definition.

There is limited work on communication lower bounds for tensor operations. Solomonik et al.~\cite{Solomonik:SISC:2021} proposed communication lower bounds for symmetric tensor contraction algorithms. 
Ballard et al.~\cite{Ballard:MTTKRP:IPDPS18} proposed communication lower bounds for MTTKRP computation with cubical tensors.
This work is extended in~\cite{Ballard:PPSC:MTTKRP} to handle varying tensor dimensions. 
A sequential lower bound for tile-based MTTKRP algorithms is proved by Ziogas et al. \cite{ZKBSH22}.
We use some results from~\cite{Ballard:MTTKRP:IPDPS18,Ballard:PPSC:MTTKRP} to prove communication lower bounds for Multi-TTM.


\section{Notations and Preliminaries}
\label{sec:notations} 
In this section, we present our notations and basic lemmas for $d$-dimensional Multi-TTM computation. In \cref{sec:3dLowerBounds,sec:3dUpperBounds,sec:experiments}, we focus on $d=3$, i.e., $\T{Y}$ = $\T{X} \times_1 {\Mn{A}{1}}^\Tra \times_2 {\Mn{A}{2}}^\Tra \times_3 {\Mn{A}{3}}^\Tra$. We present our general results in \cref{sec:genLowerBounds,sec:genUpperBounds}.

We use boldface uppercase Euler script letters to denote tensors ($\T{X}$) and boldface uppercase letters with superscripts to denote matrices ($\Mn{A}{1}$). 
We use lowercase letters with subscripts to denote sizes ($n_1$) and add the prime symbol to them to denote the indices ($n_1^\prime$). 
We use one-based indexing throughout and $[d]$ to denote the set $\{1,2,\cdots, d\}$. 
To improve the presentation, we denote the product of elements having the same lowercase letter with all subscripts by the lowercase letter only ($n_1\cdots n_d$ by $n$ and $r_1\cdots r_d$ by $r$).
We denote the product of the $i$ rightmost terms with the capital letter with subscript $i$, $N_i=\prod_{j=d-i+1}^dn_j$ and $R_i = \prod_{j=d-i+1}^d r_i$, thus $n=N_d$, and $n_d=N_1$.



Let $\Y\in \mathbb{R}^{r_1\times\cdots\times r_d}$ be the $d$-mode output tensor, $\X\in \mathbb{R}^{n_1\times\cdots\times n_d}$ be the $d$-mode input tensor, and $\Mn{A}{k} \in \mathbb{R}^{n_k\times r_k}$ be the matrix of the $k$th mode. 
Then the Multi-TTM computation can be represented as $\Y = \X \times_1 {\Mn{A}{1}}^\Tra \cdots \times_d {\Mn{A}{d}}^\Tra$.
Without loss of generality and to simplify notation, we consider that the input tensor $\X$ is larger than the output tensor $\Y$, or $n\geq r$.
This corresponds to computing the core tensor of a Tucker decomposition given computed factor matrices, for example.
However, the opposite relationship where the output tensor is larger (e.g., $\X = \Y \times_1 {\Mn{A}{1}} \cdots \times_d {\Mn{A}{d}}$) is also an important use case, corresponding to forming an explicit representation of a (sub-)tensor of a Tucker-format tensor.
Our results extend straightforwardly to this case.
\response{We consider $d$-dimensional input and output tensors and therefore assume $n_i\ge 2$ and $r_i\ge 2$ for $1\le i\le d$.}
We also assume without loss of generality that the tensor modes are ordered in such a way that $n_1r_1 \le n_2r_2 \le \cdots\le n_dr_d$.


\begin{definition}
	\label{def:amttm}
	Let $\X$ be an $n_1\times \cdots \times n_d$ tensor, $\Y$ be an $r_1\times \cdots \times r_d$ tensor, and $\Mn{A}{j}$ be an $n_j \times r_j$ matrix for $j\in[d]$.
	Multi-TTM computes
	$$\Y = \X \times_1 {\Mn{A}{1}}^\Tra \cdots \times_d {\Mn{A}{d}}^\Tra$$
	where for each $(r_1^\prime,\ldots,r_d^\prime) \in [r_{1}] \times \cdots \times [r_d]$,
	\begin{equation}\label{eq:ourMultiTTMDef}
		\Y(r_1^\prime,\ldots,r_d^\prime) = \sum_{\{n_k^\prime \in [n_k]\}_{k \in [d]}} \X(n_1^\prime,\ldots,n_d^\prime) \prod_{j \in [d]} \Mn{A}{j}(n_j^\prime,r_j^\prime).
	\end{equation}
\end{definition}

Let us consider an example when $d=2$. In this scenario, the input and output tensors are in fact matrices $\M{X},\M{Y}$, and $\M{Y} = \MnTra{A}{1}{}\M{X}\Mn{A}{2}$. As mentioned earlier, Multi-TTM computation can be performed as a sequence of TTM operations, in this case two matrix multiplications. However, we define the Multi-TTM to perform all the products at once for each term of the summation of \cref{eq:ourMultiTTMDef}. 
Our definition comes at greater arithmetic cost, as partial $(d+1)$-ary multiplies are not computed and reused, but we will see that this approach can reduce communication cost.
We describe how the extra computation can often be reduced to a negligible cost in \cref{sec:cost-3d} \response{and compare it to the computation cost of TTM-in-Sequence in \cref{sec:computation-cost-comparison}}.

We can write pseudocode for the Multi-TTM with the following:
\begin{align*}
&\text{for $n_1^\prime = 1{:}n_1$, \ldots, for $n_d^\prime = 1{:}n_d$,}\\
&\quad \text{for $r_1^\prime = 1{:}r_1$, \ldots, for $r_d^\prime = 1{:}r_d$,}\\
&\quad \quad \Y(r_1^\prime,\ldots,r_d^\prime)\  +=  \X(n_1^\prime,\ldots,n_d^\prime) \,\cdot  \Mn{A}{1}(n_{1}^\prime,r_1^\prime) \cdot \cdots \cdot \Mn{A}{N}(n_d^\prime,r_d^\prime)\\
\end{align*}

\begin{definition}
	A \emph{parallel atomic Multi-TTM algorithm} computes each term of the summation of \cref{eq:ourMultiTTMDef} atomically on a unique processor, but it can distribute the $nr$ terms over processors in any way.
\end{definition}

Here atomic computation of a single \mult{} for a parallel algorithm means that all the multiplications of this operation are performed on only one processor, i.e., all $d+1$ inputs are accessed on that processor in order to compute the single output value. 
This assumption is necessary for our communication lower bounds. 
Processors can reorganize their local atomic operations to reduce computational costs without changing the communication or violating parallel atomicity.
However it is reasonable for an algorithm to break this assumption in order to improve arithmetic costs by reusing partial results across processors, and we compare against such algorithms in \cref{sec:experiments}.

\subsection{Parallel Computation Model}
\label{sec:prelim:costModel}

We consider that the computation is distributed across $P$ processors.
Each processor has its own local memory and is connected to all other processors via a fully connected network.
Every processor can operate on data in its local memory and must communicate to access data of other processors.
Hence, communication refers to send and receive operations that transfer data from local memory to the network and vice-versa.
Communication cost mainly depends on two factors -- the amount of data communicated (bandwidth cost) and the number of messages (latency cost).
Latency cost is dominated by bandwidth cost for computations involving large messages, so we focus on bandwidth cost in this work and refer it as communication cost throughout the text.
We assume the links of the network are bidirectional and that the communication cost is independent of the number of pairs of processors that are communicating.
Each processor can send and receive at most one message at the same time. 
In our model, the communication cost of an algorithm refers to the cost along the critical path.



\subsection{Existing Results}
\label{sec:pastResults}

Our work relies on two fundamental results.
The first, a geometric result on lattices, allows us to relate the volume of computation to the amount of data accessed by determining the maximum data reuse.
The result is a specialization of the H\"{o}lder-Brascamp-Lieb inequalities \cite{BCCT10}.
This result has previously been used to derive lower bounds for tensor computations \cite{Ballard:MTTKRP:IPDPS18, Ballard:PPSC:MTTKRP,Christ:EECS-2013-61, Knight15} in a similar way to the use of the Loomis-Whitney inequality~\cite{LoomisWhitney:AMS} in derivations of communication lower bounds for several linear algebra computations~\cite{Ballard:3NL}.
The result is proved in \cite{Christ:EECS-2013-61}, but we use the statement from \cite[Lemma 4.1]{Ballard:MTTKRP:IPDPS18}. Here $\V{1}$ represents a vector of all ones \response{and $\geq$ relation between vectors applies elementwise}.

\begin{lemma}
	\label{lem:hbl}
	Consider any positive integers $\ell$ and $m$ and any $m$ projections $\phi_j:\mathbb{Z}^\ell\rightarrow\mathbb{Z}^{\ell_j}$ ($\ell_j\leq \ell$), each of which extracts $\ell_j$ coordinates $S_j\subseteq [\ell]$ and forgets the $\ell-\ell_j$ others.
	Define
	$\mathcal{C} = \big\{\V{s}=\response{[s_1\ \cdots \ s_m]^\Tra: 0\leq s_i \leq 1 \text{ for } i=1,2,\cdots,m \text{ and }} \M{\Delta}\cdot\V{s}\ge\V{1}\big\}\text,$
	where the $\ell\times m$ matrix $\M{\Delta}$ has entries
	$\M{\Delta}_{i,j} = 1 \text{ if } i\in S_j \text{ and } \M{\Delta}_{i,j} = 0 \text{ otherwise}\text.$
	If $[s_1\ \cdots \ s_m]^\Tra\in\mathcal{C}$, then for all $F\subseteq \mathbb{Z}^\ell$,
	\vspace*{-0.15cm}$$ |F| \leq \prod_{j\in [m]}|\phi_j(F)|^{s_j}\text.$$
\end{lemma}

The second result, a general constrained optimization problem, allows us to cast the communication cost of an algorithm as the objective function in an optimization problem where the constraints are imposed by properties of the computation within the algorithm. 
A version of the result is proved in \cite[Lemma 5.1]{Ballard:PPSC:MTTKRP} and used to derive the general communication lower bound for MTTKRP.

\begin{theorem}
	\label{lem:mttkrpOpt}
	Consider the constrained optimization problem: 
	$$\min \sum_{j\in[d]}x_j$$
	such that 
	$$\frac{nr}{P} \le\prod_{j\in[d]}x_j \quad \text{and} \quad 0\le x_j \le k_j \quad \text{for all} \quad 1\le j\le d$$ 
	for some positive constants $k_1 \le k_2\le \cdots\le k_d$ with $\prod_{j\in[d]}k_j = nr$.
	Then the minimum value of the objective function is
	$$I\left(K_I/P\right)^{1/I} + \sum_{j\in[d-I]}k_j$$
	where we use the notation $K_I=\prod_{j=d-I+1}^d k_j$ and \response{$1\le I \le d$ is defined such that $L_I \leq P < L_{I+1}$.
	$$\text{Here} \quad L_j=\frac{K_j}{(k_{d-j+1})^j} \quad \text{for} \quad 1 \leq j \leq d \quad \text{and} \quad L_{d+1} = \infty.\qquad\qquad\qquad$$} 
	The minimum is achieved at the point $\V{x}^*$ defined by $\starontop{x_j}=k_j$ for $1\le j \le d-I$, $\starontop{x_\ell}=\left(K_I/P\right)^{1/I}$ for $d-I<\ell\le d$.
\end{theorem}

While \cref{lem:mttkrpOpt} can be straightforwardly derived from the previous work, we provide an alternate proof in \cref{app:sec:optLemmaProof}.
We represent it in this form to be directly applicable to all the constrained optimization problems in this paper. 
\response{The constraints $nr/P \leq \prod_{j\in[d]}x_j$ and $\prod_{j\in[d]}k_j = nr$
are derived from the Multi-TTM computation. The equality constraint on $\prod_{j\in[d]}k_j$
implies that there is always a feasible solution to the optimization problem for $P\geq 1$.}
We calculate the ranges of $P$ for each $I$ in \Cref{lemma:matrixOptimalSolutions,lemma:tensorOptimalSolutions,lemma:genMatrixOptimalSolutions}.

\section{Lower Bounds for $3$-dimensional Multi-TTM}
\label{sec:3dLowerBounds}

We obtain the lower bound results for 3D tensors in this section, presented as \cref{theorem:lb:3DMultiTTM}.
The lower bound is independent of the size of the local memory of each processor, similar to previous results for matrix multiplication \cite{ABGKR22,Demmel:IPDPS:CARMA} and MTTKRP \cite{Ballard:MTTKRP:IPDPS18,Ballard:PPSC:MTTKRP}, and it varies with respect to the number of processors $P$ relative to the matrix and tensor dimensions of the problem.

\response{The proof focuses on a processor that performs $1/P$th of the computation and owns at most $1/P$th of the data.
It reduces the problem of finding a lower bound on the amount of data the processor must communicate to solving a constrained optimization problem:} we seek to minimize the number of elements of the matrices and tensors that the processor must \response{access or partially compute} in order to execute its computation subject to structure constraints of Multi-TTM.
The most important constraint derives from \cref{lem:hbl}, which relates a subset of the computation within a Multi-TTM algorithm to the data it requires.
The other constraints provide upper bounds on the data required from each array.
The upper bounds are necessary to establish the tightest lower bounds in the cases where $P$ is small.
We show that the optimization problem can be separated into two independent problems, one for the matrix data and one for the tensor data.
\Cref{lemma:matrixOptimalSolutions,lemma:tensorOptimalSolutions} state the two constrained optimization problems along with their analytic solutions, both of which follow from \cref{lem:mttkrpOpt}.
%
That is, setting $d=3$, $k_1=n_1r_1$, $k_2=n_2r_2$ and $k_3=n_3r_3$ in~\cref{lem:mttkrpOpt}, we obtain \cref{lemma:matrixOptimalSolutions}.
Similarly, setting $d=2$ with $k_1=r$ and $k_2=n$, we obtain \cref{lemma:tensorOptimalSolutions}. We recall here that $r=r_1r_2r_3$ and $n=n_1n_2n_3$.


\begin{corollary}
	\label{lemma:matrixOptimalSolutions}
	Consider the following optimization problem:
	$$\min_{x,y,z} x+y+z$$
	such that 
	\begin{align*}
	\frac{nr}{P} &\le xyz  \\ 
	0 &\le \phantom{y}x\phantom{z} \le n_1r_1 \\
	0 &\le \phantom{x}y\phantom{z} \le n_2r_2 \\
	0 &\le \phantom{x}z\phantom{y} \le n_3r_3,
	\end{align*}
	where $n_1r_1 \le n_2r_2 \le n_3r_3$, and $n_1,n_2,n_3,r_1,r_2,r_3,P \ge 1$.
	The optimal solution $(\starontop{x},\starontop{y},\starontop{z})$ depends on the relative values of the constraints, yielding three cases:
	\begin{enumerate}
		\item if $P < \frac{n_3r_3}{n_2r_2}$, then $\starontop{x}=n_1r_1$, $\starontop{y}=n_2r_2$, $\starontop{z}=\frac{n_3r_3}{P}$;
		\item if $\frac{n_3r_3}{n_2r_2}\le P < \frac{n_2n_3r_2r_3}{n_1^2r_1^2}$, then $\starontop{x}=n_1r_1$, $\starontop{y}=\starontop{z}= \big(\frac{n_2n_3r_2r_3}{P}\big)^{\frac{1}{2}}$;
		\item if $\frac{n_2n_3r_2r_3}{n_1^2r_1^2} \le P$, then $\starontop{x}=\starontop{y}=\starontop{z}= \big(\frac{nr}{P}\big)^{\frac{1}{3}}$;
	\end{enumerate}
	which can be visualized as follows.
	\begin{center}
		\begin{tikzpicture}[scale=0.815, every node/.style={transform shape}]
		\draw [->, thick] (-0.1,0) -- (15,0) node [below right, scale=1.2] {$P$};
		\draw (0, 0.1) -- node [below, pastelred, scale=1.2]{$1$}(0,-0.1);
		\draw (5, 0.1) -- node [below, pastelred, scale=1.2]{$\frac{n_3r_3}{n_2r_2}$}(5,-0.1);
		\draw (10, 0.1) -- node [below, pastelred, scale=1.2] {$\frac{n_2n_3r_2r_3}{n_1^2r_1^2}$}(10,-0.1);
		
		\node[align=left,below,scale=1.2] at (2.5, -0.4) {$\starontop{x}=n_1r_1$\\ $\starontop{y}=n_2r_2$\\ $\starontop{z}=\frac{n_3r_3}{P}$};
		\node[align=left,below,scale=1.2] at (7.5, -0.6) {$\starontop{x}=n_1r_1$\\$\starontop{y}=\starontop{z}= \big(\frac{n_2n_3r_2r_3}{P}\big)^{1/2}$};
		\node[align=center,below,scale=1.2] at (12.5, -0.6) {$\starontop{x}=\starontop{y}=\starontop{z}=$\\ $\qquad\quad \big(\frac{nr}{P}\big)^{1/3}$};	
		\end{tikzpicture}
	\end{center} 
\end{corollary}

\begin{corollary}
	\label{lemma:tensorOptimalSolutions}
	Consider the following optimization problem:
	$$\min_{u,v} u+v$$
	such that 
	\begin{align*}
	\frac{nr}{P} &\le uv \\
	0 &\le \;u\; \le r \\
	0 &\le \;v\; \le n,
	\end{align*} 
	where $n\geq r$, and $n,r,P \ge 1$.
	The optimal solution $(\starontop{u},\starontop{v})$ depends on the relative values of the constraints, yielding two cases:
	\begin{enumerate}
		\item if $P < \frac{n}{r}$, then $\starontop{u}= r$, $\starontop{v} = \frac{n}{P}$;
		\item if $ \frac{n}{r} \le P$, then $\starontop{u}=\starontop{v}= \big(\frac{nr}{P}\big)^{\frac{1}{2}}$;
	\end{enumerate}
	which can be visualized as follows.
	\begin{center}
		\begin{tikzpicture}[scale=0.815, every node/.style={transform shape}]
		\draw [->, thick] (-0.1,0) -- (15,0) node [below right, scale=1.2] {$P$};
		\draw (0, 0.1) -- node [below, pastelred, scale=1.2]{$1$}(0,-0.1);
		\draw (7.5, 0.1) -- node [below, pastelred, scale=1.2]{$\frac{n}{r}$}(7.5,-0.1);
		
		\node[align=left,below,scale=1.2] at (3.75, -0.325) {$\starontop{u}= r$\\ $\starontop{v} = \frac{n}{P}$};
		\node[align=left,below,scale=1.2] at (11.25, -0.5) {$\starontop{u}=\starontop{v}= \big(\frac{nr}{P}\big)^{1/2}$};
		\end{tikzpicture}
	\end{center} 
\end{corollary}

\subsection{Communication Lower Bounds for Multi-TTM}
\label{sec:lb:MutliTTM}
We now state the lower bounds for 3-dimensional Multi-TTM. After this, we also present a corollary for cubical tensors.
\begin{theorem}
	\label{theorem:lb:3DMultiTTM}
	Any computationally load balanced atomic Multi-TTM algorithm that starts and ends with one copy of the data distributed across processors involving 3D tensors with dimensions $n_1, n_2, n_3$ and $r_1, r_2, r_3$ performs at least $A+B-\left(\frac{n}{P}+\frac{r}{P} +\sum_{j=1}^3 \frac{n_jr_j}{P}\right)$ sends or receives where
	\begin{align*}
	A &= \begin{cases} n_1r_1 + n_2r_2 + \frac{n_3r_3}{P} & \text{ if } P<\frac{n_3r_3}{n_2r_2} \\ n_1r_1 + 2\left(\frac{n_2n_3r_2r_3}{P}\right)^{\frac{1}{2}} &\text{ if } \frac{n_3r_3}{n_2r_2}\leq P < \frac{n_2n_3r_2r_3}{n_1^2r_1^2} \\ 3\left(\frac{nr}{P}\right)^{\frac{1}{3}} &\text{ if } \frac{n_2n_3r_2r_3}{n_1^2r_1^2} \leq P\end{cases} \\
	B &= \begin{cases} r + \frac{n}{P} & \text{ if } P < \frac{n}{r} \\ 2\left(\frac{nr}{P}\right)^{\frac{1}{2}} &\text{ if } \frac{n}{r} \leq P\text.\end{cases}
	\end{align*}
\end{theorem}
\begin{proof}
	Let $F$ be the set of loop indices associated with the $4$-ary multiplication performed by a processor.
	As we assumed the algorithm is computationally load balanced, $|F| = nr/P$.
	We define $\phi_{\X}(F)$, $\phi_{\Y}(F)$ and $\phi_j(F)$ to be the projections of $F$ onto the indices of the arrays $\X, \Y$, and $\Mn{A}{j}$ for $1\leq j\leq 3$ which correspond to the elements of the array that must be accessed \response{or partially computed} by the processor.
	
	We use \Cref{lem:hbl} to obtain a lower bound on the number of array elements that must be accessed \response{or partially computed} by the processor. The computation involves $5$ arrays (2 tensors and 3 matrices) with $6$ loop indices (see the atomic Multi-TTM definition in \cref{sec:notations}), hence the $6\times 5$ matrix corresponding to the projections above is given by
	$$\M{\Delta} = \begin{bmatrix}\M{I}_{3\times 3} & \V{1}_3 & \V{0}_3\\ \M{I}_{3\times 3} & \V{0}_3 & \V{1}_3 \end{bmatrix}\text.$$
	Here $\V{1}_3$ and $\V{0}_3$ represent the 3-dimensional vectors of all ones and zeros, respectively, and $\M{I}_{3\times 3}$ represents the $3\times 3$ identity matrix. We recall from \Cref{lem:hbl} that $\M{\Delta}_{i,j}=1$ if loop index $i$ is used to access array $j$ and $\M{\Delta}_{i,j}=0$ otherwise. The first three columns of $\M{\Delta}$ correspond to matrices and the remaining two columns correspond to tensors. 
	In this case, we have
	$$\mathcal{C} =  \big\{\V{s}=\response{[s_1\ \cdots \ s_5]^\Tra: 0\leq s_i \leq 1 \text{ for } i=1,2,\cdots,5 \text{ and }} \M{\Delta}\cdot\V{s}\ge\V{1}\big\}\text.$$

	
	Recall that $\V{1}$ represents a vector of all ones. Here $\M{\Delta}$ is not full rank, therefore, we consider all vectors $\V{v} \in \mathcal{C}$ such that $\M{\Delta}\cdot\V{v}=\V{1}$. Such a vector $\V{v}$ is of the form $[a\; a\; a\; 1\text{-}a\; 1\text{-}a]$
	where $0\le a\le 1$. Therefore, we obtain 
	$$\frac{nr}{P} \leq \Big(\prod_{j\in[3]}|\phi_j(F)|\Big)^a \big(|\phi_{\X}(F)||\phi_{\Y}(F)|\big)^{1\text{-}a}\text{ for all } 0\le a\le 1.$$
	
	\noindent The above constraint is equivalent to $\frac{nr}{P} \leq \prod_{j\in[3]}|\phi_j(F)|$ and $\frac{nr}{P} \leq |\phi_{\X}(F)||\phi_{\Y}(F)|$. To see this equivalence note that the forward direction is implied by setting $a=0$ and $a=1$. For the opposite direction, taking the first of the two constraints to the power $a$ and the second to the power $1-a$ then multiplying the two terms yields the original.
	
	Clearly a projection onto an array cannot be larger than the array itself, thus $|\phi_{\X}(F)| \leq n$, $|\phi_{\Y}(F)|\leq r$, and $|\phi_j(F)|\leq n_jr_j$ for $1\leq j \leq 3$.
	
	\response{As the constraints related to projections of matrices and tensors are disjoint, we solve them separately and then sum the results to get a lower bound on the set of elements that must be accessed or partially computed by the processor. We obtain a lower bound on $A$, the number of relevant elements of the matrices by using \Cref{lemma:matrixOptimalSolutions}, and a lower bound on $B$, the number of relevant elements of the tensors by using \Cref{lemma:tensorOptimalSolutions}.} By summing both, we get the positive terms of the lower bound.
	
	
	To bound the sends or receives, we consider how much data the processor could have had at the beginning or at the end of the computation.
	Assuming there is exactly one copy of the data at the beginning and at the end of the computation, there must exist a processor which \response{owns} at most $1/P$ of the elements of the arrays at the beginning or at the end of the computation.
	By employing the previous analysis, this processor must \response{access or partially compute} $A+B$ elements of the arrays, \response{but can only own $\frac{n}{P}+\frac{r}{P} +\sum_{j\in[3]} \frac{n_jr_j}{P}$ elements of the arrays}.
	Thus it must perform the specified amount of sends or receives.
\end{proof}

We denote the lower bound of \cref{theorem:lb:3DMultiTTM} by $\lowerbound$ and use it extensively in \cref{sec:parallelAlgoritm:selectionpiqi} while analyzing the communication cost of our parallel algorithm.


We also state the result for $3$-dimensional Multi-TTM computation with cubical tensors, which is a direct application of \cref{theorem:lb:3DMultiTTM} with $n_1=n_2=n_3=n^{\frac{1}{3}}$ and $r_1=r_2=r_3=r^{\frac{1}{3}}$.

\begin{corollary}
	\label{corollary:lb:3DMultiTTM:cubicaltensors}
	Any computationally load balanced atomic Multi-TTM algorithm that starts and ends with one copy of the data distributed across processors involving 3D cubical tensors with dimensions $n^{\frac{1}{3}}\times n^{\frac{1}{3}}\times n^{\frac{1}{3}}$ and $r^{\frac{1}{3}}\times r^{\frac{1}{3}}\times r^{\frac{1}{3}}$ (with $n\geq r$) performs at least
	$$3\left(\frac{nr}{P}\right)^{\frac{1}{3}} + r - \frac{3(nr)^{\frac{1}{3}}+r}{P}$$
	sends or receives when $P<\frac{n}{r}$ and at least
	$$3\left(\frac{nr}{P}\right)^{\frac{1}{3}}+2\left(\frac{nr}{P}\right)^{\frac{1}{2}} - \frac{n+3(nr)^{\frac{1}{3}}+r}{P}$$
	send or receives when $P \geq \frac{n}{r}$.
\end{corollary}

In particular, we note that the lower bound for cubical atomic Multi-TTM algorithms is smaller than that of a TTM-in-Sequence approach for many typical scenarios in the case $P<n/r$, as we discuss further in \cref{sec:experiments}.

\section{Parallel Algorithm for $3$-dimensional Multi-TTM}
\label{sec:3dUpperBounds}
We organize $P$ processors into a $6$-dimensional $p_1 \times p_2 \times p_3 \times q_1 \times q_2 \times q_3$ logical processor grid. We arrange the
grid dimensions such that $p_1$, $p_2$, $p_3$, $q_1$, $q_2$, $q_3$ evenly distribute $n_1$, $n_2$, $n_3$, $r_1$, $r_2$, $r_3$, respectively. A processor coordinate is represented as $(p_1^\prime, p_2^\prime, p_3^\prime, q_1^\prime, q_2^\prime, q_3^\prime)$, where $1\le p_{k}^\prime \le p_k$, $1\le q_{k}^\prime \le q_k$ for $k=1,2,3$. To be consistent with our notation, we denote $p_1p_2p_3$ and $q_1q_2q_3$ by $p$ and $q$.

$\T{X}_{p_1^\prime p_2^\prime p_3^\prime}$ denotes the subtensor of $\T{X}$ owned by processors $(p_1^\prime, p_2^\prime,$ $p_3^\prime, *, *, *)$. Similarly, $\T{Y}_{q_1^\prime q_2^\prime q_3^\prime}$ denotes the subtensor of $\T{Y}$ owned by processors $(*, *, *, q_1^\prime, q_2^\prime, q_3^\prime)$. $\Mn{A}{1}_{p_1^\prime q_1^\prime}$, $\Mn{A}{2}_{p_2^\prime q_2^\prime}$ and $\Mn{A}{3}_{p_3^\prime q_3^\prime}$ denote submatrices of $\Mn{A}{1}$, $\Mn{A}{2}$ and $\Mn{A}{3}$ owned by processors $(p_1^\prime, *, *, q_1^\prime, *, *)$, $(*, p_2^\prime, *, *, q_2^\prime, *)$ and $(*, *, p_3^\prime, *, *, q_3^\prime)$, respectively.

We impose that there is one copy of data in the system at the start and end of the computation, and every array is distributed evenly among the sets of processors whose coordinates are different for the corresponding dimensions of the variable.
For example, $\T{X}_{111}$ = $\T{X}(1:\frac{n_1}{p_1}, 1:\frac{n_2}{p_2}, 1:\frac{n_3}{p_3})$ is owned by processors $(1,1,1,*,*,*)$.
Similarly, $\Mn{A}{1}_{12}$ = $\Mn{A}{1}(1:\frac{n_1}{p_1}, \frac{r_1}{q_1}+1:2\frac{r_1}{q_1})$ is owned by processors $(1,*,*,2,*,*)$.
We assume that data inside these sets of processors is also evenly distributed. 
For example, in the beginning, processor ($1,1,1,2,1,3$) owns $\frac{1}{P}$th portion of each input variable:
$\frac{p}{P}$th portion of $\T{X}_{111}$, $\frac{p_1q_1}{P}$th portion of $\Mn{A}{1}_{12}$, $\frac{p_2q_2}{P}$th portion of $\Mn{A}{2}_{11}$, and $\frac{p_3q_3}{P}$th portion of $\Mn{A}{3}_{13}$. 
\Cref{fig:dataDistribution} illustrates examples of our data distribution model for two of the arrays.

\begin{figure}[!tb]
	\centering
	\begin{tikzpicture}[scale=0.4, every node/.style={transform shape}]
	\def\xref{0.6}
	\def\yref{0.5}
	
	\draw [dotted,fill=gray!20] (-1,2) -- (0,2) -- (0,3) --(-1,3) --(-1,2);
	\draw [dotted,fill=gray!20] (-1, 3) -- (-1+\xref,3+\yref) -- (0+\xref,3+\yref) -- (0,3) -- (-1, 3) ;
	\draw [dotted,fill=gray!20] (0, 2) -- (0+\xref,2+\yref) -- (0+\xref,3+\yref) -- (0,3) -- (0,2);

	\foreach \y in {0, 1, 2, 3, 4}
	\draw (-2, \y) -- (2, \y);
	
	\foreach \x in {-2, -1, 0, 1, 2}
	\draw (\x, 0) -- (\x, 4);
	
	\foreach \y in {0, 1, 2, 3, 4}
	\draw (2, \y)--(2+4*\xref, 4*\yref+\y);

	\foreach \y in {0, 1, 2, 3, 4}
	\draw (2+\y * \xref, \y * \yref) -- (2+\y * \xref, 4+\y * \yref);
	
	\foreach \x in {-2, -1, 0, 1, 2}
	\draw (\x, 4) -- (\x + 4*\xref, 4+4*\yref);
	
	\foreach \y in {0, 1, 2, 3, 4}
	\draw (-2+\y * \xref, 4+\y * \yref) -- (2+\y * \xref, 4+\y * \yref);
	
	\node [below, lightgreen, scale=2.5] at (0,0) {$\T{X}$};
	

	\node [ lightgreen, scale=1.2] at (-0.5,2.5) {$\T{X}_{231}$};
	
	\draw [->, red] (-1,-1.5) -- (1,-1.5) node [below right, scale=2] {$n_1$};
	\draw [->, red] (-2.5,1) -- (-2.5,3) node [ above left, scale=2] {$n_2$};
	\draw [->, red] (-2.5+\xref, 4+\yref) -- (-2.5 + 3*\xref, 4+3*\yref) node [above,rotate=45, scale=2] {$n_3$};
	\end{tikzpicture}$\qquad$
	\begin{tikzpicture}[scale=0.4, every node/.style={transform shape}]
	\def\xref{0.6}
	\def\yref{0.5}
	\draw [fill=gray!20] (-2,2) -- (-1,2) -- (-1,3) --(-2,3) --(-2,2);
	
	\foreach \y in {0, 1, 2, 3, 4}
	\draw (-2, \y) -- (2, \y);
	
	\foreach \x in {-2, -1, 0, 1, 2}
	\draw (\x, 0) -- (\x, 4);
	
	\node [below, lightgreen, scale=2.5] at (0,0) {$\Mn{A}{2}$};

	\node [ lightgreen, scale=1.2] at (-1.5,2.5) {$\Mn{A}{2}_{31}$};
	
	\draw [->, red] (-1,-1.5) -- (1,-1.5) node [below right, scale=2] {$r_2$};
	\draw [->, red] (-2.5,1) -- (-2.5,3) node [ above left, scale=2] {$n_2$};
	\end{tikzpicture}
	\vspace*{-0.15cm}\caption{Subtensor $\T{X}_{231}$ is distributed evenly among processors $(2,3,1,*,*,*)$. Similarly, submatrix $\Mn{A}{2}_{31}$ is distributed evenly among processors $(*,3,*,*,1,*)$.\label{fig:dataDistribution}\vspace*{-0.25cm}}
\end{figure}
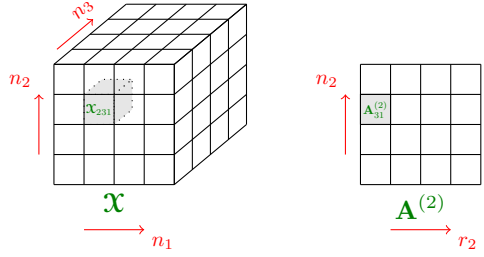


\begin{algorithm}[htb]
	\caption{\label{alg:3dmultittm}Parallel Atomic 3-dimensional Multi-TTM}
	\begin{algorithmic}[1]
		\REQUIRE $\T{X}$, $\Mn{A}{1}$, $\Mn{A}{2}$, $\Mn{A}{3}$, $p_1 \times p_2 \times p_3 \times q_1 \times q_2 \times q_3$ logical processor grid
		\ENSURE $\T{Y}$ such that $\Y = \X \times_1 {\Mn{A}{1}}^\Tra \times_2 {\Mn{A}{2}}^\Tra \times_3 {\Mn{A}{3}}^\Tra$
		\STATE $(p_1^\prime, p_2^\prime, p_3^\prime, q_1^\prime, q_2^\prime, q_3^\prime)$ is my processor id
		\STATE //All-gather input tensor $\T{X}$
		\STATE $\T{X}_{p_1^\prime p_2^\prime p_3^\prime}$ = All-Gather($\T{X}$, $(p_1^\prime, p_2^\prime, p_3^\prime, *, *, *)$)\label{alg:3dmultittm:line:allGatherInputTensor}
		\STATE //All-gather input matrices
		\STATE $\Mn{A}{1}_{p_1^\prime q_1^\prime}$ = All-Gather($\Mn{A}{1}$, $(p_1^\prime, *, *, q_1^\prime, *, *)$)\label{alg:3dmultittm:line:allGatherMatrix1}
		\STATE $\Mn{A}{2}_{p_2^\prime q_2^\prime}$ = All-Gather($\Mn{A}{2}$, $(*, p_2^\prime, *, *, q_2^\prime, *)$)\label{alg:3dmultittm:line:allGatherMatrix2}
		\STATE $\Mn{A}{3}_{p_3^\prime q_3^\prime}$ = All-Gather($\Mn{A}{3}$, $(*, *, p_3^\prime, *, *, q_3^\prime)$)\label{alg:3dmultittm:line:allGatherMatrix3}
		\STATE //Local computations in a temporary tensor $\T{T}$
		\STATE $\T{T}$ = Local-Multi-TTM($\T{X}_{p_1^\prime p_2^\prime p_3^\prime}$, $\Mn{A}{1}_{p_1^\prime q_1^\prime}$, $\Mn{A}{2}_{p_2^\prime q_2^\prime}$, $\Mn{A}{3}_{p_3^\prime q_3^\prime}$)\label{alg:3dmultittm:line:localcomputation}
		\STATE //Reduce-scatter the output tensor in $\T{Y}_{q_1^\prime q_2^\prime q_3^\prime}$
		\STATE Reduce-Scatter($\T{Y}_{q_1^\prime q_2^\prime q_3^\prime}$, $\T{T}$, $(*, *, *, q_1^\prime, q_2^\prime, q_3^\prime)$)\label{alg:3dmultittm:line:reduceScatterOutputTensor}
	\end{algorithmic}
\end{algorithm}
\Cref{alg:3dmultittm} presents a parallel algorithm to compute $3$-dimensional Multi-TTM. When it completes, $\T{Y}_{q_1^\prime q_2^\prime q_3^\prime}$ is distributed evenly among processors $(*,$ $*,$ $*,$ $q_1^\prime,$ $q_2^\prime,$ $q_3^\prime)$. \Cref{fig:stepsOfParallelAlgorithm} shows the steps of the algorithm for a single processor in a $3\times3\times3\times3\times3\times3$ grid.

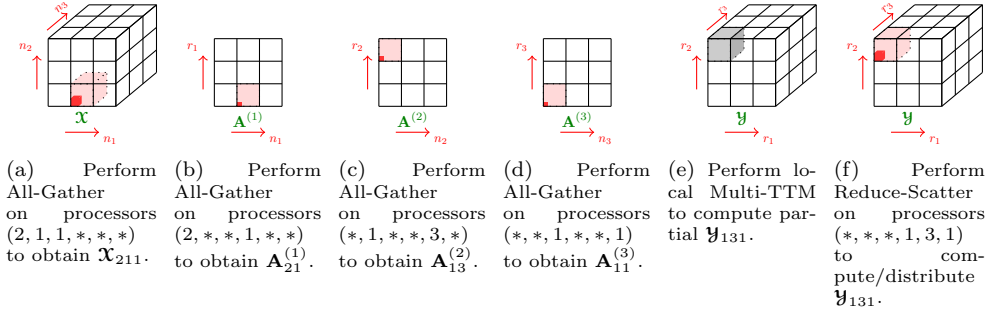
\begin{figure*}[t]
	\begin{center}
		\subfloat[{\scriptsize Perform All-Gather on processors $(2,1,1,*,*,*)$ to obtain $\T{X}_{211}$.}]{
			\begin{tikzpicture}[scale=0.3, every node/.style={transform shape}]
			\def\xref{0.6}
			\def\yref{0.5}
			
			\draw [dotted,fill=red!15] (1,0) -- (2,0) -- (2,1) --(1,1) --(1,0);
			\draw [dotted,fill=red!15] (1,1) --(1+\xref, 1+\yref) -- (2+\xref,1+\yref) -- (2,1) -- (1,1);
			\draw [dotted,fill=red!15] (2,0) -- (2+\xref, 0+\yref) -- (2+\xref, 1+\yref) -- (2,1) -- (2,0);
			
			\def\a{1}
			\def\b{0}
			\def\xstride{0.35}
			\def\ystride{0.35}
			\draw [draw=none,fill=red!80] (\a,\b) -- (\a+\xstride,\b) -- (\a+\xstride,\b+\ystride) --(\a,\b+\ystride) --(\a,\b);
			
			\def\xsmallref{0.135}
			\def\ysmallref{0.125}
			\draw [draw=none,fill=red!80] (\a,\b+\ystride) --(\a+\xsmallref, \b+\ystride+\ysmallref) -- (\a+\xstride+\xsmallref, \b+\ystride+\ysmallref) -- (\a+\xstride, \b+\ystride) -- (\a, \b+\ystride);
			
			\draw [draw=none,fill=red!80] (\a+\xstride,\b) --(\a+\xstride+\xsmallref, \b+\ysmallref) -- (\a+\xstride+\xsmallref, \b+\ystride+\ysmallref) -- (\a+\xstride, \b+\ystride) -- (\a+\xstride, \b);
			
			\foreach \y in {0, 1, 2, 3}
			\draw (0, \y) -- (3, \y);
			
			\foreach \x in {0, 1, 2, 3}
			\draw (\x, 0) -- (\x, 3);
			
			\foreach \y in {0, 1, 2, 3}
			\draw (3, \y)--(3+3*\xref, 3*\yref+\y);

			\foreach \y in {0, 1, 2, 3}
			\draw (3+\y * \xref, \y * \yref) -- (3+\y * \xref, 3+\y * \yref);
			
			\foreach \x in {0, 1, 2, 3}
			\draw (\x, 3) -- (\x + 3*\xref, 3+3*\yref);
			
			\foreach \y in {0, 1, 2, 3}
			\draw (0+\y * \xref, 3+\y * \yref) -- (3+\y * \xref, 3+\y * \yref);
			
			\node [below, lightgreen, scale=2] at (1.5,0) {$\T{X}$};

			\draw [->, red] (0.75,-1.15) -- (2.25,-1.15) node [below right, scale=1.6] {$n_1$};
			\draw [->, red] (-0.5,0.75) -- (-0.5,2.25) node [ above left, scale=1.6] {$n_2$};
			\draw [->, red] (-0.5+0.75*\xref, 3+0.75*\yref) -- (-0.5 + 2.25*\xref, 3+2.25*\yref) node [above,rotate=45, scale=1.6] {$n_3$};
			\path (0,0) -- (4.75,0);
			\end{tikzpicture}}$\ $
		\subfloat[{\scriptsize Perform All-Gather on processors $(2,*,*,1,*,*)$ to obtain $\Mn{A}{1}_{21}$.}]{
			\begin{tikzpicture}[scale=0.3, every node/.style={transform shape}]
			
			\draw [dotted,fill=red!15] (1,0) -- (2,0) -- (2,1) --(1,1) --(1,0);
			
			\def\a{1}
			\def\b{0}
			\def\xstride{0.2}
			\def\ystride{0.2}
			\draw [draw=none,fill=red!80] (\a,\b) -- (\a+\xstride,\b) -- (\a+\xstride,\b+\ystride) --(\a,\b+\ystride) --(\a,\b);
			
			\foreach \y in {0, 1, 2, 3}
			\draw (0, \y) -- (3, \y);
			
			\foreach \x in {0, 1, 2, 3}
			\draw (\x, 0) -- (\x, 3);
			
			\node [below, lightgreen, scale=2] at (1.5,0) {$\Mn{A}{1}$};
			
			\draw [->, red] (0.75,-1.15) -- (2.25,-1.15) node [below right, scale=1.6] {$n_1$};
			\draw [->, red] (-0.5,0.75) -- (-0.5,2.25) node [ above left, scale=1.6] {$r_1$};
			\path (0,0) -- (4.75,0);
			\end{tikzpicture}}$\ $
		\subfloat[{\scriptsize Perform All-Gather on processors $(*,1,*,*,3,*)$ to obtain $\Mn{A}{2}_{13}$.}]{
			\begin{tikzpicture}[scale=0.3, every node/.style={transform shape}]
			
			\draw [dotted,fill=red!15] (0,2) -- (1,2) -- (1,3) --(0,3) --(0,2);
			
			\def\a{0}
			\def\b{2}
			\def\xstride{0.2}
			\def\ystride{0.2}
			\draw [draw=none,fill=red!80] (\a,\b) -- (\a+\xstride,\b) -- (\a+\xstride,\b+\ystride) --(\a,\b+\ystride) --(\a,\b);
			
			\foreach \y in {0, 1, 2, 3}
			\draw (0, \y) -- (3, \y);
			
			\foreach \x in {0, 1, 2, 3}
			\draw (\x, 0) -- (\x, 3);
			
			\node [below, lightgreen, scale=2] at (1.5,0) {$\Mn{A}{2}$};
			
			\draw [->, red] (0.75,-1.15) -- (2.25,-1.15) node [below right, scale=1.6] {$n_2$};
			\draw [->, red] (-0.5,0.75) -- (-0.5,2.25) node [ above left, scale=1.6] {$r_2$};
			\path (0,0) -- (4.75,0);
			\end{tikzpicture}}$\ $
		\subfloat[{\scriptsize Perform All-Gather on processors $(*,*,1,*,*,1)$ to obtain $\Mn{A}{3}_{11}$.}]{
			\begin{tikzpicture}[scale=0.3, every node/.style={transform shape}]
			
			\draw [dotted,fill=red!15] (0,0) -- (1,0) -- (1,1) --(0,1) --(0,0);
			
			\def\a{0}
			\def\b{0}
			\def\xstride{0.2}
			\def\ystride{0.2}
			\draw [draw=none,fill=red!80] (\a,\b) -- (\a+\xstride,\b) -- (\a+\xstride,\b+\ystride) --(\a,\b+\ystride) --(\a,\b);
			
			\foreach \y in {0, 1, 2, 3}
			\draw (0, \y) -- (3, \y);
			
			\foreach \x in {0, 1, 2, 3}
			\draw (\x, 0) -- (\x, 3);
			
			\node [below, lightgreen, scale=2] at (1.5,0) {$\Mn{A}{3}$};
			
			\draw [->, red] (0.75,-1.15) -- (2.25,-1.15) node [below right, scale=1.6] {$n_3$};
			\draw [->, red] (-0.5,0.75) -- (-0.5,2.25) node [ above left, scale=1.6] {$r_3$};
			\path (0,0) -- (4.75,0);
			\end{tikzpicture}}$\ $
		\subfloat[{\scriptsize Perform local Multi-TTM to compute partial $\T{Y}_{131}$.\label{fig:step:partialResults}}]{
			\begin{tikzpicture}[scale=0.3, every node/.style={transform shape}]
			\def\xref{0.6}
			\def\yref{0.5}
			\draw [dotted,fill=gray!40] (0,2) -- (1,2) -- (1,3) --(0,3) --(0,2);
			\draw [dotted,fill=gray!40] (0,3) --(0+\xref, 3+\yref) -- (1+\xref,3+\yref) -- (1,3) -- (0,3);
			\draw [dotted,fill=gray!60] (1,2) -- (1+\xref, 2+\yref) -- (1+\xref, 3+\yref) -- (1,3) -- (1,2);
			
			\foreach \y in {0, 1, 2, 3}
			\draw (0, \y) -- (3, \y);
			
			\foreach \x in {0, 1, 2, 3}
			\draw (\x, 0) -- (\x, 3);
			
			\foreach \y in {0, 1, 2, 3}
			\draw (3, \y)--(3+3*\xref, 3*\yref+\y);
			
			\foreach \y in {0, 1, 2, 3}
			\draw (3+\y * \xref, \y * \yref) -- (3+\y * \xref, 3+\y * \yref);
			
			\foreach \x in {0, 1, 2, 3}
			\draw (\x, 3) -- (\x + 3*\xref, 3+3*\yref);
			
			\foreach \y in {0, 1, 2, 3}
			\draw (0+\y * \xref, 3+\y * \yref) -- (3+\y * \xref, 3+\y * \yref);
			
			\node [below, lightgreen, scale=2] at (1.5,0) {$\T{Y}$};
			
			\draw [->, red] (0.75,-1.15) -- (2.25,-1.15) node [below right, scale=1.6] {$r_1$};
			\draw [->, red] (-0.5,0.75) -- (-0.5,2.25) node [ above left, scale=1.6] {$r_2$};
			\draw [->, red] (-0.5+0.75*\xref, 3+0.75*\yref) -- (-0.5 + 2.25*\xref, 3+2.25*\yref) node [above,rotate=45, scale=1.2] {$r_3$};
			\path (0,0) -- (4.75,0);
			\end{tikzpicture}}$\ $
		\subfloat[{\scriptsize Perform Reduce-Scatter on processors $(*,*,*,1,3,1)$ to compute/distribute $\T{Y}_{131}$.}]{
			\begin{tikzpicture}[scale=0.3, every node/.style={transform shape}]
			\def\xref{0.6}
			\def\yref{0.5}
			\draw [dotted,fill=red!15] (0,2) -- (1,2) -- (1,3) --(0,3) --(0,2);
			\draw [dotted,fill=red!15] (0,3) --(0+\xref, 3+\yref) -- (1+\xref,3+\yref) -- (1,3) -- (0,3);
			\draw [dotted,fill=red!15] (1,2) -- (1+\xref, 2+\yref) -- (1+\xref, 3+\yref) -- (1,3) -- (1,2);

			\def\a{0}
			\def\b{2}
			\def\xstride{0.35}
			\def\ystride{0.35}
			\draw [draw=none,fill=red!80] (\a,\b) -- (\a+\xstride,\b) -- (\a+\xstride,\b+\ystride) --(\a,\b+\ystride) --(\a,\b);
			
			\def\xsmallref{0.135}
			\def\ysmallref{0.125}
			\draw [draw=none,fill=red!80] (\a,\b+\ystride) --(\a+\xsmallref, \b+\ystride+\ysmallref) -- (\a+\xstride+\xsmallref, \b+\ystride+\ysmallref) -- (\a+\xstride, \b+\ystride) -- (\a, \b+\ystride);
			
			\draw [draw=none,fill=red!80] (\a+\xstride,\b) --(\a+\xstride+\xsmallref, \b+\ysmallref) -- (\a+\xstride+\xsmallref, \b+\ystride+\ysmallref) -- (\a+\xstride, \b+\ystride) -- (\a+\xstride, \b);

			\foreach \y in {0, 1, 2, 3}
			\draw (0, \y) -- (3, \y);
			
			\foreach \x in {0, 1, 2, 3}
			\draw (\x, 0) -- (\x, 3);
			
			\foreach \y in {0, 1, 2, 3}
			\draw (3, \y)--(3+3*\xref, 3*\yref+\y);
			
			\foreach \y in {0, 1, 2, 3}
			\draw (3+\y * \xref, \y * \yref) -- (3+\y * \xref, 3+\y * \yref);
			
			\foreach \x in {0, 1, 2, 3}
			\draw (\x, 3) -- (\x + 3*\xref, 3+3*\yref);
			
			\foreach \y in {0, 1, 2, 3}
			\draw (0+\y * \xref, 3+\y * \yref) -- (3+\y * \xref, 3+\y * \yref);
			
			\node [below, lightgreen, scale=2] at (1.5,0) {$\T{Y}$};
			
			\draw [->, red] (0.75,-1.15) -- (2.25,-1.15) node [below right, scale=1.6] {$r_1$};
			\draw [->, red] (-0.5,0.75) -- (-0.5,2.25) node [ above left, scale=1.6] {$r_2$};
			\draw [->, red] (-0.5+0.75*\xref, 3+0.75*\yref) -- (-0.5 + 2.25*\xref, 3+2.25*\yref) node [above,rotate=45, scale=1.6] {$r_3$};
			\path (0,0) -- (4.925,0);
			\end{tikzpicture}}
		\caption{Steps of \cref{alg:3dmultittm} for processor $(2,1,1,1,3,1)$, where $p_1=p_2=p_3=q_1=q_2=q_3=3$. Highlighted areas correspond to the data blocks on which the processor is operating. The dark red highlighting represents the input/output data initially/finally owned by the processor, and the light red highlighting corresponds to received/sent data from/to other processors in All-Gather/Reduce-Scatter collectives to compute $\T{Y}_{131}$. 
			\label{fig:stepsOfParallelAlgorithm}
		}
	\end{center}	
\end{figure*}

\subsection{Cost Analysis}
\label{sec:cost-3d}

Now we analyze computation and communication costs of the algorithm.
The dimension of the local tensor $\T{X}_{p_1^\prime p_2^\prime p_3^\prime}$ is $\frac{n_1}{p_1} \times \frac{n_2}{p_2} \times \frac{n_3}{p_3}$, the dimension of the local matrix $\Mn{A}{k}_{p_k^\prime q_k^\prime}$ is $\frac{n_i}{p_i}\times \frac{r_i}{q_i}$ for $i=1,2,3$, and the dimension of the temporary tensor $\T{T}$ is $\frac{r_1}{q_1} \times \frac{r_2}{q_2} \times \frac{r_3}{q_3}$. \response{For simplicity of analysis, we assume that the numerator is divisible by the denominator for each cost expression.}

The local Multi-TTM computation in Line~\ref{alg:3dmultittm:line:localcomputation} can be performed as a sequence of TTM operations to mininimize the number of arithmetic operations. 
Assuming the TTM operations are performed in their order, first with $\Mn{A}{1}$, then with $\Mn{A}{2}$, and in the end with $\Mn{A}{3}$, then each processor performs $2\Big(\frac{n_1n_2n_3r_1}{p_1p_2p_3q_1} + \frac{n_2n_3r_1r_2}{p_2p_3q_1q_2} + \frac{n_3r_1r_2r_3}{p_3q_1q_2q_3}\Big)$ operations to perform the local computation.

Communication occurs only in the All-Gather and Reduce-Scatter collectives in Lines~\ref{alg:3dmultittm:line:allGatherInputTensor}, \ref{alg:3dmultittm:line:allGatherMatrix1}, \ref{alg:3dmultittm:line:allGatherMatrix2}, \ref{alg:3dmultittm:line:allGatherMatrix3} and~\ref{alg:3dmultittm:line:reduceScatterOutputTensor}. 
Each processor is involved in one All-Gather involving the input tensor, three All-Gathers involving input matrices and one Reduce-Scatter involving the output tensor.
\response{Therefore, the communication cost of the algorithm along the critical path is the sum of communication costs of these five collectives.} 
\response{Lines~\ref{alg:3dmultittm:line:allGatherInputTensor}, \ref{alg:3dmultittm:line:allGatherMatrix1}, \ref{alg:3dmultittm:line:allGatherMatrix2}, and \ref{alg:3dmultittm:line:allGatherMatrix3} specify $p$, $p_1q_1$, $p_2q_2$ and $p_3q_3$ All-Gathers over disjoint sets of $\frac{P}{p}$, $\frac{P}{p_1q_1}$, $\frac{P}{p_2q_2}$ and $\frac{P}{p_3q_3}$ processors respectively. Similarly, Line~\ref{alg:3dmultittm:line:reduceScatterOutputTensor} specifies $q$ Reduce-Scatters over disjoint sets of $\frac{P}{q}$ processors.}


For simplicity of discussion, we consider that the number of processors involved in the collectives is a power of $2$. We also assume that communication optimal collective algorithms are used.
The optimal latency and bandwidth costs of both collectives on $Q$ processors are $\log_2 (Q)$ and $(1-\frac{1}{Q})w$, respectively, where $w$ denotes the words of data in each processor after All-Gather or before Reduce-Scatter collective.
Each processor also performs $(1-\frac{1}{Q})w$ computations for the Reduce-Scatter collective.
We point the reader to~\cite{Thakur:CollectiveCommunications:2005, Chan:CollectiveCommunications:2007} for more details on efficient algorithms for collectives.


Hence the bandwidth costs of Lines~\ref{alg:3dmultittm:line:allGatherInputTensor}, \ref{alg:3dmultittm:line:allGatherMatrix1}, \ref{alg:3dmultittm:line:allGatherMatrix2}, \ref{alg:3dmultittm:line:allGatherMatrix3} in \cref{alg:3dmultittm} are $(1-\frac{p}{P}) \frac{n}{p}$, $(1-\frac{p_1q_1}{P}) \frac{n_1r_1}{p_1q_1}$, $(1-\frac{p_2q_2}{P}) \frac{n_2r_2}{p_2q_2}$, $(1-\frac{p_3q_3}{P})\frac{n_3r_3}{p_3q_3}$ respectively to accomplish All-Gather operations, and the bandwidth cost of performing the Reduce-Scatter operation in Line~\ref{alg:3dmultittm:line:reduceScatterOutputTensor} is $(1-\frac{q}{P}) \frac{r}{q}$. Thus the overall bandwidth cost of \cref{alg:3dmultittm} is 
\begin{equation}
\label{eq:alg-costs}
\frac{n}{p} + \frac{n_1r_1}{p_1q_1} + \frac{n_2r_2}{p_2q_2} + \frac{n_3r_3}{p_3q_3} + \frac{r}{q} - \left(\frac{n + n_1r_1 + n_2r_2 + n_3r_3 + r}{P}\right).
\end{equation} 

The latency costs of Lines~\ref{alg:3dmultittm:line:allGatherInputTensor}, \ref{alg:3dmultittm:line:allGatherMatrix1}, \ref{alg:3dmultittm:line:allGatherMatrix2}, \ref{alg:3dmultittm:line:allGatherMatrix3}, \ref{alg:3dmultittm:line:reduceScatterOutputTensor} are $\log_2 (\frac{P}{p})$, $\log_2 (\frac{P}{p_1q_1})$, $\log_2 (\frac{P}{p_2q_2})$, $\log_2 (\frac{P}{p_3q_3})$, $\log_2 (\frac{P}{q})$ respectively. Thus the overall latency cost of \cref{alg:3dmultittm} is $\log_2 \left(\frac{P}{p}\right) + \log_2 \left(\frac{P}{p_1q_1}\right) + \log_2 \left(\frac{P}{p_2q_2}\right) + \log_2 \left(\frac{P}{p_3q_3}\right) + \log_2 \left(\frac{P}{q}\right) = \log_2 \left(\frac{P^5}{p^2q^2}\right) = 3 \log_2 (P).$

Due to the Reduce-Scatter operation, each processor also performs $(1-\frac{q}{P}) \frac{r}{q}$ computations, which is \response{dominated by} the computations of Line~\ref{alg:3dmultittm:line:localcomputation} \response{(as $n_3\geq p_3$)}.

\subsection{Selection of $p_i$ and $q_i$ in Algorithm~\ref{alg:3dmultittm}}
\label{sec:parallelAlgoritm:selectionpiqi}

We must select the processor dimensions carefully such that \cref{alg:3dmultittm} is communication optimal.

We attempt to select the processor dimensions $p_i$ and $q_i$ in such a way that the terms in the communication cost match the optimal solutions of~\cref{lemma:tensorOptimalSolutions,lemma:matrixOptimalSolutions}.
In other words, we want to select $p_i$ and $q_i$ such that $\frac{n_1r_1}{p_1q_1}=\starontop{x}$, $\frac{n_2r_2}{p_2q_2}=\starontop{y}$, and $\frac{n_3r_3}{p_3q_3}=\starontop{z}$ from \cref{lemma:matrixOptimalSolutions}, and $\frac{n}{p}=\starontop{v},\frac{r}{q}=\starontop{u}$ from \cref{lemma:tensorOptimalSolutions}.

\response{We need to fix two or three processor grid dimensions for each equation, and each processor grid dimension appears in two equations.}
In general, we are able to set the processor grid dimensions in a way that is consistent with these equations.
However, they are subject to additional constraints that are not imposed by the optimization problem. 
Specifically, we have $1\le p_i\le n_i$ and $1\le q_i\le r_i$ for $1\le i \le 3$.
The lower bounds are imposed because processor grid dimensions must be at least 1.
The upper bounds are imposed to ensure that each processor performs its fair share of the computations. 
We assume that $P\leq nr$, so that every processor has at least one $4$-ary multiplication term to compute.
For simplicity, we assume that the final grid dimensions are integers and perfectly divide the corresponding input and output dimensions. However, we also discuss how to handle non-integer grid dimensions for a specific set of inputs in \cref{sec:exp:configurations}.

In order to define processor grid dimensions, we begin by determining a set of values that match the lower bound terms and denote these by $\init{p_i}, \init{q_i}$ with their products denoted $\init{p}$ and $\init{q}$.
Then, we will consider how to adapt $\init{p_i}$ and $\init{q_i}$ so that the additional constraints are met. 
During the adaption, we maintain the tensor communication costs, modify the matrix communication costs, and then bound the additional costs in terms of communication lower bounds of tensors.


As $\T{X}$ and $\T{Y}$ are $3$-dimensional tensors, we have $n_i, r_i \ge 2$ for all $1\le i\le 3$. 
For better readability, we use the notation $\odata=\frac{\sum_{j \in [3]}n_jr_j + r + n}{P}$, the amount of data owned by a single processor at the beginning and end of the algorithm.

\begin{theorem}\label{theorem:optimality:3DMultiTTMAlgorithm}
	There exist $p_i,q_i$ with $1\le p_i\le n_i, 1\le q_i\le r_i$ for $i=1,2,3$ such that \cref{alg:3dmultittm} is communication optimal to within a constant factor.
\end{theorem}

\begin{proof}
	We break our analysis into 2 scenarios which are further broken down into cases. In each case, we obtain $\init{p_i}$ and $\init{q_i}$ such that the terms in the communication cost match the corresponding lower bound terms and also satisfy at least one of the two constraints: $\forall i, 1\le \init{p_i} \le n_i, 1\le \init{q_i}$ or $\forall i, 1\le \init{q_i}\le r_i, 1\le \init{p_i}$. We handle all cases from both scenarios together in the end, and adapt these values to get $p_i$ and $q_i$ which respect both lower and upper bounds.
	
	\noindent$\bullet$ \underline{Scenario I} $\left(P < \frac{n}{r}\right)$:
	This scenario corresponds to the first case of the tensor term in $\lowerbound$.
	Thus, we set $\init{p_i}, \init{q_i}$ in such a way that the tensor terms in the communication cost match the tensor terms of $\lowerbound$:
	\begin{equation}\label{eq:s1} \init{p}=P, \init{q} = 1.\end{equation}
	This implies $\init{q_i} = 1$ for $1\le i\le 3$.
	We break this scenario into 3 cases, each corresponding to a case in the matrix term of $\lowerbound$.
	
	\noindent (Case 1) $P < \response{\frac{n_3r_3}{n_2r_2}}$:
	Setting the matrix communication costs to the matrix terms in the corresponding case of the lower bound yields
	\begin{equation}\label{eq:s1c1}\frac{n_1r_1}{\init{p_1}\init{ q_1} } = n_1r_1, \: \frac{n_2r_2}{\init{p_2}\init{q_2} } = n_2r_2, \: \frac{n_3r_3}{\init{p_3}\init{q_3}} = \frac{n_3r_3}{P}.\end{equation}
	Thus, we set $\init{p_1} = \init{p_2} = \init{q_1} = \init{q_2} = \init{q_3} = 1$ and $\init{p_3} = P$ to satisfy~\cref{eq:s1,eq:s1c1}.
	
	\noindent (Case 2) $\frac{n_3r_3}{n_2r_2} \le P < \response{\frac{n_2n_3r_2r_3}{n_1^2r_1^2}}$:
	Setting the matrix communication costs to the matrix terms in the corresponding case of the lower bound yields
	\begin{equation}\label{eq:s1c2} \frac{n_1r_1}{\init{p_1}\init{q_1}} = n_1r_1,\: \frac{n_2r_2}{\init{p_2}\init{q_2}} = \frac{n_3r_3}{\init{p_3}\init{q_3}} = \left(\frac{n_2n_3r_2r_3}{P}\right)^{1/2}.\end{equation}
	We set $\init{p_1} = \init{q_1} = \init{q_2} = \init{q_3} = 1$, $\init{p_2} = n_2r_2\left(\frac{P}{n_2n_3r_2r_3}\right)^{\frac 12}$, and $\init{p_3} = n_3r_3\left(\frac{P}{n_2n_3r_2r_3}\right)^{\frac 12}$ to satisfy~\cref{eq:s1,eq:s1c2}. \response{$\frac{n_3r_3}{n_2r_2} \le P$ implies $\init{p_2}\ge 1$ and $\init{p_3}\ge 1$.}
	
	\noindent (Case 3) $\frac{n_2n_3r_2r_3}{n_1^2r_1^2} \le P$:
	Setting the matrix communication costs to match the matrix terms in the corresponding case of the lower bound yields
	\begin{equation}\label{eq:s1c3}\frac{n_1r_1}{\init{p_1}\init{q_1}} = \frac{n_2r_2}{\init{p_2}\init{q_2}} = \frac{n_3r_3}{\init{p_3}\init{q_3}} = \left(\frac{nr}{P}\right)^{1/3}. \end{equation}
	Thus we set $\init{q_1}=\init{q_2}=\init{q_3} = 1$, $\init{p_1} = n_1r_1 \big(\frac{P}{nr}\big)^{\frac{1}{3}},$ $\init{p_2} = n_2r_2 \big(\frac{P}{nr}\big)^{\frac{1}{3}},$ and $\init{p_3} = n_3r_3 \big(\frac{P}{nr}\big)^{\frac{1}{3}}$ to satisfy~\cref{eq:s1,eq:s1c3}. \response{$\frac{n_2n_3r_2r_3}{n_1^2r_1^2} \le P$ implies $\init{p_i}\ge 1$ for $1\le i \le 3$.}

	\medskip
	\response{Note that in all the cases of this scenario we have $1\le \init{q_i}=1 < r_i, 1\le \init{p_i}$ for $1\le i \le 3$, but we cannot ensure $\init{p_i} \le n_i$.} 
	We will adapt processor grid dimensions for both scenarios in the end as they require the same steps.
	
	\medskip
	\noindent $\bullet$ \underline{Scenario II} $\left(\frac{n}{r}\le P\right)$:
	This scenario corresponds to the second case of the tensor term in $\lowerbound$.
	Thus, we set $\init{p_i}, \init{q_i}$ in such a way that
	\begin{equation}\label{eq:s2}\frac{n}{\init{p}}=\frac{r}{\init{q}}=\left(\frac{nr}{P}\right)^{1/2}.\end{equation}
	Again, we break this scenario into 3 cases each corresponding to a case in the matrix term of $\lowerbound$.
	
	\noindent (Case 1) $P < \frac{n_3r_3}{n_2r_2}$:
	Setting the matrix communication costs to the matrix terms in the corresponding case of the lower bound yields
	\begin{equation}\label{eq:s2c1}\frac{n_1r_1}{\init{p_1}\init{q_1}} = n_1r_1, \: \frac{n_2r_2}{\init{p_2}\init{q_2}} = n_2r_2, \: \frac{n_3r_3}{\init{p_3}\init{q_3}} = \frac{n_3r_3}{P}.\end{equation}
	Thus we set $\init{p_1} = \init{q_1} = \init{p_2} = \init{q_2} = 1$, $\init{p_3} = n\left(\frac{P}{nr}\right)^{1/2}$, and $\init{q_3} = r\left(\frac{P}{nr}\right)^{1/2}$ to satisfy~\cref{eq:s2,eq:s2c1}.
	As $\frac{n}{r} \leq P \le nr$ and $r \le n$, we have $1\le \init{p_3} \le n$ and $1 \le \init{q_3} \le r$, but cannot ensure $\init{p_3}\le n_3$ or $\init{q_3}\le r_3$. However, $\init{p_3}\init{q_3}=P < \frac{n_3r_3}{n_2r_2}$ implies that at least one is satisfied.
	 Therefore, we have $\forall i, 1\le \init{p_i} \le n_i, 1\le \init{q_i}$ and/or $\forall i, 1\le \init{p_i}, 1\le \init{q_i} \le r_i$.

	\noindent (Case 2) $\response{\frac{n_3r_3}{n_2r_2}} \leq P < \frac{n_2n_3r_2r_3}{n_1^2r_1^2}$:
	Setting the matrix communication costs to the matrix terms in the corresponding case of the lower bound yields
	\begin{equation}\label{eq:s2c2}\frac{n_1r_1}{\init{p_1}\init{q_1}} = n_1r_1, \frac{n_2r_2}{\init{p_2}\init{q_2}} = \frac{n_3r_3}{\init{p_3}\init{q_3}} = \left(\frac{n_2n_3r_2r_3}{P}\right)^{1/2}.\end{equation}
	\response{We set $\init{p_1}=\init{q_1}=1$. \Cref{eq:s2,eq:s2c2} do not uniquely determine $\init{p_2},\init{p_3},\init{q_2},$ and $\init{q_3}$.
	The following is one possible solution: $\init{p_2}=n_2\left(\frac{n_1P}{n_2n_3r}\right)^{1/4}$, $\init{p_3}=n_3\left(\frac{n_1P}{n_2n_3r}\right)^{1/4}$, $\init{q_2}=r_2\left(\frac{r_1P}{nr_2r_3}\right)^{1/4}$, and $\init{q_3}=r_3\left(\frac{r_1P}{nr_2r_3}\right)^{1/4}$}.
	Note that $P < \frac{n_2n_3r_2r_3}{n_1^2r_1^2}$ implies that $\init{p_2} < n_2,$ $\init{p_3} < n_3,$ $\init{q_2}<r_2,$ and $\init{q_3} < r_3$.
	We are not able to ensure $\init{p_2}, \init{p_3}, \init{q_2}, \init{q_3}$ are all at least 1 in this case.
	We will handle
	both Case 2 and Case 3 together as they require the same analysis.
	
	\noindent (Case 3) $\response{\frac{n_2n_3r_2r_3}{n_1^2r_1^2}} \le P$:
	Setting the matrix communication costs to the matrix terms in the corresponding case of the lower bound yields
	\begin{equation}\label{eq:s2c3} \frac{n_1r_1}{\init{p_1}\init{q_1}} = \frac{n_2r_2}{\init{p_2}\init{q_2}} = \frac{n_3r_3}{\init{p_3}\init{q_3}} = \left(\frac{nr}{P}\right)^{\frac{1}{3}}.\end{equation}
	Similar to Case 2, the equations \cref{eq:s2,eq:s2c3} do not uniquely determine $\init{p_i},\init{q_i}$ for $1\le i\le 3$.
	We choose a cubical distribution, namely $\frac{n_1}{p_1}=\frac{n_2}{p_2}=\frac{n_3}{p_3}=\frac{r_1}{q_1} =\frac{r_2}{q_2} =\frac{r_3}{q_3}$ and obtain the following solution, $\init{p_i} = n_i\left(\frac{P}{nr}\right)^{1/6}$, $\init{q_i} = r_i\left(\frac{P}{nr}\right)^{1/6}$ for $1\le i\le 3.$
	As $P\le nr$ we have $\init{p_i}\le n_i$ and $\init{q_i}\le r_i$ for $1\le i \le 3$.
	Again we are not able to ensure $\init{p_i}$ and $\init{q_i}$ are all greater than 1 for $1\le i\le 3$.

	Now we handle Case 2 and Case 3 of Scenario II here. \response{The communication cost for the obtained set of values matches the lower bound, and each term in the lower bound is at least 1, therefore $1\le \frac{n_ir_i}{\init{p_i}\init{q_i}}\le n_ir_i$ for $1\le i \le 3$, $1\le \frac{n}{\init{p}} \le n$ and $1\le \frac{r}{\init{q}}\le r$. This implies that $1\le \init{p_i}\init{q_i}\le n_ir_i$ for $1\le i \le 3$, $1\le \init{p} \le n$ and $1\le \init{q}\le r$.}
		\response{For $1\le i \le 3$, at most one of $\init{p}_i$ and $\init{q}_i$ can be smaller than one. In such a case, we multiply the largest by the smallest (say $\widetilde{p}_i = \init{p}_i \cdot \init{q}_i$) and set the smallest to one ($\widetilde{q}_i = 1$) so that their product remains the same ($\widetilde{p}_i \cdot \widetilde{q}_i =  \init{p}_i \cdot \init{q}_i$).
		After doing that, the products $\widetilde{p}$ and $\widetilde{q}$ might change.
		Let $f = \widetilde{q}/\init{q}$
		be the rate of change, and suppose $f>1$. 
		As $\init{q}=\widetilde{q}/f \ge 1$, we can factor $f = f_1 \cdot f_2 \cdot f_3$ with $f_i \geq 1$ so that $\init{q}_i := \widetilde{q}_i/f_i \geq 1$ and $\init{p}_i := \widetilde{p} f_i \geq 1$.
		We can obtain the factors $f_i$ by the following iterative procedure:}
		
		\begin{minipage}{0.825\linewidth}
		\begin{tcolorbox}[colframe=white]
			\begin{enumerate}
				\item for $i = 1:3$
				\item \quad if $\widetilde{q}_i \geq f$ then \quad $f_i =f$, $f=1$, $(\init{p}_i, \init{q}_i) := (\widetilde{p}_i f_i, \widetilde{q}_i/f_i)$ 
				\item \quad else \quad $f_i = \widetilde{q}_i$, $f=f/f_i$, $(\init{p}_i, \init{q}_i) := (\widetilde{p}_if_i, 1)$
			\end{enumerate}	
		\end{tcolorbox}
		\end{minipage}
		
		It is straightforward to verify that at the end of this process, we have $1 \leq \init{q}_i \leq r_i$, and $1 \leq \init{p}_i$. 
		If $f<1$, the process is applied by exchanging the $p$'s and the $q$'s so that we end up with the inequalities $1 \leq \init{p}_i \leq n_i$, and $1 \leq \init{q}_i$. 
	
        %
	
	
	
	\response{Now we consider all the cases of both scenarios. It remains to adapt $\init{p}_i$ and $\init{q}_i$ such that $\init{p}_i \le n_i$ and $\init{q}_i\le r_i$. We can note that due to our particular selections of $p_i$ and $q_i$ in each case, $\nexists i,j\in[3]$ such that $\init{p_i}>n_i$ and $\init{q_j}>r_j$. We will use this fact while assessing the additional communication cost.}
	We now obtain $p_1,p_2,p_3, q_1, q_2, q_3$ from $\init{p_i},\init{q_i}$ such that both lower and upper bounds are respected, and $p_1p_2p_3=\init{p}$ and $q_1q_2q_3=\init{q}$. The intuition is to maintain the tensor communication terms in the lower bound.

	\response{Initially we set $p_i=\init{p_i}, q_i=\init{q_i}$ for $1\le i \le 3$. If $1\le \init{q_i}\le r_i$, $1\le \init{p_i}$ for $1\le i \le 3$ and $\init{p_l} > n_l$ for some index $l$. We represent the other two indices with $j$ and $k$. As $\init{p} \le n$, therefore $\init{p_j} \le  n_j$ or/and $\init{p_k} \le n_k$. Without loss of generality, we assume that $\init{p_k} \le  n_k$. Now we first update $p_l$, and then $p_j$, and in the end $p_k$ with the following expressions: $p_l := n_l$, $p_j := \min\left\{n_j, \frac{\init{p}}{p_k p_l}\right\}$, $p_k := \frac{\init{p}}{p_l p_j}$. We note that the product is unchanged by these updates as $p_kp_lp_j=\init{p}$. The same update can be done to $q_i$'s if $1\le \init{p_i}\le n_i$, $1\le \init{q_i}$ for $1\le i \le 3$ and $\init{q_l} > r_l$ for some $l$.}

	Now we assess how much additional communication is required for the matrices. If $\nexists i\in [3]$ such that $\init{p_i} > n_i$ or $\init{q_i} > r_i$ then $\sum_{i \in [3]}\frac{n_ir_i}{p_iq_i} = \sum_{i \in [3]}\frac{n_ir_i}{\init{p_i}\init{q_i}}$. We can note that due to our particular selections of $\init{p_i}$ and $\init{q_i}$, $\nexists i,j\in [3]$ such that $\init{p_i} > n_i$ and $\init{q_j} > r_j$. Suppose $\exists i\in [3]$ such that $\init{p_i} > n_i$ then $\init{p}> 2$ and
	\begin{align*}
	\sum_{i \in [3]}\frac{n_ir_i}{p_iq_i} &\le \sum_{i \in [3]} \max\left\{\frac{n_ir_i}{\init{p_i}\init{q_i}}, \frac{r_i}{\init{q_i}}\right\}&& \text{\colorbox{shadecolor}{ $q_i=\init{q_i}$, and $p_i\ge \init{p_i}$ or $p_i=n_i$}}\\
	&= \sum_{i \in [3]} \Big(\frac{n_ir_i}{\init{p_i}\init{q_i}} + \frac{r_i}{\init{q_i}} - \min\left\{\frac{n_ir_i}{\init{p_i}\init{q_i}}, \frac{r_i}{\init{q_i}}\right\}\Big)  && \text{\colorbox{shadecolor}{$\max\{a,b\} = a+b-\min\{a,b\}$}}\\
	& < \sum_{i \in [3]}\big(\frac{n_ir_i}{\init{p_i}\init{q_i}} + \frac{r_i}{\init{q_i}}\big) -2  && \text{\colorbox{shadecolor}{$\init{p_i}\init{q_i} \le n_ir_i$ and $\init{q_i}\le r_i$}}\\
	&\le \sum_{i \in [3]}\frac{n_ir_i}{\init{p_i}\init{q_i}}+ \frac{r}{\init{q}} && \text{\colorbox{shadecolor}{$\forall a_i\ge 1, a_1+a_2+a_3$-$2 \le a_1a_2a_3$\footnotemark}}\\
	&< \sum_{i \in [3]}\frac{n_ir_i}{\init{p_i}\init{q_i}}+ 2\big(\frac{r}{\init{q}} - \frac{r}{\init{p}\init{q}}\big)&&\\
	&= \sum_{i \in [3]}\frac{n_ir_i}{\init{p_i}\init{q_i}} + 2\big(\frac{r}{\init{q}} - \frac{r}{P}\big).&&
	\end{align*}
	\response{Similarly, if $\init{q_i} > r_i$ for some $i$ then $\sum_{i \in [3]}\frac{n_ir_i}{p_iq_i}$ is bounded by $\sum_{i \in [3]} \max\{\frac{n_ir_i}{\init{p_i}\init{q_i}}, \frac{n_i}{\init{p_i}}\}$ and we can obtain $\sum_{i \in [3]}\frac{n_ir_i}{p_iq_i} < \sum_{i \in [3]}\frac{n_ir_i}{\init{p_i}\init{q_i}} + 2\big(\frac{n}{\init{p}} - \frac{n}{P}\big)$.}
	
	\footnotetext{$\forall a_i\ge 1, a_1a_2a_3=(1+a_1-1)(1+a_2-1)(1+a_3-1) \ge 1 + (a_1-1)+ (a_2-1) + (a_3-1) = a_1 + a_2+ a_3-2$.}


	\noindent Therefore, in all situations, $\sum_{i \in [3]} \frac{n_ir_i}{p_iq_i} +\frac{r}{q} + \frac{n}{p} - \odata \le 3 \Big(\sum_{i \in [3]}\frac{n_ir_i}{\init{p_i}\init{q_i}} $ $+ \frac{r}{\init{q}} + \frac{n}{\init{p}} - \odata \Big) = 3\lowerbound$.
\end{proof}

\section{Simulated Evaluation}
\label{sec:experiments}

In this section, we verify our theoretical claims on particular sets of 3D tensor dimensions with a simulated evaluation.
\response{We use \cref{eq:alg-costs} to calculate the communication cost of \cref{alg:3dmultittm}.}
In \cref{sec:exp-optimality}, we demonstrate that the communication cost of \cref{alg:3dmultittm} matches the lower bound of \cref{theorem:lb:3DMultiTTM}, and we provide intuition for relationships among the communication costs of the individual tensors and matrices.
In \cref{sec:exp-comparison}, we compare the approach of \cref{alg:3dmultittm} for evaluating Multi-TTM with a TTM-in-Sequence approach, demonstrating realistic scenarios when \cref{alg:3dmultittm} communicates significantly less data and performs a negligible amount of extra computation.

Throughout this section, we restrict to cases where all tensor dimensions and numbers of processors are powers of 2.
We vary the number of processors $P$ from $2$ to $\maxp=\min\{n_1r_1, n_2r_2, n_3r_3, n, r\}$, which ensures that each processor owns some data of every tensor and matrix. 
The costs of \cref{alg:3dmultittm} depend on the processor grid, and in these experiments, we perform an exhaustive search for the best configuration.
We describe in \cref{sec:exp:configurations} how to adapt the processor grid selection scheme described in \cref{sec:parallelAlgoritm:selectionpiqi} to obtain integer-valued processor grid dimensions, and we show that we can obtain nearly optimal configurations without exhaustive search.

\subsection{Verifying Optimality of \Cref{alg:3dmultittm}}
\label{sec:exp-optimality}

\Cref{theorem:optimality:3DMultiTTMAlgorithm} states that \cref{alg:3dmultittm} attains the communication lower bound to within a constant factor, and in this section we verify the result in a variety of scenarios.
Recall from \cref{theorem:lb:3DMultiTTM} that the lower bound is $A+B-\odata$, where 
\begin{align*}
	A &= \begin{cases} n_1r_1 + n_2r_2 + \frac{n_3r_3}{P} & \text{ if } P<\frac{n_3r_3}{n_2r_2} \\ n_1r_1 + 2\left(\frac{n_2n_3r_2r_3}{P}\right)^{\frac{1}{2}} &\text{ if } \frac{n_3r_3}{n_2r_2}\leq P < \frac{n_2n_3r_2r_3}{n_1^2r_1^2} \\ 3\left(\frac{nr}{P}\right)^{\frac{1}{3}} &\text{ if } \frac{n_2n_3r_2r_3}{n_1^2r_1^2} \leq P\end{cases} \\
	B &= \begin{cases} r + \frac{n}{P} & \text{ if } P < \frac{n}{r} \\ 2\left(\frac{nr}{P}\right)^{\frac{1}{2}} &\text{ if } \frac{n}{r} \leq P\text.\end{cases} \\
	\odata &= \frac{n_1r_1+n_2r_2+n_3r_3 + r + n}{P}.
\end{align*}
Here, $A$ corresponds to the matrix entries accessed, $B$ corresponds to the tensor entries accessed \response{or partially computed}, and $\odata$ corresponds to the data owned by a single processor.
The costs of \cref{alg:3dmultittm} are given by \cref{eq:alg-costs}, which we re-write here as
\begin{equation*}
\frac{n_1r_1}{p_1q_1} + \frac{n_2r_2}{p_2q_2} + \frac{n_3r_3}{p_3q_3} + \frac{n}{p} + \frac{r}{q} - \odata,
\end{equation*}
where $\{p_i\}$ and $\{q_i\}$ specify the processor grid dimensions.
The first three terms correspond to matrix entries and the middle two terms correspond to tensor entries.

\begin{figure*}
	\begin{center}
		$\quad$\includegraphics[scale=0.225]{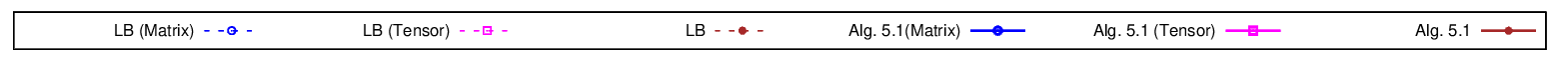}
	\end{center}
	\vspace*{-0.35cm}\begin{center}
		\subfloat[$\{n_1,n_2,n_3,r_1,r_2,r_3\}=\qquad$ $\{2^{12}, 2^{13}, 2^{19}, 2^{8},2^{13}, 2^{11}\}$.\label{fig:lb:allcases}]{\includegraphics[width=0.32\linewidth]{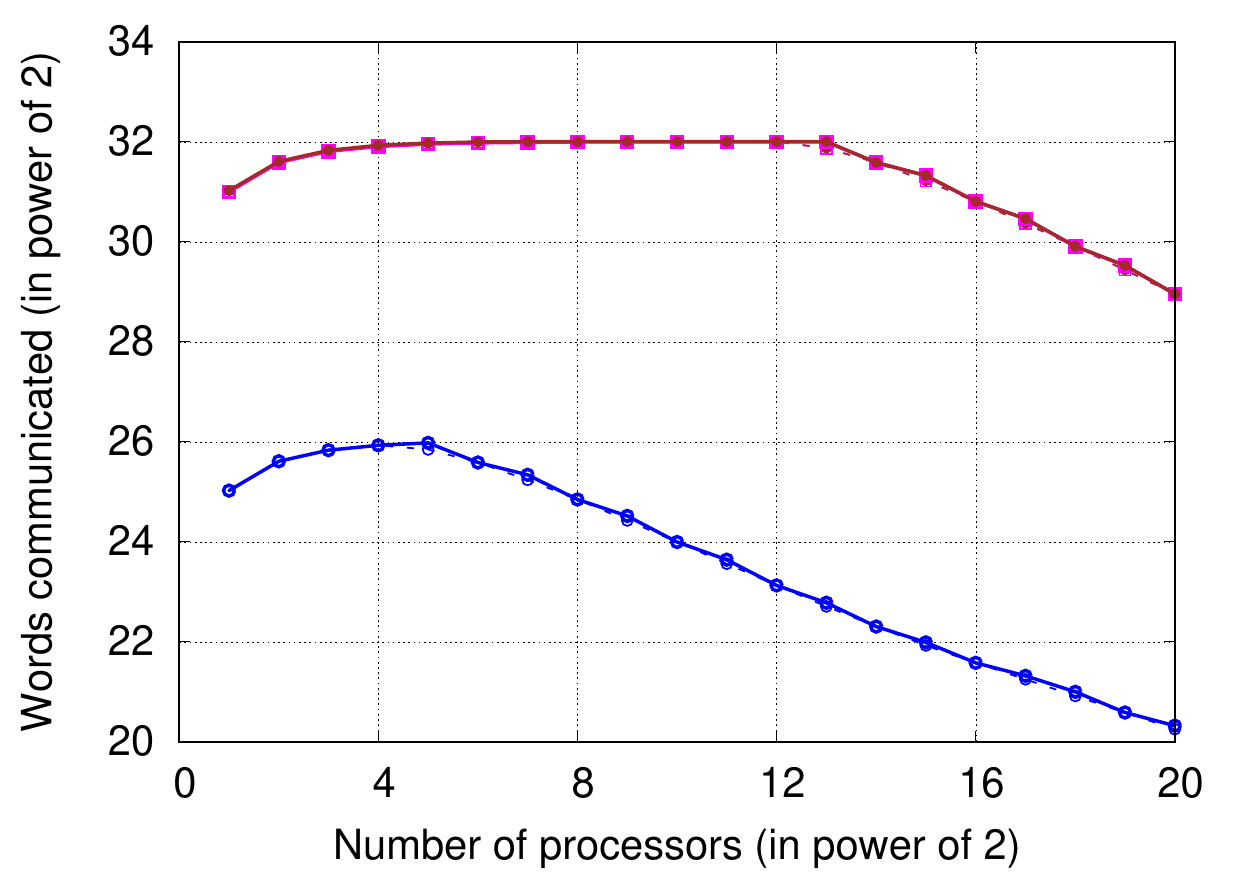}}\hfill
		\subfloat[$n_1=n_2=n_3=2^{12}$, and $r_1=r_2=r_3=2^{4}$.\label{fig:lb:matrixdominated}]{\includegraphics[width=0.32\linewidth]{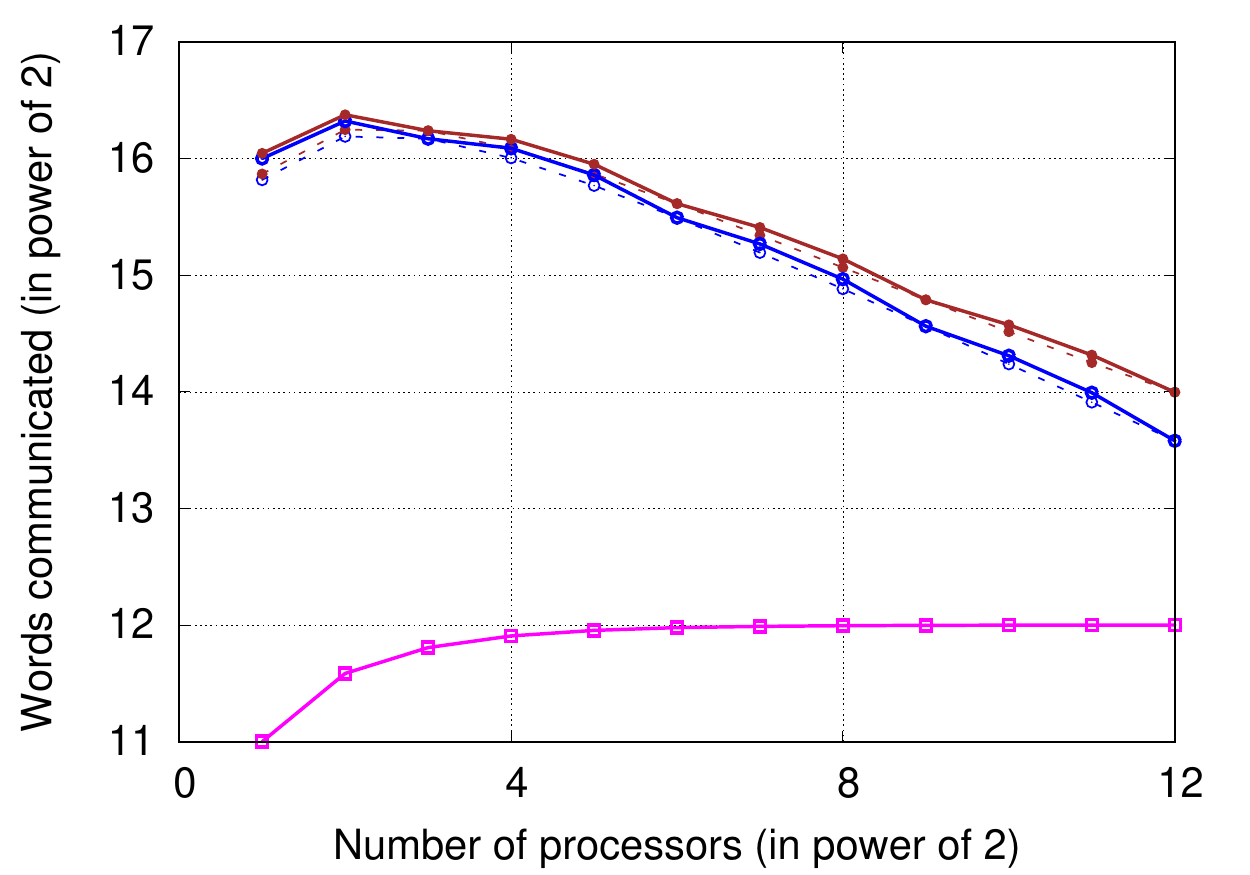}}\hfill	
		\subfloat[$n_1=n_2=n_3=2^{20}$, and $r_1=r_2=r_3=2^{8}$.\label{fig:lb:genpattern}]{\includegraphics[width=0.32\linewidth]{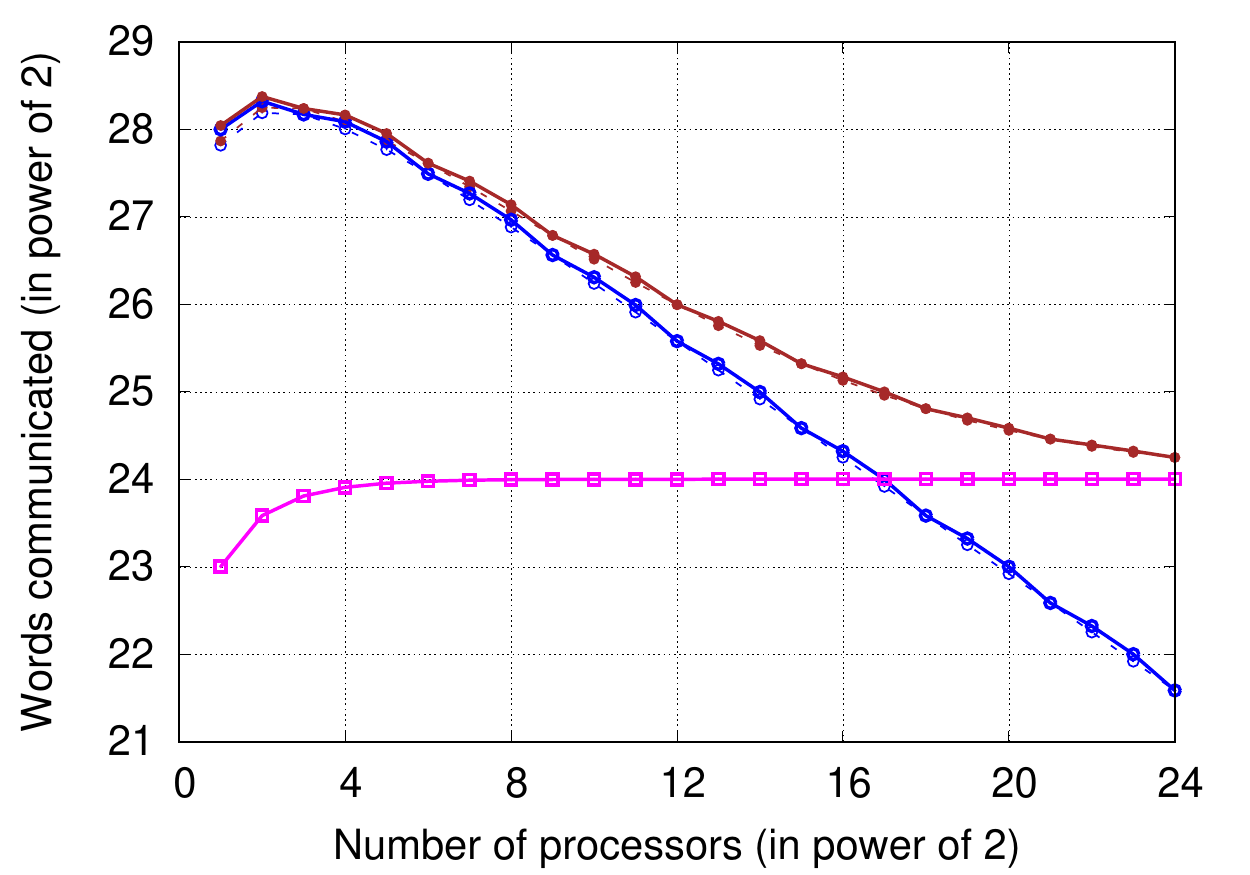}}
		\caption{Matrix and tensor communication costs in LB and \cref{alg:3dmultittm} for different configurations. The sum of $\lowerboundmatrix$ and $\lowerboundtensor$ equals to the lower bound ($\lowerbound$), and the sum of Alg. 5.1 (Matrix) and Alg. 5.1 (Tensor) equals to the upper bound (Alg. 5.1). Lower bounds are almost indistinguishable from the corresponding upper bounds.\label{fig:lb}
	 }
	\end{center}
\end{figure*}

Figure~\ref{fig:lb} shows both components, matrix and tensor communication costs, for three distinct input sizes as we vary the number of processors. 
In these plots, we show both algorithmic costs (upper bounds) and lower bounds, but they are indistinguishable because the largest differences in overall costs we observe are $9\%$ for \cref{fig:lb:allcases} at $P=2^{13}$ and $13\%$ for \Cref{fig:lb:matrixdominated,fig:lb:genpattern} at $P=2$, verifying \cref{theorem:lb:3DMultiTTM} for these scenarios.

In \cref{fig:lb:allcases}, the input and output tensors have varying dimensions: the input tensor is $2^{12}\times 2^{13}\times 2^{19}$ and the output is $2^{8}\times 2^{13}\times 2^{11}$.
We choose these dimensions so that all five cases of the values of $A$ and $B$ are represented.
For these inputs, the tensor communication cost dominates the matrix communication for all values of $P$ considered. 
When $P<2^4$, the first cases for $A$ and $B$ apply, and the algorithm selects a processor grid such that $p_3=P$, implying that only one tensor and two matrices are communicated.
In this case, both expressions simplify to $(r+n_1r_1+n_2r_2)(1-1/P)$, which is why we see initial increase as $P$ increases at the left end of the plot.
For $2^4\leq P < 2^{12}$, the second case for $A$ and the first case for $B$ apply, and the algorithm selects a processor grid with $p_2>1$ and $p_3>1$.
Here, the matrix communication begins to decrease, but it is dominated by the tensor communication, which is maintained at $r(1-1/P)$.
For $2^{12}\leq P$, the second case for $B$ applies, and we see that tensor communication decreases as $P$ increases (proportional to $P^{-1/2}$ as we see from the lower bound).
In this regime, the algorithm is selecting grids with both $p>1$ and $q>1$ and communicating both tensors.
Another transition occurs at $P=2^{16}$, switching from the second to third case of $A$, but this change in matrix cost has a negligible effect.

\Cref{fig:lb:matrixdominated} demonstrates a scenario where the matrix costs dominate the tensor costs: the input tensor is cubical with dimension $2^{12}$ and the output tensor is cubical with dimension $2^4$.
Here we scale $P$ only up to $2^{12}$, the number of entries in the output tensor.
Because the tensors are cubical, the lower bounds simplify as in \cref{corollary:lb:3DMultiTTM:cubicaltensors}, and the algorithm chooses processor grids that are as cubical as possible.
For all values of $P$ in this experiment, the third case of $A$ and the first case of $B$ apply, and the algorithm selects $p_1\approx p_2\approx p_3$ and $q=1$.
We see that the overall cost is deceasing proportional to $P^{-1/3}$ until the tensor communication cost starts to contribute more significantly.

\Cref{fig:lb:genpattern} considers cubical tensors with larger dimensions to show a more general pattern.
For tensor dimensions $n_i=2^{20}$ and $r_i=2^8$, we observe a transition point where tensor communication overtakes matrix communication.
Similar to the case of \Cref{fig:lb:matrixdominated}, matrix costs dominate for small $P$ and scale like $P^{-1/3}$.
However, for $P\geq 2^{17}$, the tensor costs dominate the matrix costs and communication costs scale less efficiently as the first case of $B$ applies.
We emphasize that for all three of these experiments, the algorithmic costs match the lower bounds nearly exactly for all values of $P$.

\subsection{Comparing \Cref{alg:3dmultittm} with TTM-in-Sequence}
\label{sec:exp-comparison}

As mentioned previously, a Multi-TTM computation may be performed as sequence of TTM operations.
In this TTM-in-Sequence approach, a single matrix is multiplied with the tensor and an intermediate tensor is computed and stored.
For each remaining matrix, single-matrix TTMs are performed in sequence until the final result is computed.
This approach can reduce the number of arithmetic operations compared to direct evaluation of atomic expression given in \cref{def:amttm}.
The computational cost depends (often significantly) on the order of the TTMs performed.
The TTM-in-Sequence approach is parallelized in the TuckerMPI library \cite{Ballard:TuckerMPI:TOMS20}.
We note that \cref{theorem:lb:3DMultiTTM} does not apply to this parallelization, as it violates the parallel atomicity assumption.

In this section, we provide a comparison between \cref{alg:3dmultittm} and the TTM-in-Sequence approach to show that our approach can significantly reduce communication in important scenarios without performing too much extra computation.
In particular, we observe greatest benefit of \cref{alg:3dmultittm} when $r$ is very small relative to $n$ (or vice versa) and $P$ is small relative to the ratio of $n$ and $r$.
These scenarios occur in the context of computing and using Tucker decompositions for highly compressible tensors that exhibit small multilinear ranks.

The computational cost of TuckerMPI's algorithm with cubical tensors is the same for all possible orderings of the TTMs. In our comparison, we consider that the TTMs are performed in increasing mode order.
While no single communication lower bound exists for all parallel TTM-in-Sequence algorithms, we show in \cref{sec:exp:lbTTM-in-Sequence} that TuckerMPI's algorithm attains nearly the same cost as tight matrix multiplication lower bounds \cite{ABGKR22} applied to each TTM it chooses to perform. 
Thus, no other parallelization of the TTM-in-Sequence approach can reduce communication without breaking the assumptions of the matrix multiplication lower bounds (e.g., using fast matrix multiplication).

The TuckerMPI parallelization uses a 3D logical processor grid with dimensions $\tilde{p_1}\times \tilde{p_2}\times \tilde{p_3}$.
When the TTMs are performed in increasing mode order, the overall communication cost of their algorithm is 
\begin{align}
\label{eq:TuckerMPI-cost}
\frac{r_1n_2n_3}{\tilde{p_2}\tilde{p_3}} + \frac{n_1r_1}{\tilde{p_1}} + \frac{r_1r_2n_3}{\tilde{p_1}\tilde{p_3}} +\frac{n_2r_2}{\tilde{p_2}} + \frac{r_1r_2r_3}{\tilde{p_1}\tilde{p_2}} + \frac{n_3r_3}{\tilde{p_3}} \phantom{\frac{r_1r_2n_3}{\tilde{p_1}\tilde{p_3}}}\qquad\qquad \\ \phantom{\frac{r_1r_2n_3}{\tilde{p_1}\tilde{p_3}}}\qquad\qquad - \frac{r_1n_2n_3+r_1r_2n_3+r_1r_2r_3 + n_1r_1 +n_2r_2+n_3r_3}{P}, \notag
\end{align}
as specified in \cite[Section 6.3]{Ballard:TuckerMPI:TOMS20}, though we include the cost of communicating the matrices (their analysis assumes the matrices are already redundantly distributed).
We use exhaustive search to determine the processor grid that minimizes the cost of \cref{eq:TuckerMPI-cost} in our comparisons.

\begin{figure*}
	\begin{center}
		\includegraphics[scale=0.195]{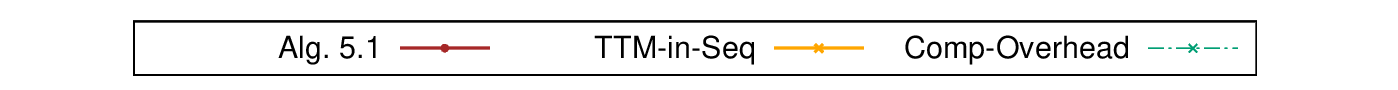}
	\end{center}
	\vspace*{-0.375cm}\begin{center}
		\subfloat[$n_i=2^{12}$, $r_i=2^{4}$.\label{fig:commcostcomparison:12-4}]{\includegraphics[width=0.31\linewidth]{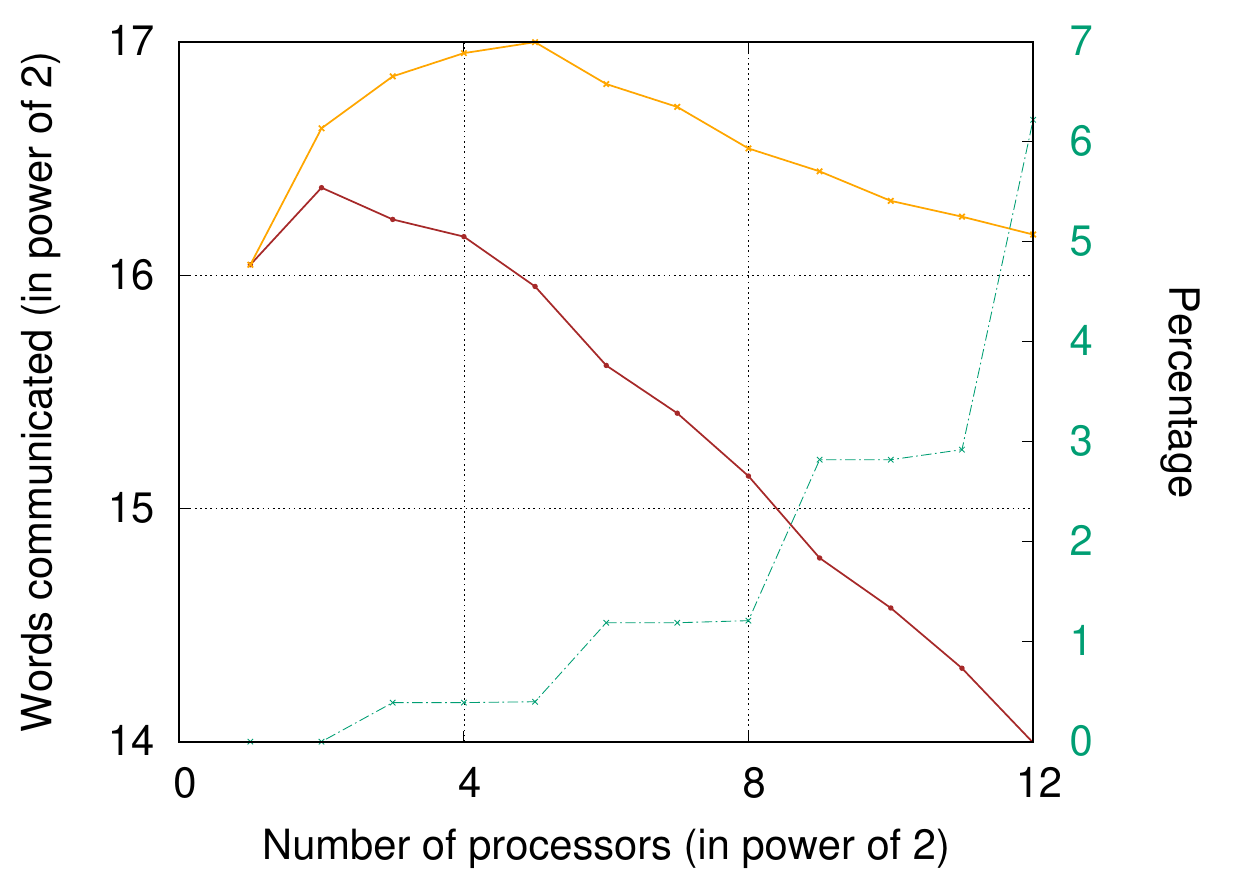}}\hfill
		\subfloat[$n_i=2^{13}$, $r_i=2^{6}$.\label{fig:commcostcomparison:13-6}]{\includegraphics[width=0.31\linewidth]{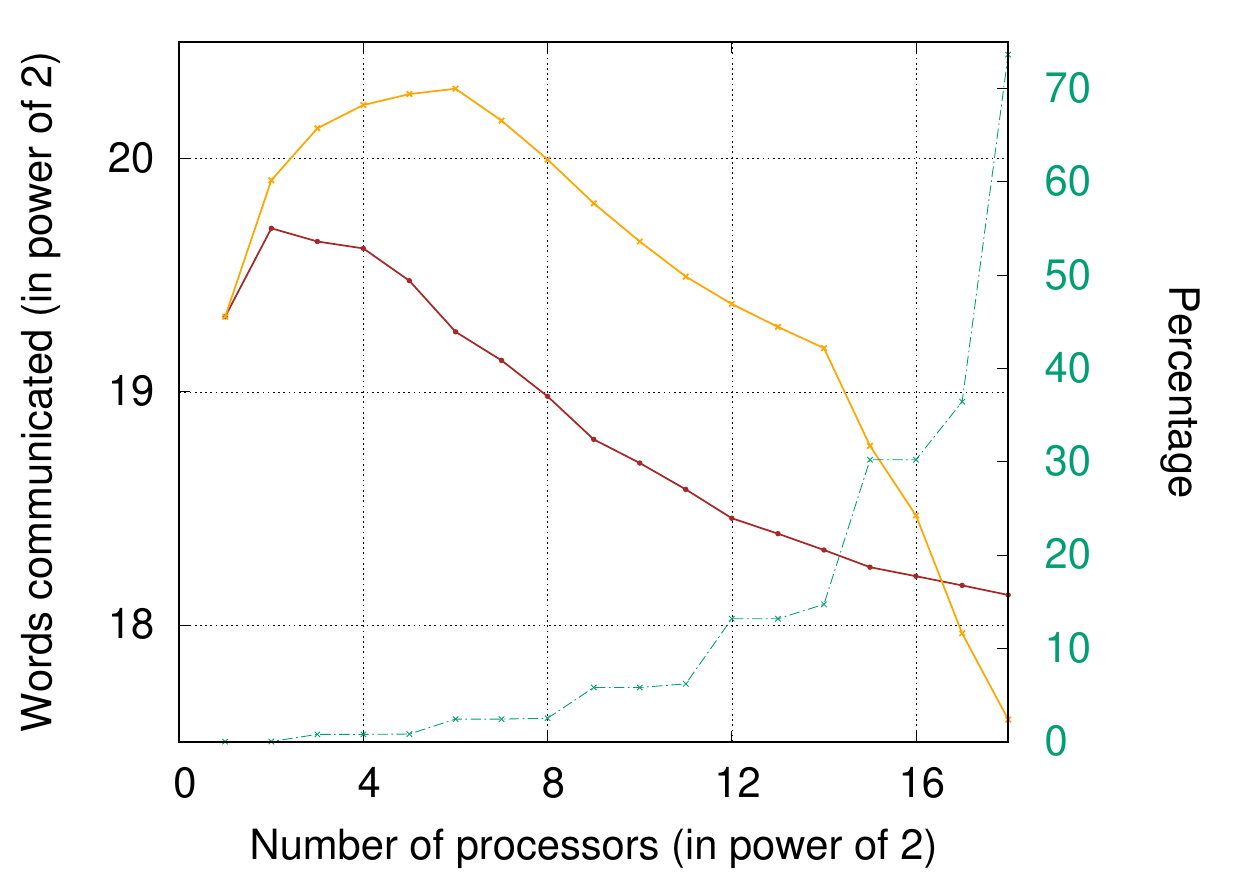}}\hfill
		\subfloat[$n_i=2^{20}$, $r_i=2^{8}$.\label{fig:commcostcomparison:25-10}]{	\includegraphics[width=0.31\linewidth]{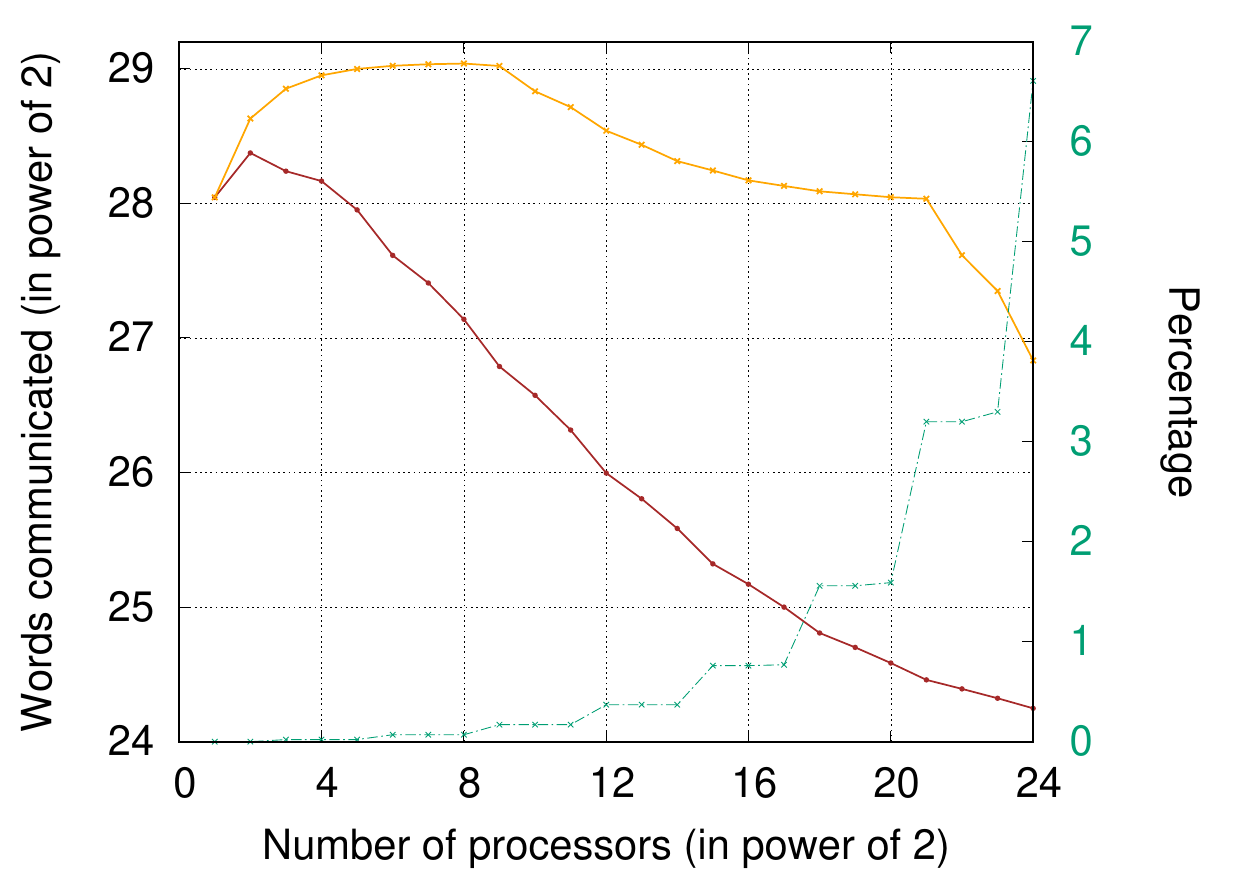}}
		\caption{Communication cost comparison of \cref{alg:3dmultittm} and TTM-in-Sequence \cite{Ballard:TuckerMPI:TOMS20}. $Comp-Overhead$ shows the percentage of computational overhead of \cref{alg:3dmultittm} with respect to the TTM-in-Sequence approach. \label{fig:commcostcomparison}
	}
		\label{fig:comp}
	\end{center}
\end{figure*}

\begin{figure}
	\begin{center}
		\includegraphics[scale=0.56]{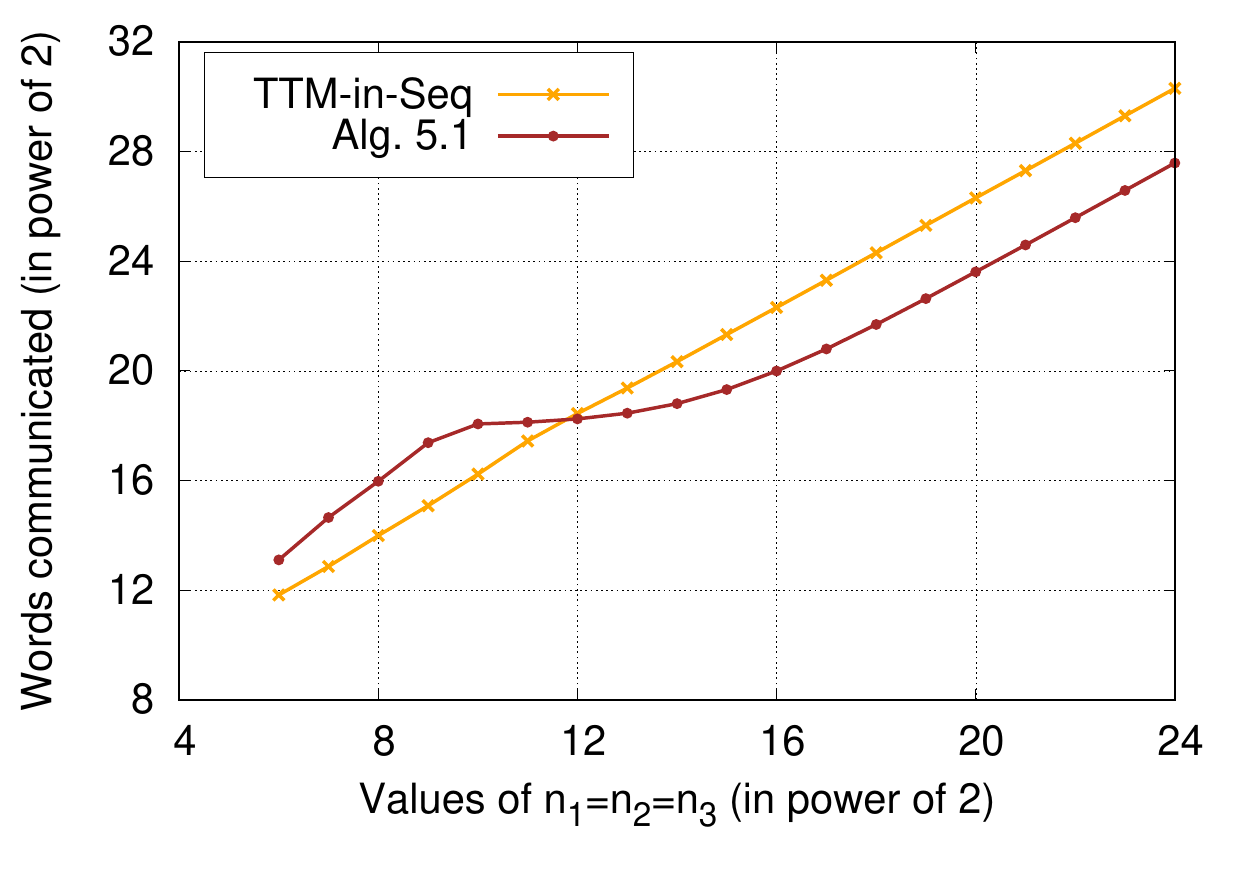}
		\vspace*{-0.215cm}\caption{Comparison of \cref{alg:3dmultittm} and the TTM-in-Sequence approach for fixed $r_1=r_2=r_3=2^6$ and $P=2^{12}$.\label{fig:commcostcomparison:riPfixed}}
	\end{center}
\end{figure}

\subsubsection{Communication Cost}

To compare communications costs, we perform 4 experiments involving cubical tensors.
The first three simulated evaluations consider strong scaling and are presented in \cref{fig:comp}.
Two of these experiments use the same tensor dimensions as the two cubical examples in \cref{fig:lb}.
The first experiment involves an input tensor of dimension $n_i=2^{12}$ and output dimension $r_i=2^4$ (\cref{fig:commcostcomparison:12-4}), the second has dimensions $n_i=2^{13}$ and $r_i=2^6$ (\cref{fig:commcostcomparison:13-6}), and the third has the largest dimensions $n_i=2^{20}$ and $r_i=2^8$ (\cref{fig:commcostcomparison:25-10}).

\Cref{fig:commcostcomparison:12-4} shows that \cref{alg:3dmultittm} performs less communication than TTM-in-Sequence for $P\leq 2^{12}<n/r$.
The largest communication reduction occurs at $P=2^{12}$ and is approximately $5\times$.
In the second experiment, we see cases where TTM-in-Sequence performs less communication than \cref{alg:3dmultittm} and in fact beats the lower bound of \cref{theorem:lb:3DMultiTTM} (which is possible because it breaks the atomicity assumption).
\Cref{alg:3dmultittm} is more communication efficient for $P\leq 2^{16}$, achieving a speedup of up to $2\times$, but communicates more for larger $P$.
In the third experiment with larger tensors, \cref{fig:commcostcomparison:25-10} demonstrates similar qualitative behavior to the first, with \cref{alg:3dmultittm} outperforming TTM-in-Sequence and a maximum communication reduction of approximately $12\times$ at $P=2^{21}$.

In the fourth experiment, with results shown in \cref{fig:commcostcomparison:riPfixed}, we fix the output tensor dimension $r_i=2^6$ and number of processors $P=2^{12}$ and vary the input tensor dimension $n_i$.
We observe that for $2^6\leq n_i< 2^{12}$, the TTM-in-Sequence approach communicates less data than \cref{alg:3dmultittm}.
For $n_i \geq 2^{12}$, \cref{alg:3dmultittm} communicates less data, and the factor of improvement is maintained at approximately $6\times$ as $n_i$ scales up.

\subsubsection{Computation Cost}
\label{sec:computation-cost-comparison}

Assuming TuckerMPI uses increasing mode order, the parallel computational cost is
$$2 \cdot \frac{r_1n_1n_2n_3+r_1r_2n_2n_3+r_1r_2r_3n_3}{P}=2\left(\frac{r^{1/3}n}{P}+\frac{r^{2/3}n^{2/3}}{P}+\frac{rn^{1/3}}{P}\right),$$
where the right hand side is simplified under the assumption of cubical tensors.
In these experiments where $n\gg r$, \cref{alg:3dmultittm} selects a processor grid such that $q=1$ and $p_1\approx p_2\approx p_3$.
In this case the computation cost given in \cref{sec:cost-3d} simplifies to
$$2\left(\frac{r^{1/3}n}{P}+\frac{r^{2/3}n^{2/3}}{P^{2/3}}+\frac{rn^{1/3}}{P^{1/3}}\right).$$
Note that this cost is much smaller than $4nr/P$, the cost of evaluating \cref{eq:ourMultiTTMDef} directly with computational load balance, and it is achieved by performing local computation using a TTM-in-Sequence approach.

While the first terms of the two computational cost expressions match, we observe greater computational cost from \cref{alg:3dmultittm} in the second and third terms.
\response{These terms are lower order when $P\ll n/r$, in which case the extra computational cost of \cref{alg:3dmultittm} is negligible. This is also validated by \cref{fig:commcostcomparison} for the first three experiments.
When $P=n/r$, the extra computational cost is no more than $2\times$.}

In the first three experiments, when our approach reduces communication, the extra computational costs were at most $6\%$, $30\%$, and $7\%$, respectively.
The extra computation required for the greatest reductions in communication in those experiments were $6\%$, $2\%$, and $7\%$.
For the fourth experiment, the extra computation is approximately $13\%$ at $n_i=2^{13}$, where \cref{alg:3dmultittm} provides communication reduction, and decreases as $n_i$ increases. 

In all these experiments, we see that when \cref{alg:3dmultittm} provides a reduction in communication costs, the extra computational costs remain negligible.

\subsection{Details for Evaluation of Our Algorithm}
Here we provide more details for the simulated evaluation of our algorithm and its comparison to the TTM-in-Sequence approach.
The analysis of the communication optimality of \cref{alg:3dmultittm} did not consider integrality constraints on the processor grid dimensions.
The simulated evaluation in the previous subsection considered all possible processor grid configurations using exhaustive search; we explain in \cref{sec:exp:configurations} a more efficient process for determining an optimal grid when $P$ is a power of two.
In the previous subsection, we also compare \cref{alg:3dmultittm} against an implementation of the TTM-in-Sequence approach as implemented by TuckerMPI \cite{Ballard:TuckerMPI:TOMS20}.
We argue in \cref{sec:exp:lbTTM-in-Sequence} that this implementation is nearly communication optimal given the computation that it performs, validating our comparison against it.
\Cref{fig:commcostcomp-bestvfast} presents results relevant to both \cref{sec:exp:configurations,sec:exp:lbTTM-in-Sequence}.

\subsubsection{Obtaining Integral Processor Grids for \cref{alg:3dmultittm}}
\label{sec:exp:configurations}

In order to determine the communication cost of \cref{alg:3dmultittm}, one must determine the processor grid.
Obtaining $p_i$ and $q_i$ from the procedure in \cref{sec:3dUpperBounds} may yield non-integer values.
The following procedure allows us to convert these to integers under our assumption that all parameters are powers of $2$.
Recall that we consider $P=pq$ with $p=p_1p_2p_3$ and $q=q_1q_2q_3$.

If $\lfloor \log_2(p)+0.5\rfloor = \lfloor \log_2(p)\rfloor$, then we set $p=2^{\lfloor\log_2(p)\rfloor}$, otherwise we set $p=2^{\lceil\log_2(p)\rceil}$, distributing the modification evenly between $p_1, p_2,$ and $p_3$.
Now, we keep $p=p_1p_2p_3$ constant, and convert each $p_i$ to an integer.
We set $p_1=2^{\lfloor\log_2(p_1)+0.5\rfloor}$ distributing the changes evenly among $p_2$ and $p_3$.
To see that our new value of $p_1$ must still be smaller than $n_1$, we note that our original $p_1$ was less than $n_1$ which is a power of 2 by our assumption.
If we increased $p$ in our first step, then distributing the modifications evenly between $p_1, p_2$ and $p_3$ increased them by at most $2^{1/6}$.
Thus $p_1\leq n_1$ will imply that $\lfloor\log_2(p_1 \cdot 2^{1/6})+0.5\rfloor\leq \log_2(n_1)$. Note that this most recent modification to $p_1$ changes $p_2$ and $p_3$. Then, we set $p_2=2^{\lfloor\log_2(p_2)+0.5\rfloor}$ and adapt $p_3$ accordingly. 
A similar argument to what is used for $p_1$ will show that $p_2$ and $p_3$ are also not larger than their corresponding dimensions. 
Having completed our work on the processor dimensions associated with the first tensor, we set $q=\frac{P}{p}$ distributing the changes evenly among the $q_i$, then force each $q_i$ to be an integer following the same procedure as for the $p_i$. 

We denote the communication cost of \cref{alg:3dmultittm} for the grid determined using this method by \lbbasedpartition and the communication cost using exhaustive search by \bestconfigAAO. 
We note that this procedure can increase the total number of accessed elements of any variable at most $4$ times, but we see in \cref{fig:commcostcomp-bestvfast} that the communication costs of both procedures are exactly the same for the examples we consider.
These problems match those presented in \cref{fig:commcostcomparison}.

\begin{figure*}
	\begin{center}
		$\quad$\includegraphics[scale=0.235]{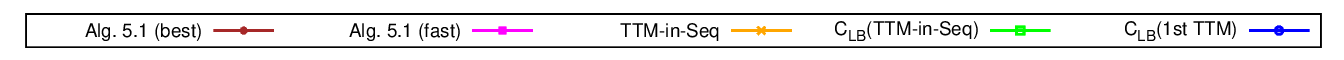}
	\end{center}
	\vspace*{-0.425cm}\begin{center}
		\subfloat[$n_i=2^{12}$, $r_i=2^{4}$.\label{fig:commcostcomp-bestvfast:12-4}]{\includegraphics[width=0.31\linewidth]{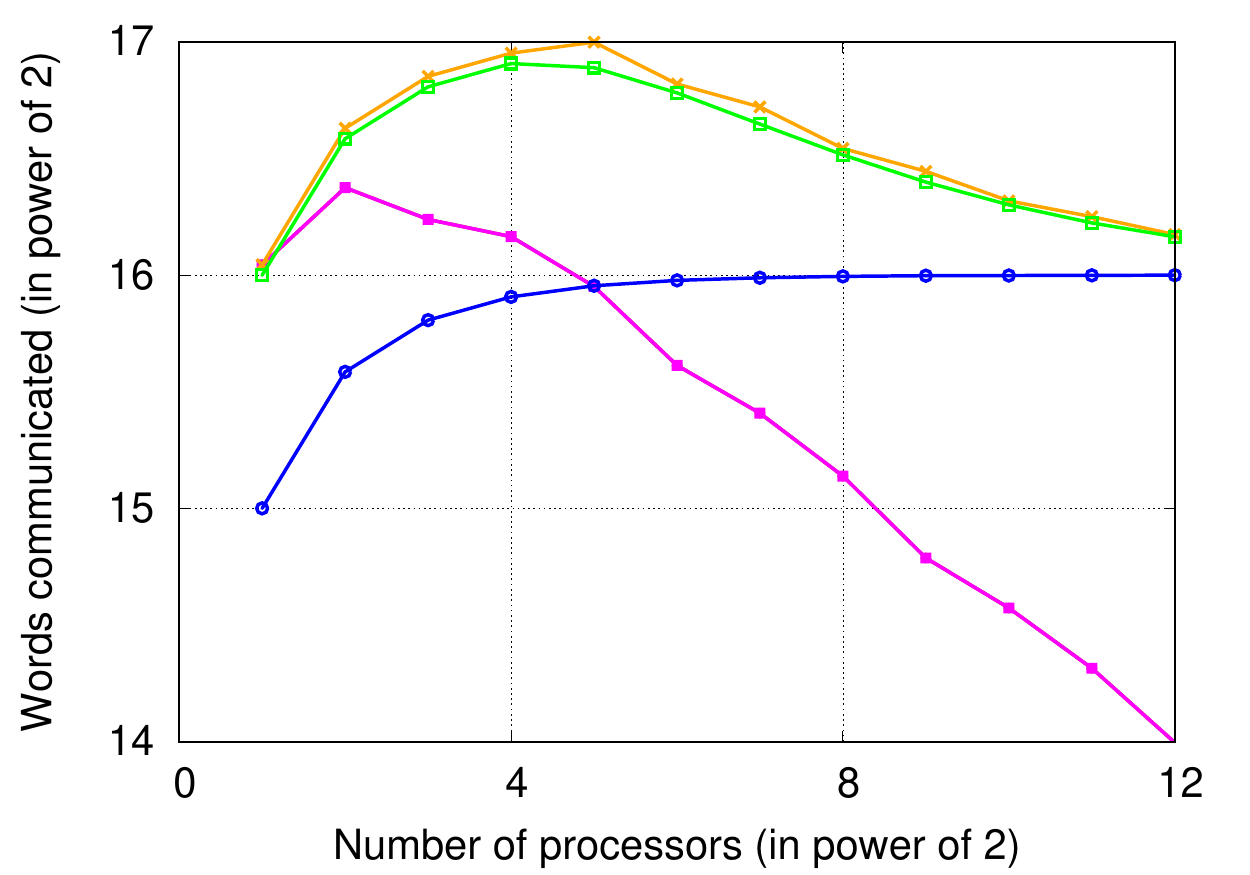}}\hfill
		\subfloat[$n_i=2^{13}$, $r_i=2^{6}$.\label{fig:commcostcomp-bestvfast:13-6}]{\includegraphics[width=0.31\linewidth]{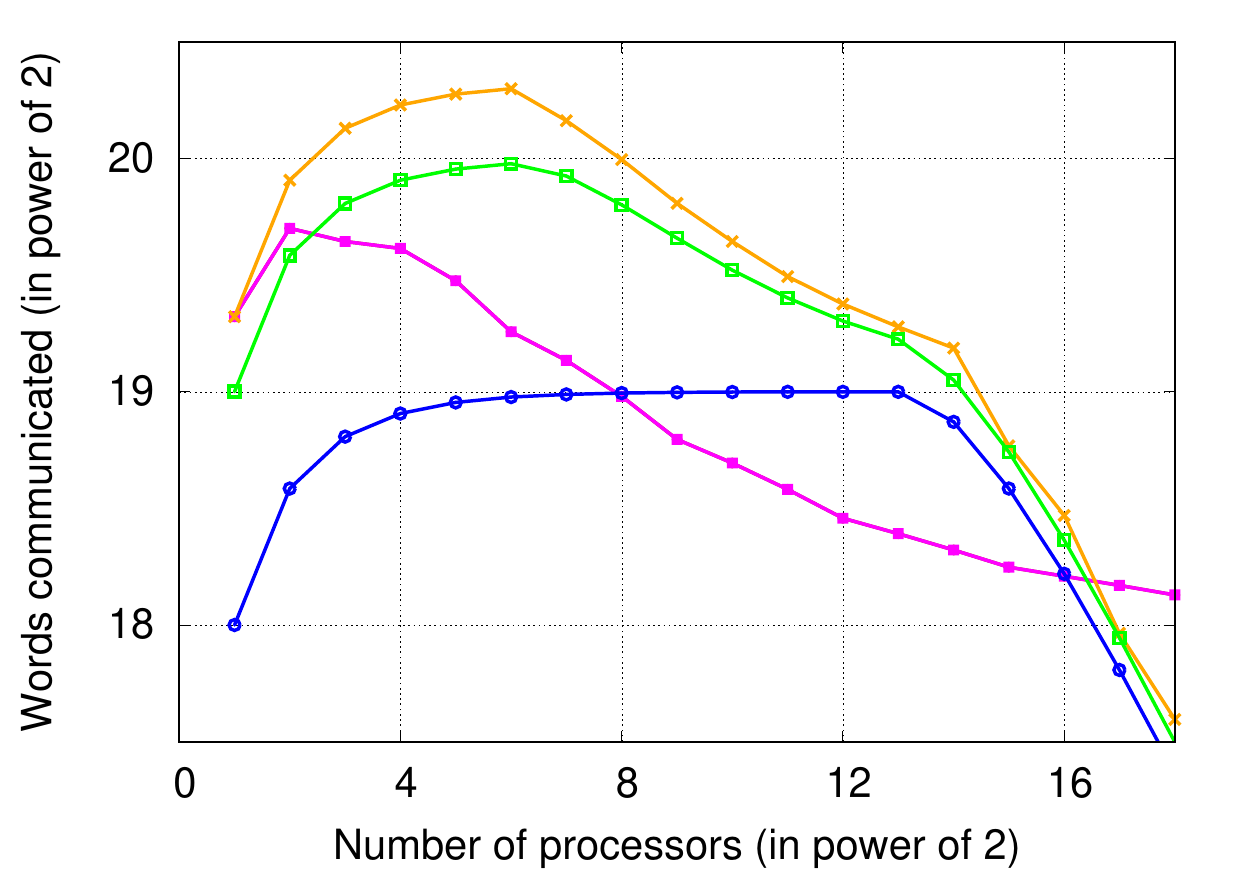}}\hfill
		\subfloat[$n_i=2^{20}$, $r_i=2^{8}$.\label{fig:commcostcomp-bestvfast:20-8}]{	\includegraphics[width=0.31\linewidth]{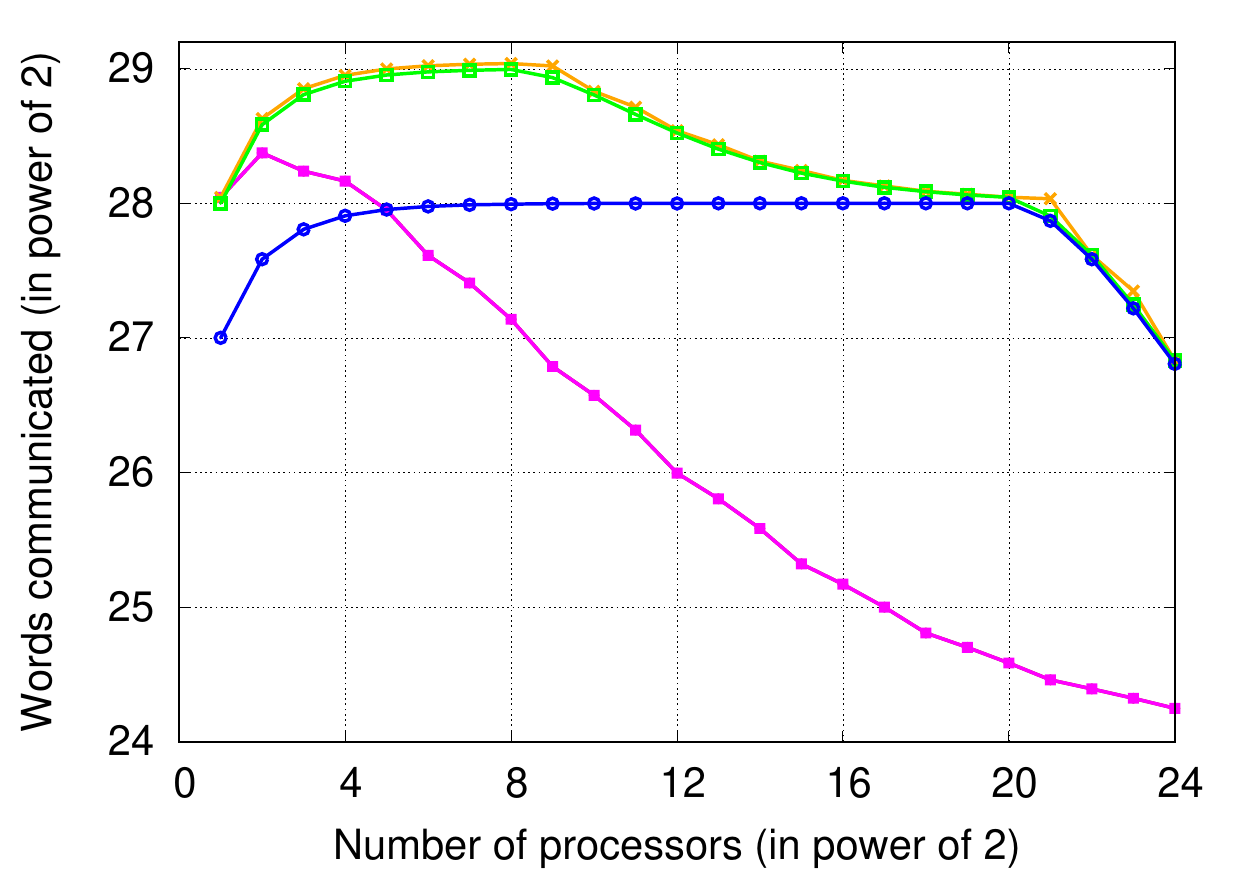}}
		\caption{Communication cost comparison of \cref{alg:3dmultittm} using best processor grid against fast method and of the \emph{TTM-in-Sequence} approach implemented by TuckerMPI against the lower bounds. \lbbasedpartition and \bestconfigAAO are the same for all the configurations.}
		\label{fig:commcostcomp-bestvfast}
	\end{center}
\end{figure*}

\subsubsection{TTM-in-Sequence Lower Bounds}
\label{sec:exp:lbTTM-in-Sequence}

Here we discuss communication low\-er bounds for the TTM-in-Sequence approach with cubical tensors. 
There has not been any proven bound for this approach other than individual bounds for each TTM (a single matrix multiply) computation, assuming the sequence of TTMs has been specified. 
The sum of individual bounds provides a communication lower bound for this approach. 
We obtain the tightest (and obtainable) lower bound for each TTM from~\cite{ABGKR22-report}, which depends on the relative matrix dimensions and number of processors, and represent the sum by \lbSeq. 
We also note that \lbSeq may not be always attainable as data distributions for two successive TTMs may be non-compatible and require extra communication. 
When the input tensor dimensions are much larger than the output tensor dimensions, most of the computation and communication occur in the first TTM, so we also consider the communication lower bound of only that matrix multiplication, which also provides a valid lower bound for the entire TTM-in-Sequence computation.
Recall that we obtain the algorithmic cost of TTM-in-Sequence by exhaustively searching for the best processor grid configuration given the communication costs specified by \cref{eq:TuckerMPI-cost}.
\Cref{fig:commcostcomp-bestvfast} shows a comparison of \bestconfigSeq and \lbSeq for the tensor dimensions presented in \cref{fig:commcostcomparison}. 
We can see that the communication costs of \bestconfigSeq are very close to \lbSeq,
the largest differences are $7.9\%$ for \cref{fig:commcostcomp-bestvfast:12-4} at $P=2^5$, $25\%$ for \cref{fig:commcostcomp-bestvfast:13-6} at $P=2^6$, and $9.3\%$ for \cref{fig:commcostcomp-bestvfast:20-8} at $P=2^{21}$.
Comparing \lbFirstTTM and \lbSeq, we see that for these examples at least half the communication of the entire TTM-in-Sequence is required by the first TTM, and it is completely dominated by the first TTM when $P$ is large.

\section{Lower Bounds of General Multi-TTM}
\label{sec:genLowerBounds}
We present our lower bound results for $d$-dimensional tensors in this section.
Similar to the 3-dimensional lower bound proof, \response{we consider a single processor that performs $1/P$th of the computation and owns at most $1/P$th of the data.
We again seek to minimize the number of elements of the matrices and tensors that the processor must access or partially compute in order to execute its computation subject to the constraints of the structure of Multi-TTM by solving two independent problems, one for the matrix data and one for the tensor data.}

\subsection{General Constrained Optimization Problems}
\label{sec:lb:genOptimizationProblem}

Here we present a generalization of \cref{lemma:matrixOptimalSolutions} for $d$ dimensions. As before, this corollary is a direct result of \cref{lem:mttkrpOpt}.
Recall the notation $N_i=\prod_{j=d-i+1}^dn_j$ and $R_i = \prod_{j=d-i+1}^d r_i$.

\begin{corollary}\label{lemma:genMatrixOptimalSolutions}
	Consider the following optimization problem:
	$$\min_{\V{x}} \sum_{i \in [d]} x_i$$
 such that 
$$\frac{nr}{P} \le \prod_{i\in[d]}x_i  \quad\text{and}\quad 0 \le x_i \le n_ir_i \quad \text{for all}\quad 1\leq i\leq d, $$
where $n_i, r_i,P \geq 1$ and $n_ir_i \le n_{i+1}r_{i+1}$. 
	The optimal solution $\mathbf{x} = [\starontop{x_1}$ $\cdots$ $\starontop{x_d}]$ depends on the values of constants, yielding $d$ cases.
	\smallskip
	\begin{center}
		\begin{tikzpicture}[scale=0.495, every node/.style={transform shape}]
		\draw [thick] (-0.1,0) -- (10.5,0);
		\draw [thick, dashed] (10.5,0) -- (14.5,0);
		\draw [->, thick] (14.5,0) -- (25,0) node [below right,scale=1.6] {$P$};
		
		\draw (0, 0.1) -- node [below, pastelred, scale=1.6]{$1$}(0,-0.1);
		\draw (5, 0.1) -- node [below, pastelred, scale=1.6]{$\frac{N_1R_1}{n_{d-1}r_{d-1}}$}(5,-0.1);
		\draw (10, 0.1) -- node [below, pastelred, scale=1.6] {$\frac{N_2R_2}{(n_{d-2}r_{d-2})^2}$}(10,-0.1);
		\draw (15, 0.1) -- node [below, pastelred, scale=1.6] {$\frac{N_{d-2}R_{d-2}}{(n_2r_2)^{d-2}}$}(15,-0.1);
		\draw (20, 0.1) -- node [below, pastelred, scale=1.6] {$\frac{N_{d-1}R_{d-1}}{(n_1r_1)^{d-1}}$}(20,-0.1);
		
		\node[align=left,below, scale=1.5] at (2.5, -0.4) { $\starontop{x_1}=n_1r_1$\\ $\qquad\vdots$\\ $\starontop{x_{d-1}}= n_{d-1}r_{d-1}$\\$\starontop{x_d} =\frac{N_1R_1}{P}$};
		\node[align=left,below, scale=1.5] at (8.15, -0.75) {$\starontop{x_1}=n_1r_1$\\ $\qquad\vdots$\\ $\starontop{x_{d-2}}= n_{d-2}r_{d-2}$\\$\starontop{x_{d-1}} = \starontop{x_d}$=\\ $\quad\left(\frac{N_2R_2}{P}\right)^{1/2}$};
		\node[align=center,below, scale=1.5] at (16.5, -0.75) {$\starontop{x_1}=n_1r_1$\\$\starontop{x_2}= \cdots= \starontop{x_d}=$\\ $\qquad\quad \big(\frac{N_{d-1}R_{d-1}}{P}\big)^\frac{1}{d-1}$};
		\node[align=center,below, scale=1.5] at (22.5, -1.25) {$\starontop{x_1}= \cdots = \starontop{x_d}$=\\ $\qquad\quad \big(\frac{N_dR_d}{P}\big)^{1/d}$};	
		\end{tikzpicture}
	\end{center}
	
	\begin{itemize}[label=$\bullet$]
		\item If $P <\frac{N_1R_1}{n_{d-1}r_{d-1}}$, then 
		$$\starontop{x_j} = n_jr_j \quad\text{for}\quad 1\le j\le d-1 \quad\text{and}\quad \starontop{x_d}=\frac{N_1R_1}{P}.$$
		\item If $\frac{N_{i-1}R_{i-1}}{(n_{d+1-i}r_{d+1-i})^{i-1}} \leq P < \frac{N_iR_i}{(n_{d-i}r_{d-i})^i}$ for some $i=2,\cdots, d-1$, then
		$$\starontop{x_j} = n_jr_j \quad\text{for}\quad  1\le j\le d-i \quad\text{and}\quad \starontop{x_{d+1-i}}=\cdots=\starontop{x_d} =\left(N_iR_i/P\right)^{1/i}.$$
		\item If $\frac{N_{d-1}R_{d-1}}{(n_1r_1)^{d-1}} \leq P$, then 
		$$\starontop{x_1}=\cdots=\starontop{x_d} = \left(N_dR_d/P\right)^{1/d}.$$
	\end{itemize}
\end{corollary}

\subsection{Communication Lower Bounds}
\label{sec:lb:genLB}
We now present the lower bounds for the general Multi-TTM computation. We prove this by applying \Cref{lemma:genMatrixOptimalSolutions,lemma:tensorOptimalSolutions} and extending the arguments of \Cref{theorem:lb:3DMultiTTM} in a straightforward way (though with more complicated notation).


\begin{theorem}
	\label{theorem:lb:genMultiTTM}
	Any computationally load balanced atomic Multi-TTM algorithm that starts and ends with one copy of the data distributed across processors and involves $d$-dimensional tensors with dimensions $n_1, n_2, \ldots, n_d$ and $r_1, r_2, \ldots, r_d$ performs at least $A+B-\left(\frac{n}{P}+\frac{r}{P} +\sum_{j=1}^d \frac{n_jr_j}{P}\right)$ sends or receives where
	\begin{align*}
	A &= \begin{cases} \sum_{j=1}^{d\text{-}1}n_jr_j + \frac{N_1R_1}{P} &\text{ if } P<\frac{N_1R_1}{n_{d\text{-}1}r_{d\text{-}1}}\text,\\
	\sum_{j=1}^{(d\text{-}i)} n_jr_j + i\left(\frac{N_{i}R_{i}}{P}\right)^\frac{1}{i} &\text{ if } \frac{N_{i\text{-}1}R_{i\text{-}1}}{(n_{d+1\text{-}i}r_{d+1\text{-}i})^{i\text{-}1}} \leq P < \frac{N_{i}R_{i}}{(n_{d\text{-}i}r_{d\text{-}i})^i}\text, \\
	& \hfill \text{for some } 2\leq i\leq d-1, \\
	d\left(\frac{N_dR_d}{P}\right)^\frac{1}{d} &\text{ if } \frac{N_{d\text{-}1}R_{d\text{-}1}}{(n_1r_1)^{d\text{-}1}} \leq P\text.
	\end{cases}\\
	B &= \begin{cases} r + \frac{n}{P} & \text{ if } P < \frac{n}{r}\text, \\ 2\left(\frac{nr}{P}\right)^{\frac{1}{2}} &\text{ if } \frac{n}{r} \leq P\text.\end{cases}
	\end{align*}
\end{theorem}
\begin{proof}
	Let $F$ be the set of loop indices associated with the \mult{s} performed by a processor.
	As we assumed the algorithm is computationally load balanced, $|F| = nr/P$.
	\response{We define $\phi_{\X}(F)$, $\phi_{\Y}(F)$ and $\phi_j(F)$ to be the projections of $F$ onto the indices of the arrays $\X, \Y$, and $\Mn{A}{j}$ for $1\leq j\leq d$ which correspond to the elements of the arrays that must be accessed or partially computed by the processor.}
	
	\response{We use \Cref{lem:hbl} to obtain a lower bound on the number of array elements that must be accessed or partially computed by the processor.}
	The matrix corresponding to the projections above is given by
	$$\M{\Delta} = \begin{bmatrix}\M{I}_{d\times d} & \V{1}_d & \V{0}_d\\ \M{I}_{d\times d} & \V{0}_d & \V{1}_d \end{bmatrix}\text. $$
	Here $\V{1}_d$ and $\V{0}_d$ denote the $d$-dimensional vectors of all ones and zeros, respectively, and $\M{I}_{d\times d}$ denotes the $d\times d$ identity matrix. As before we define
	$$\mathcal{C} =  \big\{\V{s}=\response{[s_1\ \cdots \ s_{d+2}]^\Tra: 0\leq s_i \leq 1 \text{ for } i=1,2,\cdots,d+2 \text{ and }} \M{\Delta}\cdot\V{s}\ge\V{1}\big\}\text.$$
	
	We recall that $\V{1}$ represents a vector of all ones. As in the proof of \cref{theorem:lb:3DMultiTTM}, $\M{\Delta}$ is not full rank, so we again consider each vector $\V{v} \in \mathcal{C}$ such that $\M{\Delta}\cdot\V{v}=\V{1}$.
	Such a vector $\V{v}$ is of the form $\begin{bmatrix} a & \cdots & a & 1-a & 1-a \end{bmatrix}$ where $0\le a\le 1$.
	Thus, we obtain 
	$$\frac{nr}{P} \leq \Big(\prod_{j\in[d]}|\phi_j(F)|\Big)^a \big(|\phi_{\X}(F)||\phi_{\Y}(F)|\big)^{1\text{-}a}\text.$$
	
	Similar to the 3D case, the above constraint is equivalent to $\frac{nr}{P} \leq \prod_{j\in[d]}|\phi_j(F)|$ and $\frac{nr}{P} \leq |\phi_{\X}(F)||\phi_{\Y}(F)|$.
	
	Clearly a projection onto an array can not be larger than the array itself, thus $|\phi_{\X}(F)| \leq n$, $|\phi_{\Y}(F)|\leq r$, and $|\phi_j(F)|\leq n_jr_j$ for $1\leq j \leq d$.
	
	\response{As the constraints related to the projections of matrices and tensors are disjoint, we solve them separately and then sum the results to get a lower bound on the number of elements that must be accessed or partially computed by the processor.
	We obtain a lower bound on $A$, the number of relevant elements of the matrices by using \Cref{lemma:genMatrixOptimalSolutions}, and a lower bound on $B$, the number of relevant elements of the tensors by using \Cref{lemma:tensorOptimalSolutions}.
	By summing both, we get the positive terms of the lower bound.}
	
	
	To bound the sends or receives, we consider how much data the processor could have had at the beginning or at the end of the computation.
	\response{Assuming there is exactly one copy of the data at the beginning and at the end of the computation, there must exist a processor which owns at most $1/P$ of the elements of the arrays at the beginning or at the end of the computation.
	By employing the previous analysis, this processor must access or partially compute $A+B$ elements of the arrays, but can only own $\frac{n}{P}+\frac{r}{P} +\sum_{j\in[d]} \frac{n_jr_j}{P}$ elements of the arrays.
	Thus it must perform the specified amount of sends or receives.}
\end{proof}


\section{Parallel Algorithm for General Multi-TTM}
\label{sec:genUpperBounds}
We present a parallel algorithm to compute $d$-dimensional Multi-TTM in \cref{alg:genMultittm}, which is analogous to
\cref{alg:3dmultittm}. We organize $P$ processors into a $2d$-dimensional logical processor grid with dimensions $p_1 \times \cdots\times p_d \times q_1 \times \cdots \times q_d$. As before, we consider that $\forall i \in [d]$, $p_i$ and $q_i$ evenly divide $n_i$ and $r_i$, respectively. A processor coordinate is represented as $(p_1^\prime, \cdots, p_d^\prime, q_1^\prime,\cdots, q_d^\prime)$, where $\forall i \in [d]$, $1\le p_{i}^\prime \le p_i$ and $1\le q_{i}^\prime \le q_i$. 



\begin{algorithm}[H]
	\caption{\label{alg:genMultittm}Parallel Atomic d-dimensional Multi-TTM}
	\begin{algorithmic}[1]
		\REQUIRE $\T{X}$, $\Mn{A}{1}$, $\cdots$, $\Mn{A}{d}$, $p_1 \times \cdots \times p_d \times q_1 \times \cdots \times q_d$ logical processor grid
		\ENSURE $\T{Y}$ such that $\Y = \X \times_1 {\Mn{A}{1}}^\Tra \cdots \times_d {\Mn{A}{d}}^\Tra$
		\STATE $(p_1^\prime, \cdots , p_d^\prime, q_1^\prime, \cdots, q_d^\prime)$ is my processor id
		\STATE //All-gather input tensor $\T{X}$
		\STATE $\T{X}_{p_1^\prime \cdots p_d^\prime}$ = All-Gather($\T{X}$, $(p_1^\prime, \cdots , p_d^\prime, *, \cdots, *)$)\label{alg:genMultittm:line:allGatherInputTensor}
		\STATE //All-gather all input matrices
		\FOR{$i=1,\cdots, d$} 
		\STATE $\Mn{A}{i}_{p_i^\prime q_i^\prime}$ = All-Gather($\Mn{A}{i}$, $(*,\cdots,*, p_i^\prime,* \cdots,*, q_i^\prime, *)$)\label{alg:genMultittm:line:allGatherMatrixi} 
		\ENDFOR
		\STATE //Perform local computations in a temporary tensor $\T{T}$
		\STATE $\T{T}$ = Local-Multi-TTM($\T{X}_{p_1^\prime \cdots p_d^\prime}$, $\Mn{A}{1}_{p_1^\prime q_1^\prime}$,$\cdots$, $\Mn{A}{d}_{p_d^\prime q_d^\prime}$)\label{alg:genMultittm:line:localcomputation}
		\STATE //Reduce-scatter the output tensor in $\T{Y}_{q_1^\prime \cdots q_d^\prime}$
		\STATE Reduce-Scatter($\T{Y}_{q_1^\prime \cdots q_d^\prime}$, $\T{T}$, $(*, \cdots , *, q_1^\prime, \cdots, q_d^\prime)$)\label{alg:genMultittm:line:reduceScatterOutputTensor}
	\end{algorithmic}
\end{algorithm}

Here we discuss our data distribution model for \cref{alg:genMultittm}, which is similar to that of \cref{alg:3dmultittm}.
$\T{X}_{p_1^\prime \cdots p_d^\prime}$ and $\T{Y}_{q_1^\prime \cdots q_d^\prime}$ denote the subtensors of $\T{X}$ and $\T{Y}$ owned by processors $(p_1^\prime,\cdots, p_d^\prime, *,\cdots, *)$ and $(*, \cdots, *, q_1^\prime,\cdots, q_d^\prime)$, respectively. $\Mn{A}{i}_{p_i^\prime q_i^\prime}$ denotes the submatrix of $\Mn{A}{i}$ owned by processors $(*,\cdots,*, p_i^\prime,*,\cdots, *, q_i^\prime, *,\cdots, *)$.
We impose that there is one copy of data in the system at the beginning
and the end of the computation,
and each subarray is distributed evenly among the set of processors which own the data.

When \cref{alg:genMultittm} completes, $\T{Y}_{q_1^\prime \cdots q_d^\prime}$ is distributed evenly among processors $(*, \cdots, *,\linebreak q_1^\prime, \cdots, q_d^\prime)$. 
We recall that $\prod_{i=1}^dp_i$ and $\prod_{i=1}^dq_i$ are denoted by $p$ and $q$, respectively.

\subsection{Cost Analysis}
\label{app:sec:genUpperBounds:costAnalysis}
Now we analyze computation and communication costs of the algorithm.
As before, the local Multi-TTM computation in Line~\ref{alg:genMultittm:line:localcomputation} can be performed as a sequence of TTM operations to mininimize the number of arithmetic operations. Assuming the TTM operations are performed in their order, first with $\Mn{A}{1}$, then with $\Mn{A}{2}$, and so on until the last is performed with $\Mn{A}{d}$, then each processor performs $\sum_{k=1}^d \left(2\prod_{i=1}^k \frac{r_i}{q_i}\prod_{j=k}^d\frac{n_j}{p_j}\right)$ operations.
In Line~\ref{alg:genMultittm:line:reduceScatterOutputTensor}, each processor also performs $(1-\frac{q}{P}) \frac{r}{q}$ computations due to the Reduce-Scatter operation.

Communication occurs only in All-Gather and Reduce-Scatter collectives in \linebreak Lines~\ref{alg:genMultittm:line:allGatherInputTensor}, \ref{alg:genMultittm:line:allGatherMatrixi}, and \ref{alg:genMultittm:line:reduceScatterOutputTensor}.
\response{Line~\ref{alg:genMultittm:line:allGatherInputTensor} specifies $p$ All-Gathers over disjoint sets of $\frac{P}{p}$ processors, Line~\ref{alg:genMultittm:line:allGatherMatrixi} specifies ${p_iq_i}$ All-Gathers over disjoint sets of $\frac{P}{p_iq_i}$ processors in the $i$th loop iteration, and Line~\ref{alg:genMultittm:line:reduceScatterOutputTensor} specifies $q$ Reduce-Scatters over disjoint sets of $\frac{P}{q}$ processors.}
Each processor is involved in one All-Gather involving the input tensor, $d$ All-Gathers involving input matrices and one Reduce-Scatter involving the output tensor.

As before, we assume bandwidth and latency optimal algorithms are used for the All-Gather and Reduce-Scatter collectives.
Hence the bandwidth costs of the All-Gather operations are $(1-\frac{p}{P}) \frac{n}{p}$ for Line~\ref{alg:3dmultittm:line:allGatherInputTensor}, and $\sum_{i=1}^d(1-\frac{p_iq_i}{P}) \frac{n_ir_i}{p_iq_i}$ for the $d$ iterations of Line~\ref{alg:genMultittm:line:allGatherMatrixi}.
The bandwidth cost of the Reduce-Scatter operation in Line~\ref{alg:genMultittm:line:reduceScatterOutputTensor} is $(1-\frac{q}{P}) \frac{r}{q}$.
Hence the overall bandwidth cost of \cref{alg:genMultittm} along the critical path is $\frac{n}{p} + \frac{r}{q} + \sum_{i=1}^d\frac{n_ir_i}{p_iq_i} - \left(\frac{n+r+\sum_{i=1}^d n_ir_i}{P}\right)$. The latency costs are $\log_2\left(\frac{P}{p}\right)$ and $\log_2\left(\frac{P}{q}\right)$ for Lines~\ref{alg:3dmultittm:line:allGatherInputTensor} and \ref{alg:genMultittm:line:reduceScatterOutputTensor} respectively, and $\sum_{i=1}^d \log_2\left(\frac{P}{p_iq_i}\right)$ for the $d$ iterations of Line~\ref{alg:genMultittm:line:allGatherMatrixi}. Thus the overall latency cost of \cref{alg:genMultittm} along the critical path is $\log_2\left(\frac{P}{p}\right)+ \sum_{i=1}^d \log_2\left(\frac{P}{p_iq_i}\right) + \log_2\left(\frac{P}{q}\right) = d\log_2(P).$


 

	


We can prove the following theorem by extending the arguments of \Cref{theorem:optimality:3DMultiTTMAlgorithm}.

\begin{theorem}\label{theorem:optimality:genMultiTTMAlgorithm}
	There exist $p_i,q_i$ with $1\le p_i\le n_i, 1\le q_i\le r_i$ for $i=1,\cdots,d$ such that \cref{alg:genMultittm} is communication optimal to within a constant factor.
\end{theorem}

\begin{proof}
	As we did previously, we break our analysis into 2 scenarios which are further broken down into all possible cases.
	
	In each case, we obtain $\init{p_j}$ and $\init{q_j}$ such that the terms in the communication cost match the corresponding lower bound terms and satisfy at least one of the two sets of constraints, $1\le \init{p_j} \le n_j$, $1\le \init{q_j}$ or $1\le \init{q_j}\le r_j$, $1\le \init{p_j}$ for $1\le j \le d$.
	We handle all cases of both scenarios together in the end, and adapt these values to get $p_j$ and $q_j$ which respect both lower and upper bounds for all values of $j$. Then we determine how much additional communication may be required. We denote $\prod_{i=1}^d\init{p_i}$ and $\prod_{i=1}^d\init{q_i}$ by $\init{p}$ and $\init{q}$. 
	
	\noindent$\bullet$ \underline{Scenario I} $\left(P < \frac{n}{r}\right)$:
	This scenario corresponds to the first case of the tensor term in $\lowerbound$.
	Thus, we set $\init{p_j}, \init{q_j}$ in such a way that the tensor terms in the communication cost match the tensor terms of $\lowerbound$:
	\begin{equation}\label{eq:genS1} \init{p}=P, \init{q} = 1.\end{equation}
	This implies $\init{q_j} = 1$ for $1\le j\le d$.
	We break this scenario into $d$ cases parameterized by $I$: $\frac{N_{I-1}R_{I-1}}{(n_{d-I+1}r_{d-I+1})^{I-1}} \le P < \min\big\{\frac{N_IR_I}{(n_{d-I}r_{d-I})^I}, \frac{n}{r}\big\}$.
	The cases degenerate to $P < \min\big\{\frac{N_1R_1}{n_{d-1}r_{d-1}}, \frac{n}{r}\big\}$, when $I=1$, and $\frac{N_{d-1}R_{d-1}}{(n_1r_1)^{d-1}} \le P < \frac{n}{r}$ when $I=d$.
	
	\noindent Setting the matrix communication costs to the matrix terms of the lower bound in the corresponding cases yields
	\begin{equation}\label{eq:genS1Ci}\frac{n_jr_j}{\init{p_j}\init{q_j}} = n_jr_j\text{ if } 1\le j\le d-I,\qquad \frac{n_jr_j}{\init{p_j}\init{q_j}} = \left(\frac{N_IR_I}{P}\right)^\frac{1}{I}\text{ if }d-I < j \le d.\end{equation}
	Thus, $\init{q_j}=1$ for all $1\le j\le d$, $\init{p_j} = 1$ if $1\le j\le d-I$ and $\init{p_j} = n_jr_j\big(\frac{P}{N_IR_I}\big)^\frac{1}{I}$ if $d-I < j\le d$ to satisfy \cref{eq:genS1,eq:genS1Ci}.
	Note that when $I=1$, $p_d=P\ge 1$, and for the other values of $I$, $p_j \ge 1$ because $n_1r_1 \le \cdots \le n_dr_d$ and $\frac{N_{I-1}R_{I-1}}{(n_{d-I+1}r_{d-I+1})^{I-1}} \le P$.
	Additionally we have that $1= \init{q_j} < r_j$ for $1\le j\le d$.
	However, we are not able to ensure $\init{p_j}\le n_j$ when $d-I < j\le d$. We will handle all cases of both scenarios together as they require the same analysis.
	
	\medskip
	\noindent $\bullet$ \underline{Scenario II} $\left(\frac{n}{r}\le P\right)$:
	This scenario corresponds to the second case of the tensor term in $\lowerbound$.
	Thus, we set $\init{p_i}, \init{q_i}$ in such a way that
	\begin{equation}\label{eq:genS2}\frac{n}{\init{p}}=\frac{r}{\init{q}}=\left(\frac{nr}{P}\right)^{1/2}.\end{equation}
	Again, we break this scenario into $d$ cases parameterized by $I$: $$\max\left\{\frac{N_{I-1}R_{I-1}}{(n_{d-I+1}r_{d-I+1})^{I-1}},\frac{n}{r}\right\} \le P < \frac{N_IR_I}{(n_{d-I}r_{d-I})^I}\;\cdot$$
	The cases degenerate to $P < \frac{N_1R_1}{n_{d-1}r_{d-1}}$, when $I=1$, and $\max\left\{\frac{N_{d-1}R_{d-1}}{(n_1r_1)^{d-1}}, \frac{n}{r}\right\} \le P$ when $I=d$.
	
	\noindent Setting the matrix communication costs to match the corresponding matrix terms in the lower bound yields
	\begin{equation}\label{eq:genS2Ci}\frac{n_jr_j}{\init{p_j}\init{q_j}} = n_jr_j\text{ if }1\le j\le d-I, \qquad \frac{n_jr_j}{\init{p_j}\init{q_j}} = \left(\frac{N_IR_I}{P}\right)^\frac{1}{I}\text{ if }d-I < j \le d.\end{equation}
	Thus we set $\init{q_j} = \init{p_j} = 1$ for all $1\le j\le d-I$.
	When $2\le I \le d$, the equations above do not uniquely determine $\init{p_j},\init{q_j}$ for $d-I<j\le d$.
	However, setting $\init{p_j} = n_j\left(\frac{nP}{rN_I^2}\right)^{1/2I}$ and $\init{q_j}=r_j\left(\frac{rP}{nR_I^2}\right)^{1/2I}$ for $d-I < j\le d$ satisfies equations \ref{eq:genS2}, \ref{eq:genS2Ci} in all cases. Note that we cannot ensure lower and upper bounds on $\init{p_j}$ and $\init{q_j}$. We now look for new solutions to the equations twice. First, we ensure that all lower bounds are respected, i.e., $1\le\init{p_j}$ and $1\le\init{q_j}$, and then we guarantee that all upper bounds of $\init{p_j}$ or $\init{q_j}$ are satisfied, i.e., $\init{p_j} \le n_j$ or $\init{q_j} \le r_j$.

	As we compared communication cost of each term with its corresponding lower bound to obtain $\init{p_j}$ and $\init{q_j}$, we have $1\le \init{p_j}\init{q_j} \le n_jr_j$ for $1\le j \le d$, $1\le \init{p}\le n$ and $1\le \init{q}\le r$ in all $d$ cases. However, we may not have $1\le \init{p_j}$ or $1\le \init{q_j}$ for some $j$. We now seek new solutions that are all greater than $1$. First we will increase all $\init{q_j}, \init{p_j}$ that are less than $1$ in a way that preserves products $\init{p_j}\init{q_j}$ but does not preserve $\init{p}$ and $\init{q}$. Then we will adjust $\init{p_j}$ and $\init{q_j}$ to force the products $\init{p}$ and $\init{q}$ back to their initial values.
	
	Let $q^b$ denote the product of all $\init{q_j}$ such that $\init{q_j} < 1$, and $p^b$ denote the product of all $\init{p_j}$ such that $\init{p_j} < 1$. Without loss of generality, if $q^b \leq p^b$, set $a=1$ $\init{p}^{orig} = \init{p}$, and $\init{q}^{orig}=\init{q}$. We perform the following updates:\\
	Looping over the index $j$ from 1 to $d$, if $\init{q_j} < 1$ then set $a=a\cdot \init{q_j}, \init{p_j}=\init{p_j}\init{q_j}, \init{q_j}=1$; else if $\init{p_j} < 1$ then set $a=a/\init{p_j}, \init{q_j}=\init{p_j}\init{q_j}, \init{p_j}=1.$
	This step preserves all products $\init{p_j}\init{q_j}$ and enforces $1\le\init{p_j}, 1\le\init{q_j}$ for $1\le j\le d$, but it does not preserve $\init{p}, \init{q}$.
	At the end of this step, we have $a = q^b/p^b < 1$, $\init{p} = a\cdot\init{p}^{orig}$, and $\init{q}=\init{q}^{orig}/a$.
	In order to force $\init{p}$ and $\init{q}$ to match their initial values, we decrease some $\init{q_j}$ in such a way that $\init{q}$ is decreased by a factor of $a$. This is possible because $1\le\init{q}^{orig} =a\cdot\init{q}$.
	Looping over the index $j$ from 1 to $d$, if $\init{q_j} > 1$ then set $\init{q_j}^{prev}= \init{q_j}, \init{q_j} = \max(1, a\cdot\init{q_j}), a = a\left(\frac{\init{q_j}^{prev}}{\init{q_j}}\right), \init{p_j} = \init{p_j}\left(\frac{\init{q_j}^{prev}}{\init{q_j}}\right)$.
	At the end of this step, $\init{q}$ has been decreased by a factor of $q^b/p^b$, $\init{p_j}\init{q_j}$ were all preserved, and thus, $\init{p}$ has been increased by a factor of $p^b/q^b$, hence $\init{q}=\init{q}^{orig}$ and $\init{p}=\init{p}^{orig}$.
	After the above updates, we have $1\le \init{p_j}, 1\le \init{q_j}$ for $1\le j \le d$, and the products match the initial products thus are valid solutions to the original equations.
	If $p_b < q_b$ we would perform the same process, but changing the actions on the $\init{q_j}$ to be performed on the $\init{p_j}$ and vice versa.
	
	
	We now handle upper bounds of $\init{p_j}$ and $\init{q_j}$. If $\exists j,k$ such that $\init{p_j} > n_j$ and $\init{q_k}>r_k$, we again seek new solutions such that all $\init{p_j}$ or all $\init{q_j}$ satisfy upper bounds while respecting lower bounds of all variables. Let $p^t$ and $q^t$ denote the products of all $\init{p_j}$ and $\init{q_j}$, respectively, such that $\init{p_j}\init{q_j} \ne 1$. As $\forall j, 1\le \init{p_j}\init{q_j} \le n_jr_j$, therefore $p^t$ is not more than the product of the corresponding $n_j$ and/or $q^t$ is not more than the product of the corresponding $r_j$. If the first constraint is satisfied, then we perform the following updates:\\
	Looping over the index $j$ from 1 to $d$, if $p^t >1$ and $\init{p_j} \init{q_j} \ne 1$ then set $a=\init{p_j} \init{q_j}, \init{p_j}= \min(p^t, n_j), \init{q_j}=\frac{a}{\init{p_j}}, p^t = \frac{p^t}{\init{p_j}}$. After the above updates, we have $1\le \init{p_j} \le n_j, 1\le \init{q_j}$ for $1\le j \le d$, and $\init{p_j}\init{q_j}$, $\init{p}$ and $\init{q}$ are back to their original values. If the first constraint is not satisfied, we would perform the same process on $\init{q_j}$ instead of $\init{p_j}$.
	
	
	\medskip
	Now for all cases of both scenarios, we know that $1\le \init{p}_j$ and $1\le \init{q_j}$ for $1\le j\le d$, and either $\init{p}_j\le n_j$ for $1\le j\le d$ or $\init{q_j}\le r_j$ for $1\le j\le d$.
	It remains to adapt $\init{p}_j$ and $\init{q}_j$ such that both $\init{p}_j \le n_j$ and $\init{q}_j\le r_j$ for $1\le j\le d$.
	We obtain $p_1,\ldots, p_d, q_1,\ldots,q_d$ from $\init{p_j}$ and $\init{q_j}$ such that $p_1\cdots p_d=\init{p}$ and $q_1\cdots q_d=\init{q}$.
	The intuition is to maintain the tensor communication terms in the lower bound.
	
	Initially, we set $p_j = \init{p}_j$ and $q_j = \init{q}_j$ for $1\le j \le d$.
	If $1\le p_j\le n_j$ and $1\le q_j\le r_j$ for $1\le j\le d$, then $\sum_{j \in [d]}\frac{n_jr_j}{p_jq_j} = \sum_{j \in [d]}\frac{n_jr_j}{\init{p_j}\init{q_j}}$ and the communication cost exactly matches the lower bound. Otherwise, we adapt the values and determine the effects on the matrix communication costs. We recall that due of our particular selections of $\init{p_j}$ and $\init{q_j}$, $\nexists j,\ell\in [d]$ such that $\init{p_j} > n_j$ and $\init{q_\ell} > r_\ell$.
	If $\init{p}_j > n_j$ for some $j\in[d]$, then we iterate over the index $j$ from $d$ to 1 setting $p_j=\min\left\{n_j, \frac{\init{p}}{\prod_{\ell\in[d]-\{j\}}p_\ell}\right\}$. We iterate again from $d$ to 1 with the same expression. Iterating twice ensures that all updates are visible to all $p_j$.

	Now we assess how much additional communication is required for the matrices.
	As $\init{p_j} > n_j$ for some $j$, it must be the case that $\init{p} \ge 2$.
	Thus
	\begin{align*}
	\sum_{j \in [d]}\frac{n_jr_j}{p_jq_j} &\le \sum_{j \in [d]}\max\left\{\frac{n_jr_j}{\init{p_j}\init{q_j}}, \frac{r_j}{\init{q_j}}\right\}\\
	& = \sum_{j \in [d]}\Big(\frac{n_jr_j}{\init{p_j}\init{q_j}} + \frac{r_j}{\init{q_j}} - \min\big\{\frac{n_jr_j}{\init{p_j}\init{q_j}},\frac{r_j}{\init{q_j}}\big\}\Big)\\
	&< \sum_{j \in [d]}\left(\frac{n_jr_j}{\init{p_j}\init{q_j}} + \frac{r_j}{\init{q_j}}\right) -(d-1)\\
	&\le \sum_{j \in [d]} \frac{n_jr_j}{\init{p_j}\init{q_j}} + \frac{r}{\init{q}}\\
	& <  \sum_{j \in [d]} \frac{n_jr_j}{\init{p_j}\init{q_j}} + 2 \left(\frac{r}{\init{q}} - \frac{r}{\init{p}\init{q}}\right)\\
	&= \sum_{j \in [d]} \frac{n_jr_j}{\init{p_j}\init{q_j}} + 2 \left(\frac{r}{\init{q}} - \frac{r}{P}\right)\text.
	\end{align*}

	Similarly, if $\exists j\in[d]$ such that $\init{q_j} > r_j$, the same update can be performed to the $q_j$, and we obtain $\sum_{j \in [d]}\frac{n_jr_j}{p_jq_j} < \sum_{j \in [d]} \frac{n_jr_j}{\init{p_j}\init{q_j}} + 2 \left(\frac{n}{\init{p}} - \frac{n}{P}\right)$.

	\smallskip
	\noindent Therefore, $\sum_{j \in [d]}\frac{n_jr_j}{p_jq_j} + \frac{r}{q} + \frac{n}{p} -\odata \le 3\left(\sum_{j \in [d]}\frac{n_jr_j}{\init{p_j}\init{q_j}} + \frac{r}{\init{q}} + \frac{n}{\init{p}}-\odata\right) = 3\lowerbound$.
\end{proof}
\subsection{Simulated Evaluation}
\label{app:sec:generalMultiTTM:evaluation}
Similar to \cref{sec:experiments}, we compare communication costs of our algorithm and a TTM-in-Sequence approach implemented in the TuckerMPI library. We again restrict to cases where all dimensions are powers of $2$, and vary the number of processors $P$ from $2$ to $\maxp$ in multiples of $2$, where $\maxp=\min\{n_1r_1, \cdots, n_dr_d, n, r\}$.

Like \cref{sec:experiments}, we look at all possible processor grid dimensions and represent the minimum communication costs of our algorithm and TuckerMPI algorithm by \bestconfigAAOGen and \bestconfigSeq, respectively.
The TTM-in-Sequence approach described in~\cite{Ballard:TuckerMPI:TOMS20} organizes $P$ in a $d$-dimensional $\tilde{p_1}\times \cdots\times\tilde{p_d}$ logical processor grid. Assuming TTMs are performed in increasing mode order, the overall communication cost of this algorithm is 
{\small\begin{align*}
	\label{eq:TuckerMPIgen-cost}
	\frac{r_1n_2\cdots n_d}{\frac{P}{\tilde{p_1}}} + \frac{r_1r_2n_3\cdots n_d}{\frac{P}{\tilde{p_2}}} +\cdots + \frac{r_1r_2\cdots r_d}{\frac{P}{\tilde{p_d}}} -\frac{r_1n_2\cdots n_d + r_1r_2n_3\cdots n_d + \cdots + r_1r_2\cdots r_d}{P}  \\ \qquad\qquad 
	+\frac{n_1r_1}{\tilde{p_1}} +\cdots +\frac{n_dr_d}{\tilde{p_d}} - \frac{n_1r_1+\cdots+n_dr_d}{P}. \notag
	\end{align*}}
The first line corresponds to tensor communication and the second line corresponds to matrix communication. As mentioned earlier, the TTM-in-Sequence approach forms a tensor after each TTM computation. Each positive term of the first line corresponds to the number of entries of such a tensor partially computed by a processor in TuckerMPI.

\begin{figure*}[htb]
	\begin{center}
		$\quad$\includegraphics[scale=0.175]{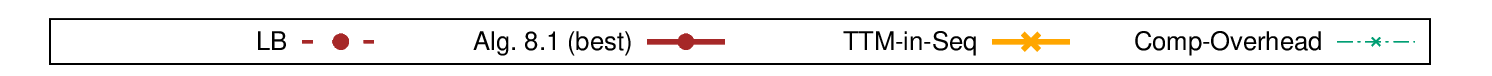}
	\end{center}
	\vspace*{-0.35cm}\begin{center}
		\subfloat[$n_1=n_2= n_3=2^{20},r_1=r_2=r_3=2^{6}$]{	\includegraphics[width=0.3\linewidth]{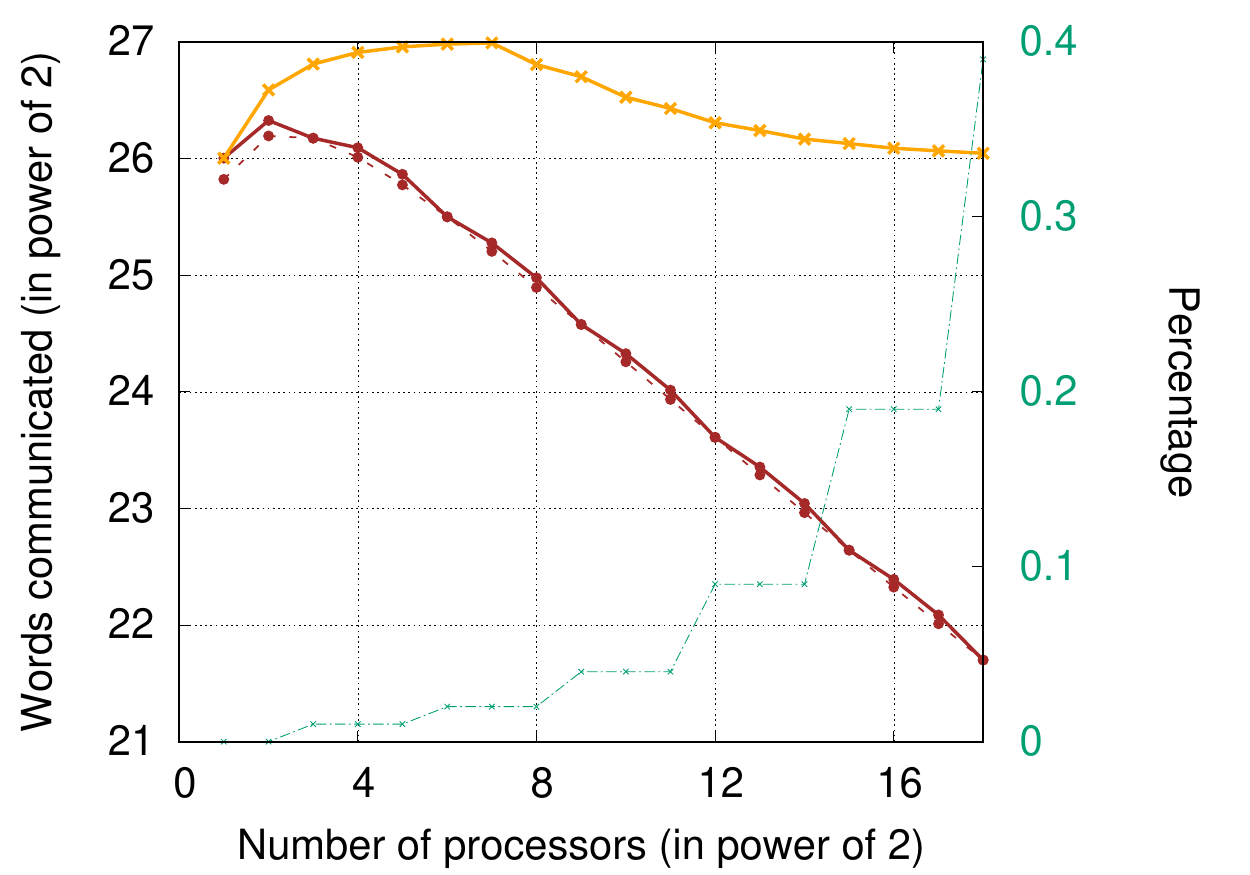}}$\quad$
		\subfloat[$n_1=n_2= n_3=n_4=2^{20},r_1=r_2=r_3=r_4=2^{6}$]{	\includegraphics[width=0.3\linewidth]{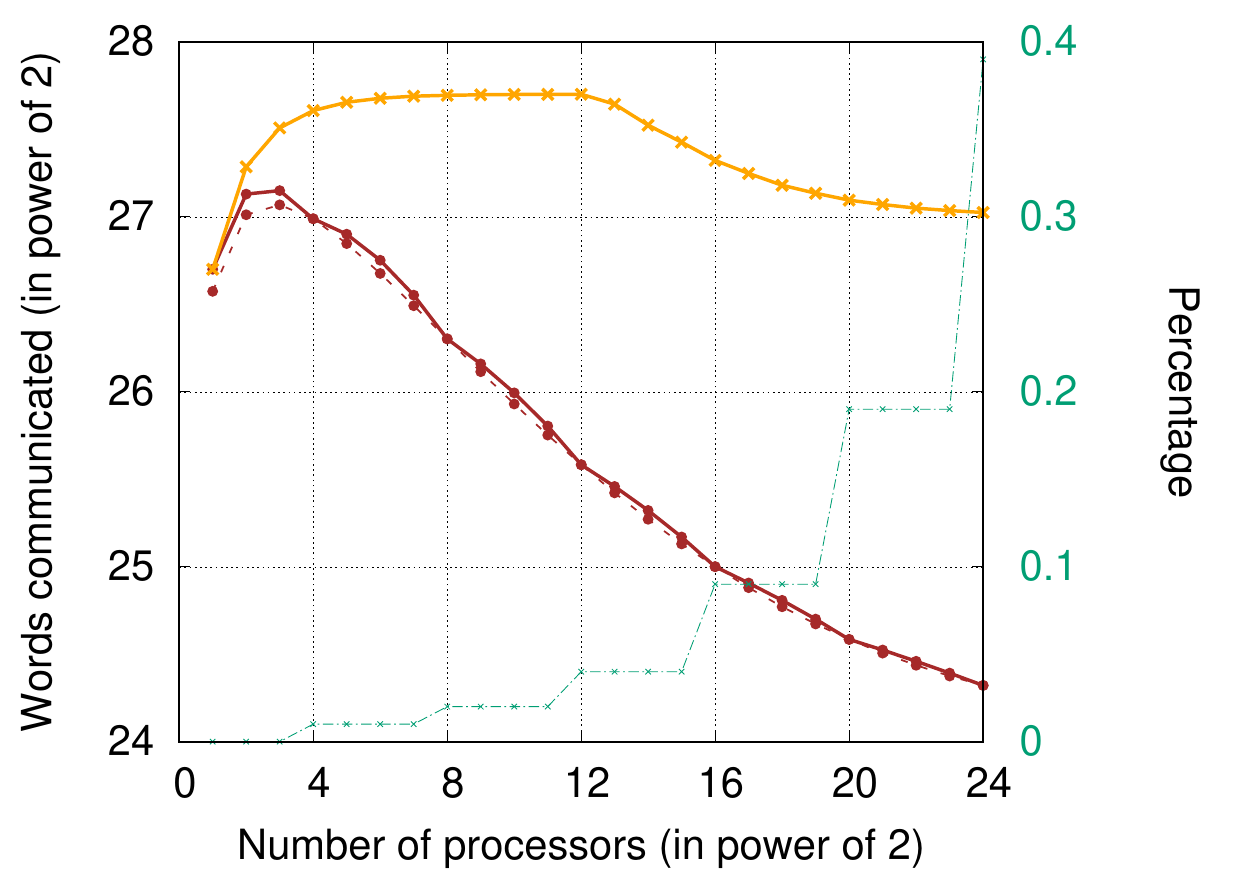}}$\quad$
		\subfloat[$n_1=\cdots=n_5=2^{20},r_1=\cdots=r_5=2^{6}$\label{subfig:TTMinSequenceOutPerorms}]{	\includegraphics[width=0.3\linewidth]{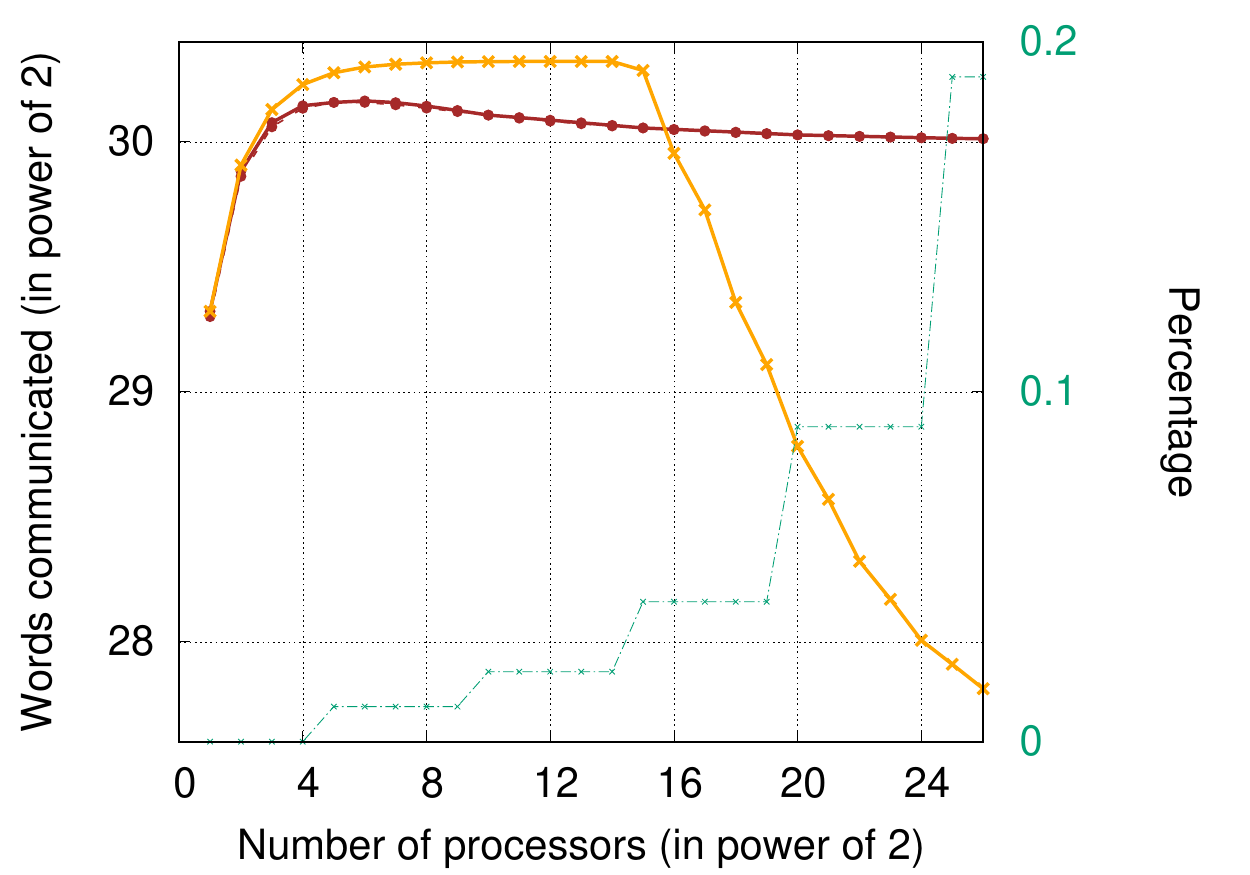}}
		\caption{Communication cost comparison of \cref{alg:genMultittm} and the TTM-in-Sequence approach implemented by the TuckerMPI library. Note that $\lowerbound$ is a communication lower bound for atomic Multi-TTM algorithms, not for the TTM-in-Sequence approach. Communication cost of our approach (\bestconfigAAOGen) is very close to the lower bound ($\lowerbound$).\label{fig:genMultiTTM:experiments:niriconstants}}
	\end{center}
\end{figure*}

\begin{figure*}[htb]
	\begin{center}
		\subfloat[$n_1=n_2= n_3=2^{20},r_1=r_2=r_3=2^{6}$]{	\includegraphics[width=0.3\linewidth]{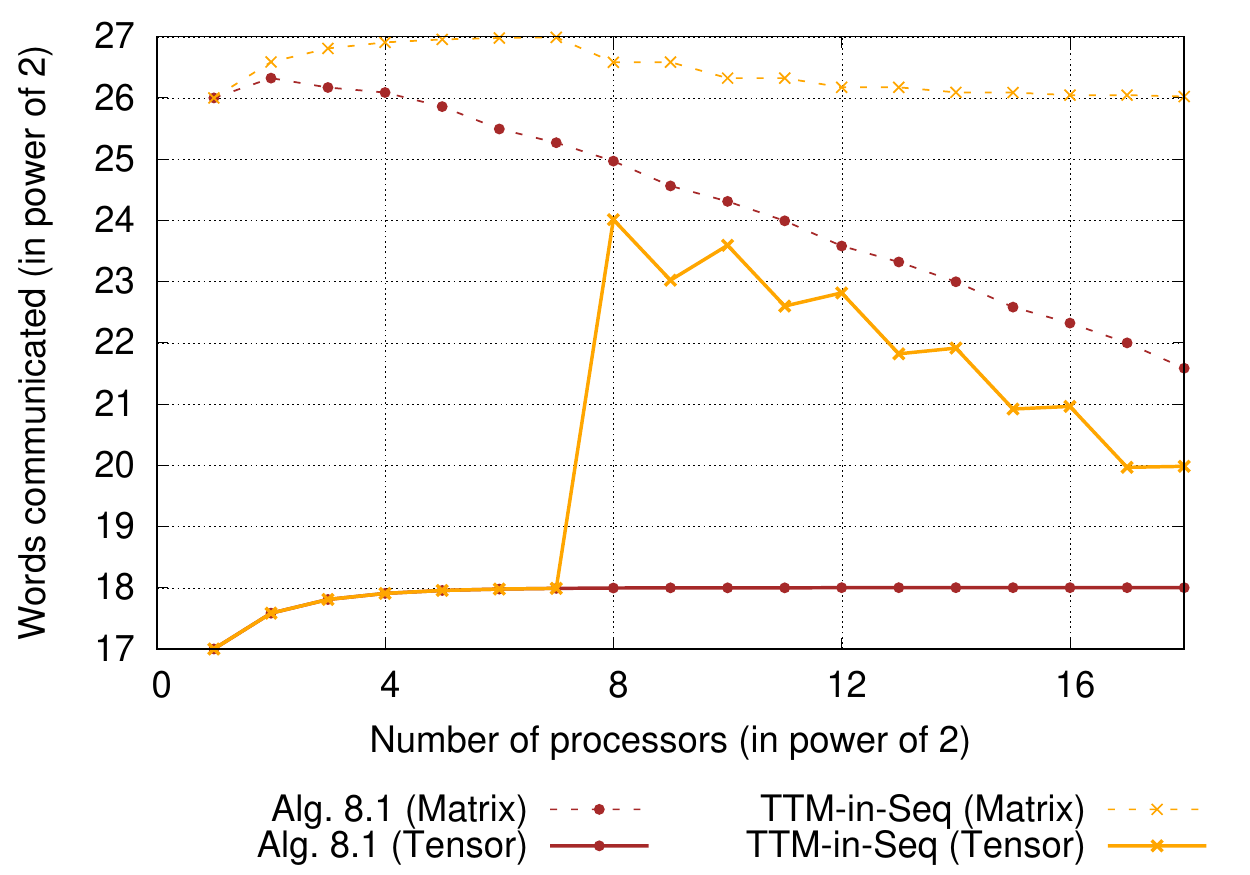}}$\quad$
		\subfloat[$n_1=n_2= n_3=n_4=2^{20},r_1=r_2=r_3=r_4=2^{6}$]{	\includegraphics[width=0.3\linewidth]{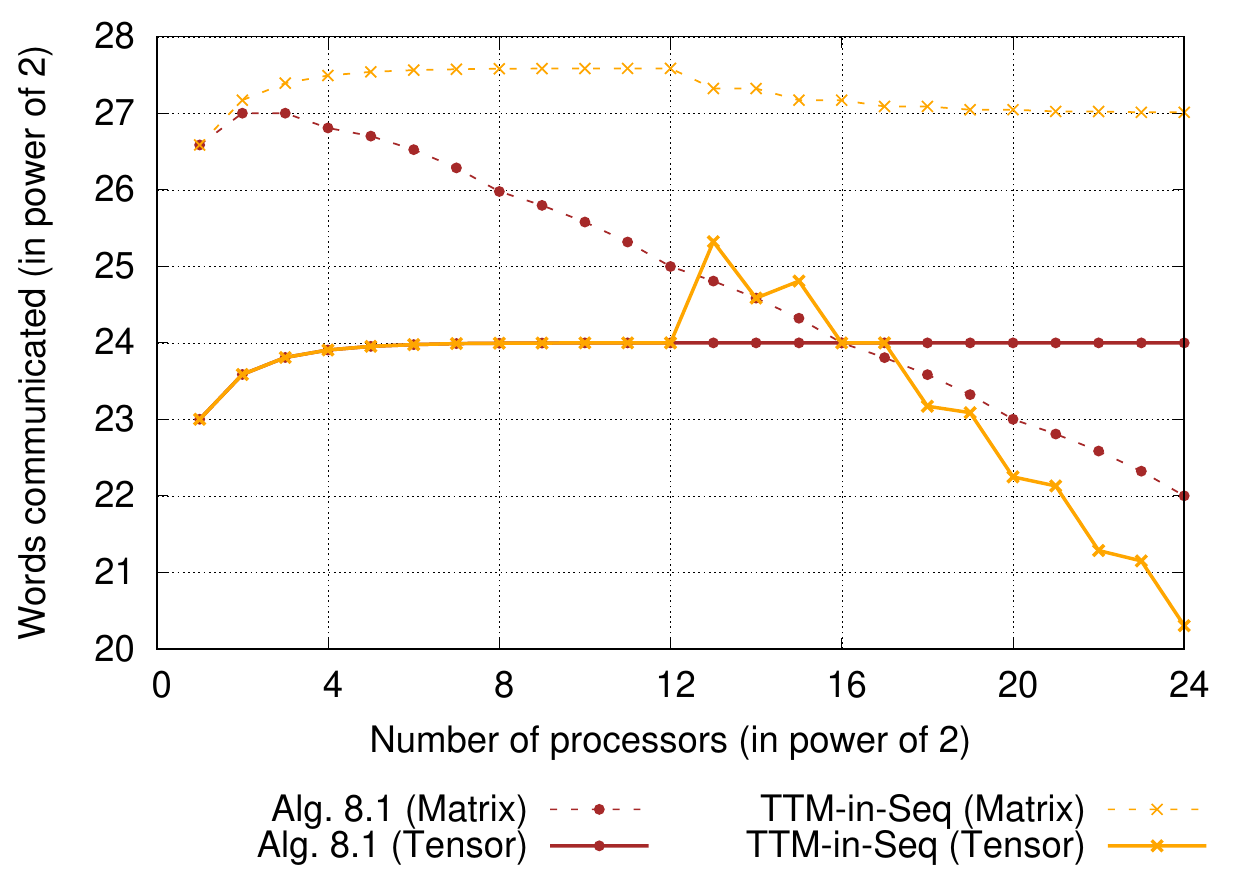}}$\quad$
		\subfloat[$n_1=\cdots=n_5=2^{20},r_1=\cdots=r_5=2^{6}$]{	\includegraphics[width=0.3\linewidth]{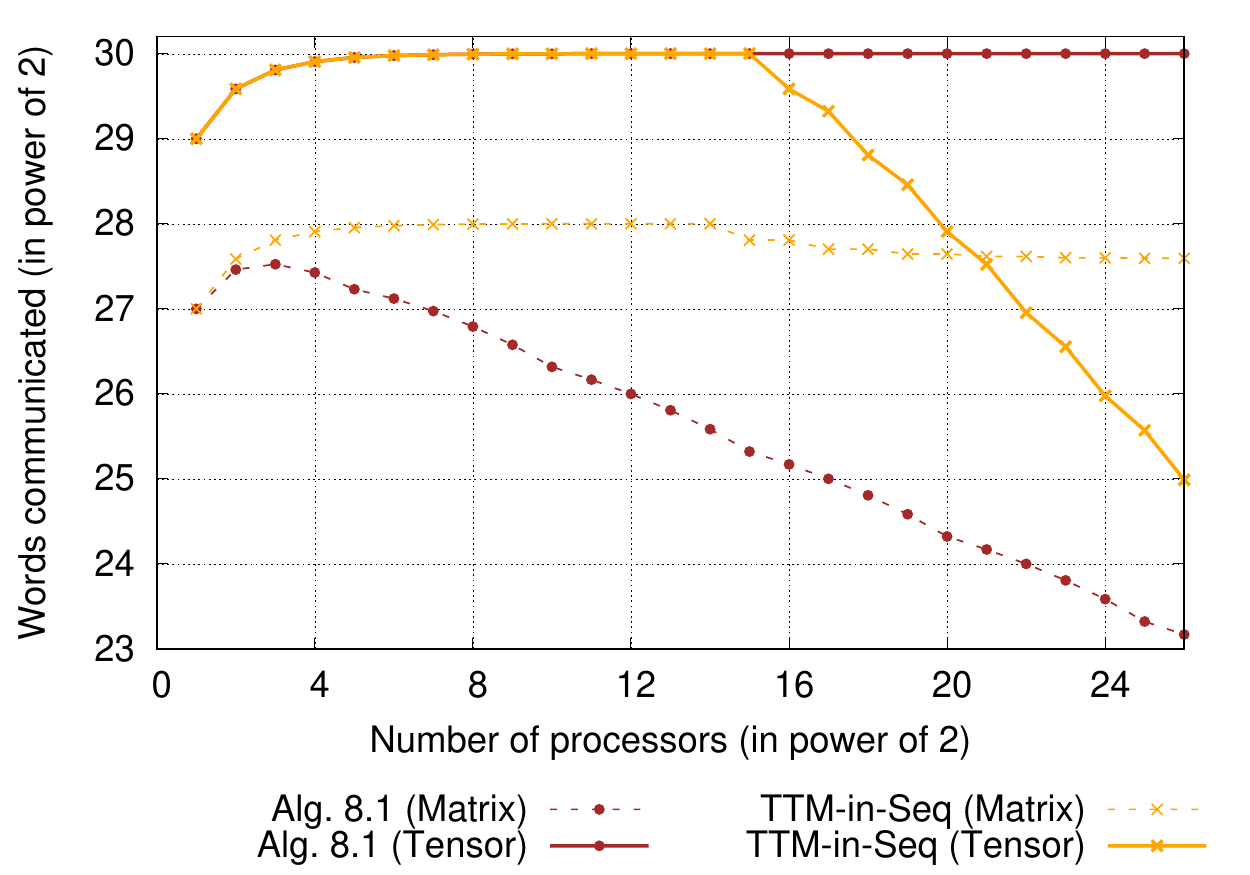}}
		\caption{Matrix and Tensor communication costs in \cref{alg:genMultittm} and the TTM-in-Sequence approach. \label{fig:genMultiTTM:experiments:niriconstants:distribution}}
	\end{center}
\end{figure*}

\begin{figure*}[htb]
	\begin{center}
		\subfloat[$n_1=n_2= n_3=2^{10}, r_1=r_2=r_3=2^{4}$]{	\includegraphics[width=0.3\linewidth]{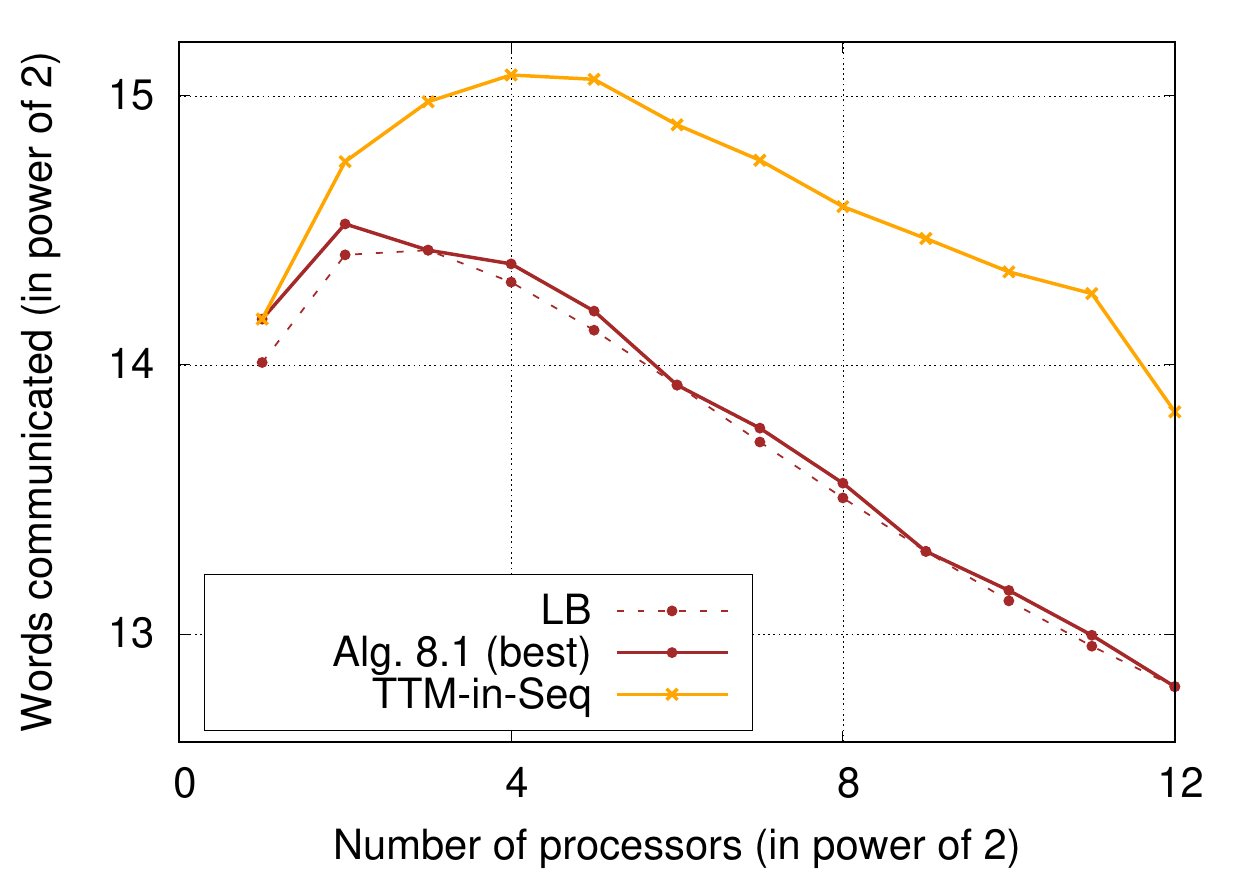}}$\quad$
		\subfloat[$n_1=n_2= n_3=n_4=2^{10},r_1=r_2=r_3=r_4=2^{3}$]{	\includegraphics[width=0.3\linewidth]{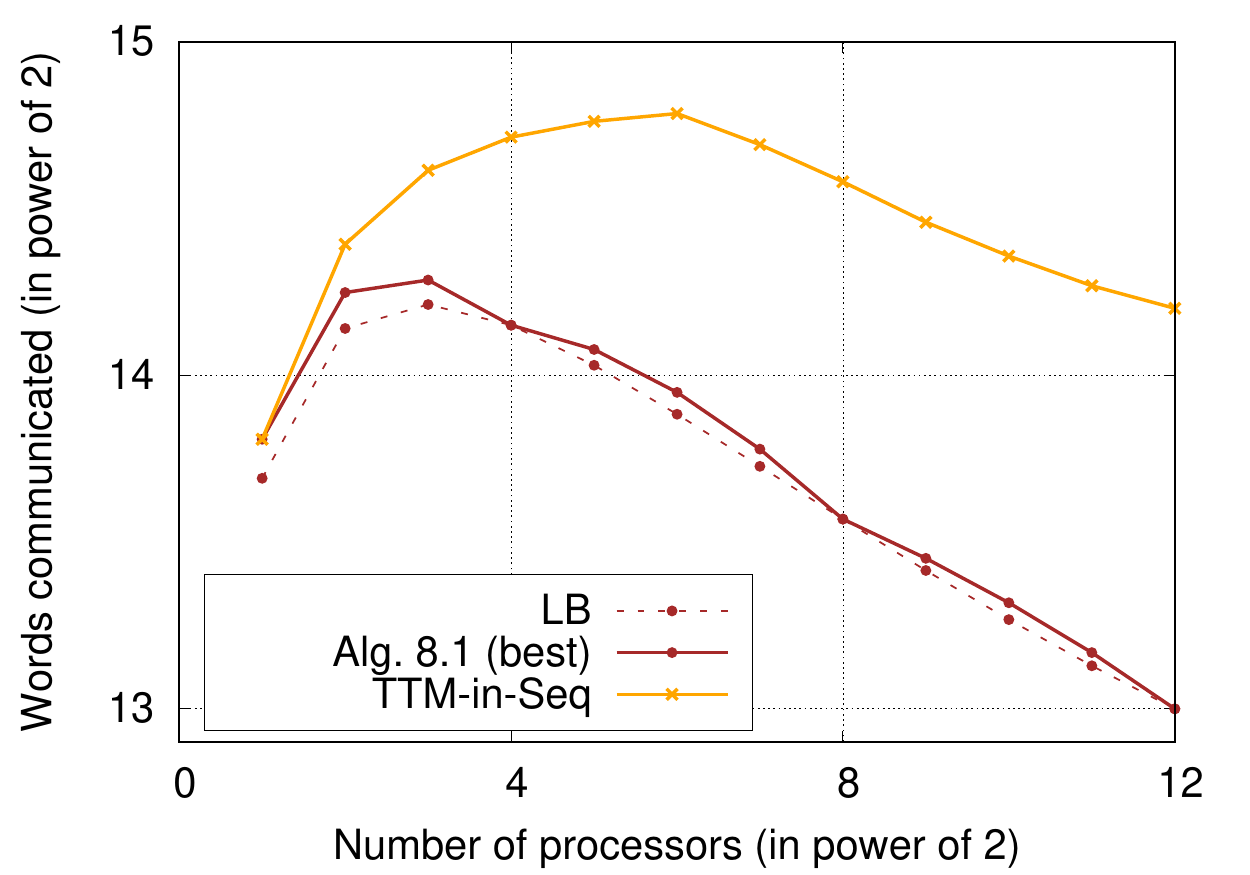}}$\quad$
		\subfloat[$n_1=\cdots=n_6=2^{10},r_1=\cdots=r_6=2^{2}$]{	\includegraphics[width=0.3\linewidth]{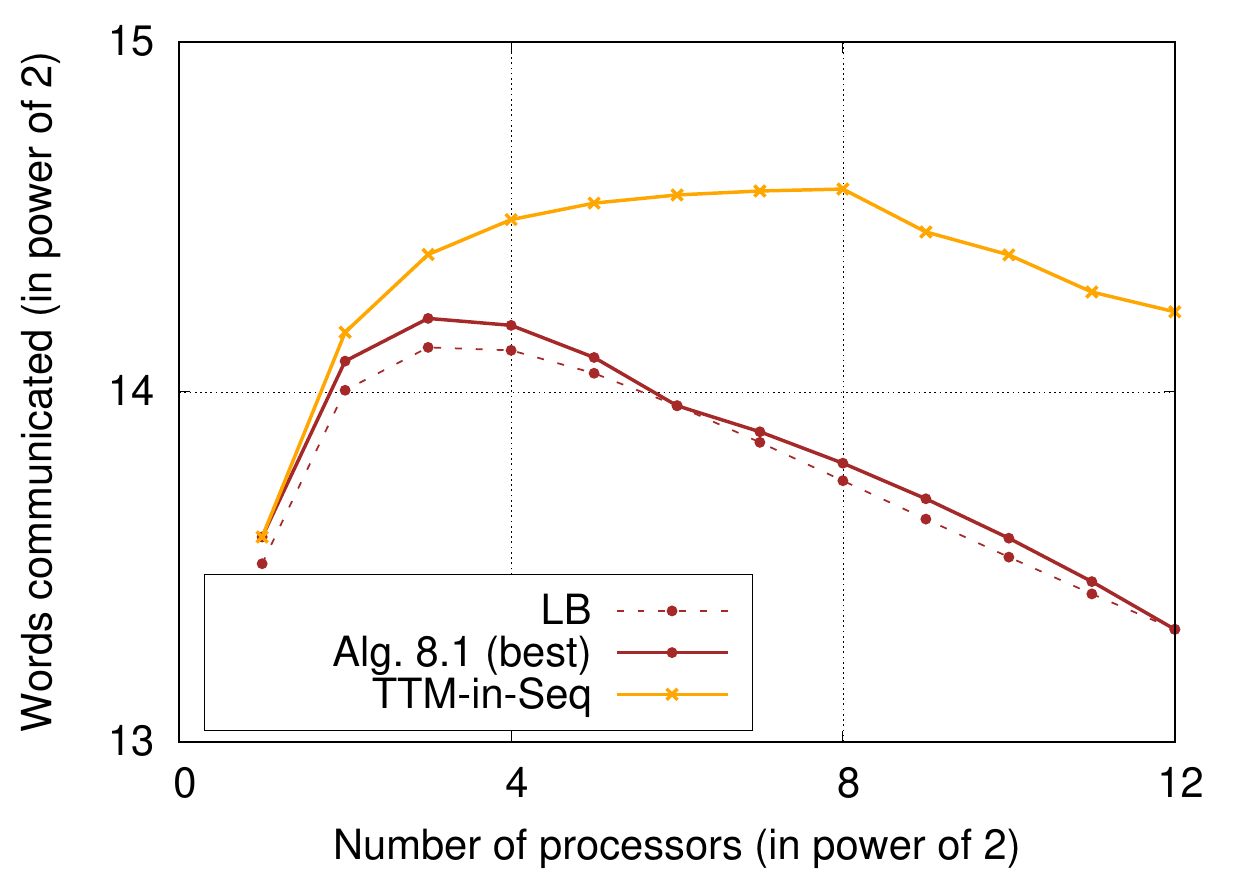}}
		\caption{Communication cost comparison of \cref{alg:genMultittm} and the TTM-in-Sequence approach for 3/4/6-dimensional Multi-TTM computations.\label{fig:genMultiTTM:experiments:2-12}}
	\end{center}
\end{figure*}

We again look at cases where the input tensors are large and the output tensors are small. \Cref{fig:genMultiTTM:experiments:niriconstants} shows comparison of \bestconfigAAOGen and \bestconfigSeq with our communication lower bounds ($\lowerbound$) for 3/4/5-dimensional Multi-TTM computations. For $P=2$, both approaches perform the same amount of communication. After that, the total number of accessed elements in both approaches decreases, however the rate of owned elements decreases at the faster rate. Hence we see slight increase in both curves. This behavior continues roughly till $2^{n_i-r_i}$ processors for \bestconfigSeq curve.
In this region, the TTM-in-Sequence approach selects $\tilde{p_1}=\cdots =\tilde{p_{d-1}} = 1$ and $\tilde{p_d}=P$, and our algorithm selects $p_1\approx\cdots\approx p_d$ and $q_1=\cdots=q_d=1$. These processor grid dimensions result in the same tensor communication cost for both approaches. However our approach reduces matrix communication cost roughly \linebreak$(1-\frac{1}{d})P^\frac{1}{d}$ times, hence it is better than the TTM-in-Sequence approach. \cref{fig:genMultiTTM:experiments:niriconstants:distribution} shows the distribution of matrix and tensor communication costs in both approaches. In general, our approach significantly minimizes the matrix communication costs in all the plots and is better when the number of the entries in the output tensor is less than that of the matrices. When the communication cost is dominated by the output tensor,
our approach is outperformed by the TTM-in-Sequence approach, which is the case in \cref{subfig:TTMinSequenceOutPerorms}.

\response{The parallel computational cost of TuckerMPI with cubical tensors is 
$$2\left(\frac{r^{\frac{1}{d}}n}{P}+\frac{r^{\frac{2}{d}}n^{\frac{d-1}{d}}}{P}+ \cdots + \frac{rn^{\frac{1}{d}}}{P}\right).$$
When $n\gg r$, \cref{alg:genMultittm} selects a processor grid such that $q=1$ and $p_1\approx p_2\approx \cdots \approx p_d$. In this case the computation cost given in \cref{app:sec:genUpperBounds:costAnalysis} simplifies to
$$2\left(\frac{r^{\frac{1}{d}}n}{P}+\frac{r^{\frac{2}{d}}n^{\frac{d-1}{d}}}{P^\frac{d-1}{d}}+ \cdots + \frac{rn^{\frac{1}{d}}}{P^\frac{1}{d}}\right).$$}

\response{While the first terms of the two computational cost expressions match, we observe greater computational cost from \cref{alg:genMultittm} in the remaining terms. These terms are lower order when $P\ll n/r$, in which case the extra computational cost of \cref{alg:genMultittm} is negligible. We plot computational overheads of our algorithm, with $Comp-Overhead$ label, in \cref{fig:genMultiTTM:experiments:niriconstants}. We can note that the overheads are negligible (less than 0.5\%) for the considered experiments.}

Now we consider a different set of experiments. Here the number of entries in the input tensor is $(2^{10})^d$ for $d$-dimensional computations. We fix the number of entries in the output tensor to $2^{12}$ and present comparisons of the considered approaches in \Cref{fig:genMultiTTM:experiments:2-12} for 3/4/6-dimensional Multi-TTM computations. These dimensions allow both tensors to be cubical. As the number of entries in the matrices are greater than the number of entries in the output tensor, our approach is always superior to the TTM-in-Sequence approach.


Our results are consistent with what we observe for $3$-dimensional Multi-TTM computations in \cref{sec:experiments}. When the input tensor is much larger than the output tensor and the number of entries in the output tensor is less than that of the matrices, our algorithm significantly reduces communication compared to the TTM-in-Sequence approach. As in the 3D case, when $P\ll n/r$, the extra computation is negligible when the TTM-in-Sequence approach is used locally to reduce computation.


\section{Conclusions}
\label{sec:conclusions}
In this work, we establish communication lower bounds for the parallel Multi-TTM computation and present an optimal parallel algorithm that organizes the processors in a $2d$-dimensional grid for $d$-dimensional tensors. 
By judiciously selecting the processor grid dimensions, we prove that our algorithm attains the lower bounds to within a constant factor. 
To verify the theoretical analysis, we simulate Multi-TTM computations using a variety of values for the number of processors, $P$, the dimension, $d$, and sizes, $n_i$ and $r_i$; compute the communication costs of our algorithm corresponding to each simulation; and compute the optimal communication cost provided by the theoretical lower bound. These simulations show that the communication costs of the proposed algorithm are close to optimal.
When one of the tensors is much larger than the other tensor, which is typical in compression algorithms based on the Tucker decomposition, our algorithm significantly reduces communication costs over the conventional approach of performing the computation as a sequence of tensor-times-matrix operations. 


Motivated by the simulated communication cost comparisons, our next goal is to implement the parallel atomic algorithm and verify the performance improvement in practice.
Further, because neither the atomic or TTM-in-sequence approach is always superior in terms of communication, we wish to explore hybrid algorithms to account for significant dimension reduction in some modes but modest reduction in others.
Given the computation and communication capabilities of a parallel platform, it would also be interesting to study the computation-communication tradeoff for these two approaches and how to minimize the overall execution time in practice.
Finally, this work considers that each processor has enough memory. 
A natural extension is to study communication lower bounds for Multi-TTM computations with limited memory sizes.

\appendix

\section{Proof of \cref{lem:mttkrpOpt}}
\label{app:sec:optLemmaProof}

In this section, we prove \cref{lem:mttkrpOpt} as it is written, instead of relying on the reader to derive this result from \cite[Lemma 5.1]{Ballard:PPSC:MTTKRP}.
This proof relies on two additional results.
The first, \cite[Lemma 2.2]{Ballard:PPSC:MTTKRP}, states that the first constraint of the optimization problem is quasiconvex~\cite{ABGKR22-report}.
The second, \cite[Lemma 3]{ABGKR22-report}, states that satisfying the Karush-Kuhn-Tucker (KKT) conditions is sufficient for a solution to the optimization problem to be optimal as the optimization problem minimizes a differentiable convex function and the contraints are all differentiable quasiconvex functions.



%

\begin{proof}[Proof of \Cref{lem:mttkrpOpt}]
	To begin we note that the objective and all but the first constraint are affine functions, which are differentiable, convex, and quasiconvex.
	The first constraint is differentiable, and it is quasiconvex in the positive orthant by~\cite[Lemma 2.2]{Ballard:PPSC:MTTKRP}.
	Thus the KKT conditions are sufficient to demonstrate the optimality of any solution by ~\cite[Lemma 3]{ABGKR22-report}, and we will prove the optimality of the solution $\mathbf{\starontop{x}}=\begin{bmatrix}\starontop{x_1} &\starontop{x_2} & \cdots & \starontop{x_d}\end{bmatrix}$ by finding dual variables $\starontop{\mu_i}$ for $0\leq i\leq d$ such that the KKT conditions are satisfied.
	
	We now convert the problem to standard notation.
	The minimization objective function is
	$$f(\mathbf{x}) = \sum_{j \in [d]} x_j,$$
	and the constraints are given by
	\begin{align*}
	g_0(\mathbf{x}) &= \frac{nr}{P} - \prod_{j \in [d]} x_j\text,\\
	g_i(\mathbf{x}) &= x_i -k_i \text{ for all }i\in[d]\text.
	\end{align*}
	Partial derivatives for $j\in[d]$ are given by 
	\begin{align*}
	\frac{\partial f}{\partial x_j} (\mathbf{x}) &=1 \text, \\
	\frac{\partial g_0}{\partial x_j} (\mathbf{x}) &= -\prod_{\ell \in [d]-\{j\} }x_\ell\text, \\
	\frac{\partial g_i}{\partial x_j} (\mathbf{x}) &= \begin{cases} 1 & \text{if} \quad i=j \\ 0 & \text{else}. \end{cases} \\
	\end{align*}
	
	The KKT conditions of $(\starontop{\mathbf{x}}, \starontop{\mathbf{\mu}})$ are:
	\begin{itemize}
		\item \emph{Primal feasibility}: $g_i(\mathbf{\starontop{x}}) \le 0$, for $0\leq i\leq d$.
		\item \emph{Stationarity}: $\frac{\partial f}{\partial x_j}(\mathbf{\starontop{x}}) + \sum_{i=0}^{d}\starontop{\mu_i} \frac{\partial g_i}{\partial x_j}(\mathbf{\starontop{x}}) = 0$, for $j\in[d]$.
		\item \emph{Dual feasibility}: $\starontop{\mu_i} \ge 0$, for $0\leq i\leq d$.
		\item \emph{Complementary slackness}: $\starontop{\mu_i} g_i(\mathbf{\starontop{x}})=0$, for $0\leq i\leq d$.
	\end{itemize}
	
	Recall from the statement of \cref{lem:mttkrpOpt} that $K_I=\prod_{j=d-I+1}^d k_j$ and  $1\le I \le d$ is defined such that $L_I \leq P < L_{I+1}$.
	$$\text{Here} \quad L_j=\frac{K_j}{(k_{d-j+1})^j} \quad \text{for} \quad 1 \leq j \leq d \quad \text{and} \quad L_{d+1} = \infty.\qquad\qquad\qquad$$
	We claim that the optimal primal solution is
	$$\starontop{x_j} = \begin{cases} k_j &\qquad\text{if } j \le d-I\text, \\ (K_I/P)^{1/I} &\qquad\text{if } d-I < j \le d\text.\end{cases}$$
	and the optimal dual solution is 
	$$\starontop{\mu_i} = \begin{cases}
	\frac{(K_I/P)^{1/I}}{nr} & \qquad\text{if } i=0\text, \\
	\frac{(K_I/P)^{1/I}}{k_i} -1 &\qquad\text{if } 0<i\le d-I\text,\\ 
	0&\qquad\text{if } d-I < i \le d\text.
	\end{cases}$$
	
	We now check that $\mathbf{\starontop{x}}$ satisfies the primal feasibility condition.
	By direct verification (and the fact that $nr=\prod_{j\in[d]} x_j$), we have $g_0(\starontop{\mathbf{x}})=0$.
	Clearly $g_i(\mathbf{\starontop{x}})= 0$ for all $i\in[d-I]$, as $\starontop{x_i} = k_i$ for all $i\in[d-I]$.
	To see that $g_i(\mathbf{\starontop{x}})\le 0$ for $d-I<i\leq d$, it is sufficient to recall that $\frac{K_I}{(k_{d-I+1})^I}\le P $ by the definition of $I$, and that $\starontop{x_i} = (K_I/P)^{1/I}$ and $k_{d-I+1}\le k_i$ for $d-I<i\leq d$. 
	
	Stationarity follows from direct verification of the condition for $j\in [d]$.
	
	To check dual feasibility, we note that all the factors of $\starontop{\mu_0}$ are positive, thus $\starontop{\mu_0} > 0$.
	To show that $\starontop{\mu_i} >0$ for $i\in[d-I]$, it is sufficient to show that $k_{d-I} < (K_I/P)^{1/I}$ as $k_1\le\cdots\le k_{d-I}$.
	This is implied by $P<\frac{K_{I+1}}{(k_{d-I})^{I+1}} = \frac{k_{d-I}K_{I}}{(k_{d-I})^{I+1}}$, where the inequality comes from the definition of $I$, and the equality from the definition of the right products $K_j$.
	
	Finally, complementary slackness is satisfied because $g_i(\mathbf{\starontop{x}}) = 0$ for $0\le i\le d-I$, and $\mu_i = 0$ for $d-I < i \le d$.
	\begin{align*}
	\phantom{This is here to move the proof symbol to its correct location...}
	\end{align*}
	
\end{proof}

\section*{Acknowledgments}
This work is supported by the National Science Foundation under
Grant No. CCF-1942892 and OAC-2106920. This project has received
funding from the European Research Council (ERC) under the
European Union’s Horizon 2020 research and innovation program
(Grant agreement No. 810367).


\bibliographystyle{siamplain}
\bibliography{report}

\end{document}